%% file: main.tex
\documentclass[11pt,letterpaper,onecolumn]{article}

\usepackage[letterpaper,margin=1in]{geometry}

\usepackage[T1]{fontenc}
\usepackage{amsmath,amssymb,amsfonts,amsthm,mathtools}
\usepackage{graphicx,xcolor}
\usepackage[ruled,vlined]{algorithm2e}
\SetKwInOut{KwIn}{Input}
\SetKwInOut{KwOut}{Output}
\usepackage{enumitem}
\usepackage[hidelinks]{hyperref}
\usepackage[nameinlink,capitalise]{cleveref}
\usepackage{bbm}

\usepackage[numbers,sort&compress]{natbib}

\usepackage{authblk}

\newtheorem{theorem}{Theorem}
\newtheorem{lemma}[theorem]{Lemma}
\newtheorem{corollary}[theorem]{Corollary}
\newtheorem{proposition}[theorem]{Proposition}

\newtheorem{remark}[theorem]{Remark}

\title{Effective Horizons for Prophet Inequalities\\with Discounted I.I.D.\ Rewards}

\author[1]{Jung-hun Kim\thanks{\texttt{junghun.kim@ensae.fr}}}
\author[1,2]{Vianney Perchet\thanks{\texttt{vianney.perchet@normalesup.org}}}

\affil[1]{FairPlay Team, CREST, ENSAE, Institut Polytechnique de Paris}
\affil[2]{Criteo AI Lab}

\date{}

\begin{document}
\pagestyle{plain}
\pagenumbering{roman}

\maketitle
\thispagestyle{empty}

\begin{abstract}
\normalsize
We study prophet inequalities for $n$ i.i.d.\ base rewards subject to multiplicative discounts across $S$ phases. For equal-length phases and a common discount factor $\widetilde\gamma$, we characterize the optimal competitive ratio of single-quantile policies as $n\to\infty$ with $S$ and $\widetilde\gamma$ fixed. The resulting sharp guarantee describes the dependence on the phase parameters: it recovers $1-1/e$ without discounting and approaches $1/2$ in the worst-case regime. We also establish a distribution-free comparison that bounds the
discounted prophet by a fractional i.i.d.\ maximum with effective
horizon $B$, the sum of the discount weights across all observations. The comparison yields acceptance probability $1/B$, determined only by the phase profile. This calibration attains the canonical optimum even among single-quantile policies that may choose their acceptance probability using the known base distribution. The comparison also gives guarantees for heterogeneous discounts and unequal phase lengths. 
A complementary optimal-stopping upper bound establishes the same limiting optimum of $1/2$ for arbitrary stopping rules in the canonical model. Finally, the same effective-horizon calibration sets acceptance probability $1-\widetilde\gamma$ under infinite-horizon geometric discounting. This explicit quantile rule attains Boshuizen's (1994) known optimal worst-case ratio $1/(1+\widetilde\gamma)$  without solving a Bellman equation.
\end{abstract}

\newpage
\pagenumbering{arabic}

\input{STOC2027/Introduction}
\input{STOC2027/related}
\input{STOC2027/problem}

\input{STOC2027/common_decay}
\input{STOC2027/upper_bound_common_decay}
\input{STOC2027/DP}
\input{STOC2027/upper_bound}
\input{STOC2027/infinite}

\input{STOC2027/experiments}

\input{STOC2027/conclusion}

\appendix
\newpage
\section{Appendix}

\input{STOC2027/comparision_value_quantile}

\input{STOC2027/proof_upper}

\input{STOC2027/proof_dis}

\input{STOC2027/proof_lower}
\input{STOC2027/proof_DP}

\input{STOC2027/proof_general}

\input{STOC2027/proof_lower_gen}

\input{STOC2027/proof_upper_final}

\input{STOC2027/classical_continuous_decay}

\input{STOC2027/proof_S=n}

\input{STOC2027/proof_continuous_decay_lower}

\section*{Acknowledgements}
\begin{sloppypar}
This work was supported in part by the French National Research Agency (ANR) through grants ANR-20-CE23-0007 and ANR-23-CE23-0002, and through the PEPR IA FOUNDRY project (ANR23-PEIA-0003). It was also supported by Hi! PARIS and the ANR/France 2030 program (ANR23-IACL-0005). JK was supported by Basic Science Research Program through the National Research Foundation of Korea (NRF) funded by the Ministry of Education (RS-2026-25536608).
\end{sloppypar}

\bibliographystyle{plainnat}

\bibliography{STOC2027/mybib}

\paragraph{AI methodology.}
The main ideas, including the effective-horizon approach and the
key ideas behind the comparison lemma, were developed by the authors.
We used OpenAI's ChatGPT to refine proof ideas for the extension
to heterogeneous discounts and unequal phase lengths
(Appendix~\ref{app:general}), draft proof sketches, and improve the
exposition of the manuscript. The authors take full responsibility
for the correctness and originality of the final manuscript.

\end{document}

%% file: STOC2027/Introduction.tex
\section{Introduction}

Prophet inequalities measure the cost of making an irrevocable selection before seeing all available opportunities. A gambler observes independent nonnegative rewards sequentially and may accept one reward; a prophet sees the entire sequence and selects its maximum. The objective is to maximize the ratio of their expected rewards, uniformly over the relevant distributions. This framework is central to online optimal stopping and underlies applications such as posted-price mechanisms and online allocation \citep{hill1992survey,lucier2017economic,alaei2012online}.

The classical guarantees depend on both the reward model and the policy class. For independent, possibly non-identically distributed rewards, a static value threshold achieves the optimal worst-case ratio $1/2$ \citep{samuel1984comparison}. For stationary i.i.d.\ rewards, the sharp asymptotic guarantee for single-threshold policies is $1-1/e$; more general stopping policies can do better \citep{correa2017posted,correa2019recent,goldenshluger2024optimal}. Discounting provides a structured departure from stationarity: the base distribution stays the same, but the value of an opportunity decreases with its arrival time. For example, in hiring for a project with a fixed completion date, later hires have less time to contribute. The theoretical question is how this known change in scale affects the guarantees of simple stopping policies.

For infinite-horizon geometric discounting, where the discount
weight decreases by a factor of $\widetilde\gamma$ at each step,
the optimal worst-case ratio $1/(1+\widetilde\gamma)$ is classical
\citep{boshuizen1994optimal}.
We study models with finitely many phases, where the discount
weight remains constant within each phase.
Our focus is how the number and lengths of phases, together with
their discount weights, determine single-quantile guarantees,
and how to choose the acceptance probability accordingly.

We consider a horizon of $n$ observations divided into $S$ consecutive phases. In phase $s$, the observed reward is $X_i=\Gamma_{s-1}Y_i$, where the $Y_i$ are i.i.d.\ from a known distribution, $\Gamma_0=1$, and the cumulative discounts $\Gamma_{s-1}$ are nonincreasing. The phase lengths and discounts are known. We compare two natural classes of threshold policies, each specified by a single scalar parameter: an acceptance probability or a value threshold. A \emph{single-quantile threshold policy} fixes an acceptance probability and uses the corresponding quantile threshold in each phase. Its value threshold scales with $\Gamma_{s-1}$, so fixing the acceptance probability determines all phase thresholds without introducing additional parameters. This differs from a \emph{static value-threshold policy}, which uses one threshold on the observed rewards. The two parameterizations coincide in the stationary i.i.d.\ setting but separate under discounting. Static value thresholds retain their universal $1/2$ guarantee; our separation result shows that, with equal-length phases and certain intermediate discount factors, our single-quantile policy gives a higher worst-case guarantee that static value thresholds cannot match.

The organizing idea is to calibrate acceptance to the discounted mass of opportunities,
\[
    B:=\sum_{s=1}^S n_s\Gamma_{s-1},
\]
rather than to the number $n$ of observations. Our comparison bounds the discounted prophet by an i.i.d.\
benchmark with effective horizon $B$. When $B$ is an integer,
this benchmark is the expected maximum of $B$ draws from the
base distribution. The comparison motivates acceptance probability
$1/B$ and helps establish the policy's performance guarantee.



Our main contributions are a sharp finite-phase characterization and the effective-horizon comparison that yields its optimal calibration. We complement these results with bounds for arbitrary stopping rules and an infinite-horizon application.
\begin{itemize}[leftmargin=*]
\item \textbf{Sharp finite-phase guarantees for single-quantile policies.}
For equal-length $S$-finite phases and a common discount factor $\widetilde\gamma$, we characterize how the optimal asymptotic competitive ratio of single-quantile policies depends on $S$ and $\widetilde\gamma$. We prove matching upper and lower bounds as horizon $n\to\infty$ with these parameters fixed (Theorems~\ref{cor:upper} and~\ref{thm:lower}). The explicit guarantee interpolates between the classical i.i.d.\ single-threshold value $1-1/e$ without discounting and $1/2$ in the worst-case regime. The acceptance probability $1/B$, determined only by the phase profile, attains this optimum even among single-quantile policies that may choose their acceptance probability using the known base distribution.

\item \textbf{An effective-horizon comparison for quantile calibration.}
We prove a distribution-free comparison that bounds the discounted prophet by a fractional i.i.d.\ maximum with effective horizon $B$ (Lemma~\ref{lem:multistage} and Appendix~\ref{app:dis}). This comparison justifies the $1/B$ acceptance-probability calibration and is the key step in proving the matching lower bound in the canonical model. It also yields explicit guarantees for heterogeneous discounts and unequal phase lengths. For fixed $S$ and positive limiting phase proportions $n_s/n\to\lambda_s>0$, these guarantees are asymptotically tight among single-quantile policies  whose
acceptance probability depends only on the phase profile   (Theorems~\ref{thm:lower2} and~\ref{thm:upper}).


\item \textbf{Explicit quantile calibration for the classical infinite-horizon optimum.}
For infinite-horizon geometric discounting with $0<\widetilde\gamma<1$, the same calibration gives acceptance probability $1/B_\infty=1-\widetilde\gamma$. This explicit quantile rule attains the known optimal worst-case ratio $1/(1+\widetilde\gamma)$ established by \citet*{boshuizen1994optimal}, without solving a Bellman equation (Theorems~\ref{prop:S=n} and~\ref{prop:continuous_lower}). It connects the finite-phase calibration to the classical infinite-horizon minimax benchmark.
\end{itemize}


%% file: STOC2027/related.tex
\section{Related Work}

We review classical and discounted prophet inequalities and discuss connections to random horizons and secretary problems.

\paragraph{Classical prophet inequalities and threshold policies.}
The study of prophet inequalities originates in the work of \citet*{krengel1977semiamarts}; see \citet*{hill1992survey} for a survey of optimal-stopping comparisons. For independent nonnegative rewards, \citet*{samuel1984comparison} proved the sharp $1/2$ guarantee using a static value threshold. In the i.i.d.\ setting, \citet*{hill1982comparisons} studied comparisons between stopping values and expected maxima. Single-parameter threshold policies have a sharp asymptotic guarantee of $1-1/e$, while policies with more flexible thresholds achieve stronger guarantees \citep{correa2017posted,correa2019recent,goldenshluger2024optimal}. 

\paragraph{Discounted prophet inequalities.}
Earlier work on discounted prophet inequalities includes \citet*{boshuizen1991prophet} and \citet*{allaart2006prophet}. The closest work is \citet*{boshuizen1994optimal}, which treats both finite and infinite i.i.d.\ sequences under geometric discounting. For bounded rewards, it derives finite-horizon additive prophet inequalities and characterizes an optimal stationary threshold rule for the infinite horizon (Proposition~2.2 and Sections~3--4). Its Corollary~4.5(ii) establishes the sharp infinite-horizon worst-case ratio $1/(1+\widetilde{\gamma})$ for nonnegative i.i.d.\ rewards. Our focus is how the optimal asymptotic single-quantile guarantee depends on the number of phases and the discount factor in the canonical finite-phase model. Our effective-horizon comparison yields a simple, explicit quantile calibration for finite-phase models, including heterogeneous discounts and unequal phase lengths. The same calibration also attains the known optimal worst-case ratio under infinite-horizon geometric discounting, without solving a Bellman equation.

\paragraph{Random horizons.}
\citet*{giambartolomei2024iid} study i.i.d.\ prophet inequalities with a random horizon whose distribution is known but whose realization is not. When the horizon is independent of the rewards, its survival probabilities act as discount weights in the online stopping objective. However, the prophet benchmarks generally differ, so competitive-ratio guarantees do not transfer directly between the two models.

\paragraph{Discounted secretary problems.}
\citet*{babaioff2009secretary} consider time-dependent weights and discounts in secretary problems with random-order arrivals of fixed values. Our model instead considers independent draws from a known base distribution.

%% file: STOC2027/problem.tex
\section{Problem Statement}
We consider the classical prophet-inequality
framework over a finite horizon $n\in\mathbb{N}$, where a nonnegative random
reward $X_i\ge 0$ is revealed sequentially to a
{gambler} at each stage $i\in[n]$.
To model \emph{decaying reward distributions}, we partition the time horizon into
$S \ge 1$ consecutive phases, or time intervals,
$\{\mathcal T_s\}_{s=1}^S$, which form a disjoint partition of $[n]$.
We denote by $n_s := |\mathcal T_s|$ the length of phase $s$, so that
$\sum_{s=1}^S n_s = n$.

Let \(Y_1,\ldots,Y_n\) be i.i.d.\ samples from a base reward distribution
\(D\) supported on \(\mathbb R_+\). Let $\gamma_s\in[0,1]$ for all $s\in[S-1]$ denote phase-wise decay factors, and define
the cumulative decay coefficients
\[
\Gamma_s:=\prod_{t=1}^s \gamma_t,
\quad s=1,\dots,S-1, \text{ and } \; \Gamma_0:=1.
\]
Then, for each phase $s\in[S]$ and index $i\in\mathcal T_s$, the observed (decayed)
reward is defined as
\[
X_i :=
\begin{cases}
Y_i, & s=1,\\[4pt]
\Gamma_{s-1}\,Y_i, & s\ge 2,
\end{cases}
\]
where $\Gamma_{s-1}$ captures the cumulative decay in reward magnitude across
phases.


We first analyze the canonical setting with a common decay factor and
equal-length phases, where \(\gamma_s=\widetilde\gamma\in[0,1]\) for all \(s\in[S-1]\) and \(n_s=n/S\) for all \(s\in[S]\) (or asymptotically \(n_s/n\to 1/S\)). In this case, the observed rewards take the simple form \(X_i=\widetilde\gamma^{s-1}Y_i\) for \(i\in\mathcal T_s\). This canonical setting isolates the effect of accumulated discounting in its simplest form. We later return to the general phase-dependent formulation to handle heterogeneous decay factors and unequal phase lengths.


Following most of the prior literature in prophet inequalities
\citep{hill1982comparisons,samuel1984comparison,correa2017posted,ehsani2018prophet,goldenshluger2024optimal},
we assume that the reward distributions are known, which corresponds
to knowing the base distribution $\mathcal D$, as well as the phase partition
$\{\mathcal T_s\}_{s\in[S]}$ and decay factors $(\gamma_s)_{s\in[S-1]}$ in advance.

\paragraph{Stopping rule.}
At each stage $i$, after observing $X_i$, the decision maker must either \emph{stop and accept the current reward} or \emph{continue to the next stage}, thereby irrevocably forfeiting the option to stop at any earlier stage.
A {stopping rule} (or policy) $\tau$ is a random variable taking values in
$\{1,2,\ldots,n\}\cup\{n+1\}$ such that $\tau$ is a stopping time with respect to
the natural filtration generated by $(X_i)_{i\ge1}$. The value $\tau = n+1$ corresponds to the event that no reward is
accepted within the horizon, in which case the payoff is defined to be zero.
The expected reward of a policy $\tau$ is given by
\[
\mathbb{E}_D[X_\tau] = \sum_{i=1}^{n} \mathbb{E}_D[X_i \,\mathbbm{1}\{\tau = i\}],
\]
where \(\mathbb{E}_{D}\) denotes expectation under the decayed reward sequence
induced by \(Y_i\sim D\). When the
base distribution is clear from context, we omit the subscript \(D\).

\paragraph{Prophet benchmark and competitive ratio.}
We compare the performance of an online stopping rule with that of a
\emph{prophet} who knows all realizations \({X_1,\ldots,X_n}\) in advance and
can always select the maximum reward. For a base  reward distribution \(D\), the prophet's
expected reward is
\[
\mathbb{E}_{D}\left[\max_{i\in[n]}X_i\right].
\]

Let \(\mathfrak D\) denote the class of admissible base reward distributions on
\(\mathbb R_+\) such that $0<\mathbb{E}_D[Y]< \infty$, where \(Y\sim D\). For a fixed decay profile and phase partition,
the worst-case competitive ratio of a stopping policy \(\tau\) is defined as
\[
\mathrm{CR}_n(\tau)
:=
\inf_{D\in\mathfrak D}
\frac{\mathbb{E}_{D}[X_{\tau}]}
{\mathbb{E}_{D}\left[\max_{i\in[n]}X_i\right]}.
\]  
The quantity $\mathrm{CR}_n(\tau)$ implicitly depends on the decay profile and
phase partition; we suppress this dependence to keep the notation light. We seek stopping policies that maximize this
worst-case competitive ratio.


\paragraph{Quantile versus value thresholds.}
We focus on simple threshold rules specified by a single scalar acceptance parameter. Among such rules, two particularly natural parameterizations are to fix a value threshold in the observed reward scale, or to fix an acceptance quantile of the current reward distribution.
These two parameterizations coincide in the stationary i.i.d.\ setting, up to tie-breaking, but they differ once rewards are rescaled over time by discounting.

Throughout the paper, we use the quantile parameterization. A fixed acceptance quantile induces phase-dependent value thresholds that automatically track the reward scale. Specifically, if \(\alpha_s\) denotes the value threshold in phase \(s\), then
$
\alpha_{s+1}=\gamma_s\alpha_s.
$
Thus, once the initial threshold \(\alpha_1\) (equivalently, the acceptance quantile) is chosen, all subsequent value thresholds are determined by the decay factors; no additional phase-specific parameters need to be optimized. This preserves the simplicity and computational tractability of single-parameter threshold rules while adapting to the changing reward scale.

By contrast, a fixed-value threshold is not invariant to changes in the reward scale and can be strictly suboptimal. Appendix~\ref{app:comparison} provides a separating example showing that fixed-value thresholds cannot attain the fixed-quantile benchmark characterized below. Hence, between these two natural single-parameter threshold rules, the quantile parameterization is the natural choice under discounting. We refer to fixed-quantile policies as
\emph{single-quantile threshold policies} throughout.




%% file: STOC2027/common_decay.tex
\section{Fundamental Limits in the Canonical Decay Model}

We begin with the canonical case of a common decay factor
\(\gamma_s=\widetilde{\gamma}\) for all \(s\in[S-1]\) and equal asymptotic
phase lengths, \(n_s=n/S\) for all \(s\in[S]\) where we assumed, for the sake of simplicity, that
 \(n\) is a multiple of \(S\) so that $n_s$ are integers. This canonical setting
 captures the central phenomenon: accumulated discounting can drive the
competitive ratio from the classical i.i.d.\ value \(1-1/e\) down to the
non-i.i.d.\ barrier \(1/2\). The following theorem formalizes this phenomenon for arbitrary single-quantile threshold policies.

\begin{theorem}\label{cor:upper}
Assume a common decay factor and equal-length phases. Then,  for any single-quantile threshold stopping policy $\tau$, 
the competitive ratio satisfies
\[
\limsup_{n\to\infty}\mathrm{CR}_n(\tau) \;\le\;
\begin{cases}
\displaystyle
\frac{\big(1-\widetilde{\gamma}^{S}\exp\!\big(-S\tfrac{1-\widetilde{\gamma}}{1-\widetilde{\gamma}^{S}}\big)\big)
      \big(1-\exp\!\big(-\tfrac{1-\widetilde{\gamma}}{1-\widetilde{\gamma}^{S}}\big)\big)}
     {1-\widetilde{\gamma}\,\exp\!\big(-\tfrac{1-\widetilde{\gamma}}{1-\widetilde{\gamma}^{S}}\big)},
& \text{if } \; 0 \le \widetilde{\gamma} < 1,
\\[12pt]
1 - \tfrac{1}{e},
& \text{if }\; \widetilde{\gamma} = 1.
\end{cases}
\]
\end{theorem}
\begin{proof}[Proof Sketch] Here we provide a proof sketch and the full version is provided in
Appendix~\ref{app:upper_adaptive}.
 The key is to construct a single hard
distribution, depending only on $(S,\widetilde{\gamma})$, that upper bounds every
single-quantile threshold policy.

Let the phases have equal length $n/S$, and consider
\[
Y_i=
\begin{cases}
n, & \text{with probability }1/n^2,\\
y_{\widetilde{\gamma}}, & \text{with probability }1-1/n^2,
\end{cases}
\qquad
X_i=\widetilde{\gamma}^{s-1}Y_i,\quad i\in\mathcal T_s,
\]
where the constant $y_{\widetilde{\gamma}}\ge 0$ will be chosen below. Define
\[
c_{\widetilde{\gamma}}:=\frac1S\sum_{s=1}^S\widetilde{\gamma}^{s-1}
=
\begin{cases}
\dfrac{1-\widetilde{\gamma}^S}{S(1-\widetilde{\gamma})},&0\le\widetilde{\gamma}<1,\\[6pt]
1,&\widetilde{\gamma}=1.
\end{cases}
\]
Then we show that 
\begin{align}
    \mathbb E\!\left[\max_{i\in[n]}X_i\right]
=y_{\widetilde{\gamma}}+c_{\widetilde{\gamma}}+o(1).\label{eq:opt_proof1}
\end{align}

A single-quantile threshold is characterized by a per-stage acceptance
probability $p$. If $R_s(p)$ is the
expected reward conditional on acceptance in phase $s$, then we show that, uniformly in $p$,
\[
pR_s(p)
\le \frac{\widetilde{\gamma}^{s-1}}{n}
   +\widetilde{\gamma}^{s-1}y_{\widetilde{\gamma}} p.
\]
Consequently,
\[
\mathbb E[X_\tau]=\sum_{s=1}^S\sum_{i\in\mathcal{T}_s} p\,R_{s}(p)\,(1-p)^{i-1}
\le
\sum_{s=1}^S
\left(\frac{\widetilde{\gamma}^{s-1}}{n}
+\widetilde{\gamma}^{s-1}y_{\widetilde{\gamma}} p\right)
(1-p)^{(s-1)n/S}
\frac{1-(1-p)^{n/S}}{p}.
\]
Let \(p=p_n\) denote the per-stage acceptance probability and set \(x_n:=np_n\). In the interior regime \(x_n\to x\in(0,\infty)\), the preceding bound yields
\begin{align}
\limsup_{n\to\infty}\mathbb E[X_{\tau}]
\le
\left(\frac1x+y_{\widetilde{\gamma}}\right)A_S(x),
\label{eq:X_tau_proof1}
\end{align}
where
$
A_S(x):=
\sum_{s=1}^S
\widetilde{\gamma}^{s-1}e^{-(s-1)x/S}(1-e^{-x/S})
=
\frac{(1-\widetilde{\gamma}^S e^{-x})(1-e^{-x/S})}
     {1-\widetilde{\gamma} e^{-x/S}}.
$
Thus,  \begin{align}
\limsup_{n\to\infty}\mathrm{CR}_n(\tau)
\le f_S(x,y_{\widetilde{\gamma}}),
\qquad
f_S(x,y):=
\frac{(\frac1x+y)A_S(x)}{y+c_{\widetilde{\gamma}}}.
\label{eq:CR_bd_proof1}
\end{align}
The boundary regimes \(x_n\to0\) and \(x_n\to\infty\) are handled in the full proof.


Define
\[
x_{\widetilde{\gamma}}:=\frac1{c_{\widetilde{\gamma}}}\quad \text{ and }\quad
y(x):=\frac{A_S(x)}{x^2A_S'(x)}-\frac1x.
\]
We show that, by setting $y_{\widetilde{\gamma}}=y(x_{\widetilde{\gamma}})$,
$f_S(\cdot,y_{\widetilde{\gamma}})$ is increasing before $x_{\widetilde{\gamma}}$ and decreasing
after $x_{\widetilde{\gamma}}$. Hence
\begin{align}
\sup_{x>0}f_S(x,y_{\widetilde{\gamma}})
=f_S(x_{\widetilde{\gamma}},y_{\widetilde{\gamma}})
=A_S(x_{\widetilde{\gamma}}), \label{eq:f_A_proof1}
\end{align}
where the last equality follows from $1/x_{\widetilde{\gamma}}=c_{\widetilde{\gamma}}$. Therefore this yields the desired upper bound uniformly over all single-quantile threshold choices \(x\).

For $0\le\widetilde{\gamma}<1$,
$
x_{\widetilde{\gamma}}=\frac{S(1-\widetilde{\gamma})}{1-\widetilde{\gamma}^S},$
and hence
\[
A_S(x_{\widetilde{\gamma}})
=
\frac{
\left(1-\widetilde{\gamma}^S
\exp\!\left(-S\frac{1-\widetilde{\gamma}}{1-{\widetilde{\gamma}}^S}\right)\right)
\left(1-
\exp\!\left(-\frac{1-\widetilde{\gamma}}{1-{\widetilde{\gamma}}^S}\right)\right)}
{1-\widetilde{\gamma}
\exp\!\left(-\frac{1-\widetilde{\gamma}}{1-{\widetilde{\gamma}}^S}\right)}.
\]
When $\widetilde{\gamma}=1$, we have $x_{\widetilde{\gamma}}=1$ and
$A_S(1)=1-1/e$. Therefore, with \eqref{eq:CR_bd_proof1} and \eqref{eq:f_A_proof1}, for any single-quantile threshold policy $\tau$,
\[
\limsup_{n\to\infty}\mathrm{CR}_n(\tau)
\le A_S(x_{\widetilde{\gamma}}),
\]
which is the claimed bound.
\end{proof}
\begin{figure}[t]
    \centering
    \includegraphics[width=0.5\linewidth]{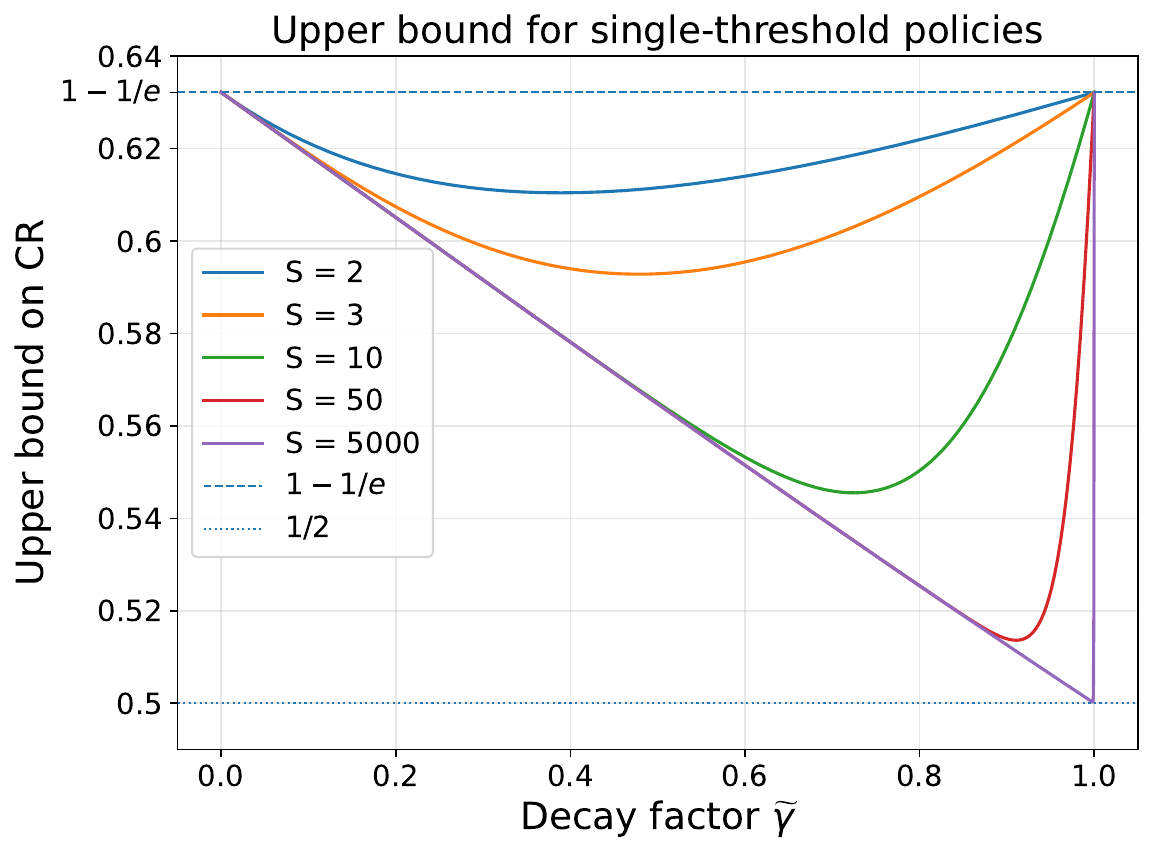}
    \caption{Upper bound on the competitive ratio as a function of the decay factor.}
    \label{fig:upperbound_single_threshold}
\end{figure}
Figure~\ref{fig:upperbound_single_threshold} plots the upper bound in
Theorem~\ref{cor:upper} as a function of the decay factor
\(\widetilde{\gamma}\) for several values of \(S\). For any fixed \(S\), the bound
recovers the classical value \(1-1/e\) at the stationary points
\(\widetilde{\gamma}=1\) and\footnote{Only the first phase contributes
non-negligibly and the problem reduces to an i.i.d. instance with effective
horizon~\(n_1\)} $\widetilde{\gamma}=0$. However, as \(S\) grows, the worst-case guarantee
deteriorates, 
and the minimum approaches the \(1/2\) barrier in the near-stationary regime \(\widetilde{\gamma} \uparrow 1\). The following corollary makes this convergence to the \(1/2\) barrier explicit.



\begin{corollary}[Worst-case asymptotics]\label{cor:worst}
Under the assumptions of Theorem~\ref{cor:upper}, let
\(U_S(\widetilde{\gamma})\) denote its upper-bound value. Then, for every fixed
$\widetilde{\gamma}\in[0,1)$,
\[
    \lim_{S\to\infty} U_S(\widetilde{\gamma})
    =\lim_{S\to \infty} \frac{\big(1-\widetilde{\gamma}^{S}\exp\!\big(-S\tfrac{1-\widetilde{\gamma}}{1-\widetilde{\gamma}^{S}}\big)\big)
      \big(1-\exp\!\big(-\tfrac{1-\widetilde{\gamma}}{1-\widetilde{\gamma}^{S}}\big)\big)}
     {1-\widetilde{\gamma}\,\exp\!\big(-\tfrac{1-\widetilde{\gamma}}{1-\widetilde{\gamma}^{S}}\big)}=
    \frac{1-\exp(-(1-\widetilde{\gamma}))}
         {1-\widetilde{\gamma}\exp(-(1-\widetilde{\gamma}))}
    =:U_\infty(\widetilde{\gamma}).
\]
Moreover, since $U_\infty(\widetilde{\gamma})$ is non-increasing on
$[0,1)$, the worst-case ratio satisfies
\[
    \lim_{\widetilde{\gamma}\uparrow 1}
    U_\infty(\widetilde{\gamma})
    =
    \frac12 .
\]
Equivalently, the same conclusion holds along any joint limit $S\to\infty$   and $\widetilde{\gamma}\uparrow1$
such that $\widetilde{\gamma}^{S}\to 0$ (i.e., $S(1-\widetilde{\gamma})\to\infty$).
\end{corollary}

\paragraph{Discontinuous breakdown at the stationarity boundary}
 Under $S\to \infty$, the regime $\widetilde{\gamma} \uparrow 1$ corresponds to the boundary between stationary and
decay reward sequences. 
Over short time scales, the reward process is nearly
indistinguishable from the i.i.d.\ case, while over long horizons the cumulative
effect of decay becomes significant. At this boundary, the algorithm faces the sharpest trade-off: waiting is tempting because over short horizons the process looks nearly i.i.d., but waiting is expensive because over long horizons the cumulative decay becomes substantial.
This trade-off is reflected in the worst-case bound, whose limiting
value approaches \(1/2\) as \(\widetilde{\gamma}\uparrow 1\) under \(S\to\infty\).
Thus the boundary regime exhibits a \emph{discontinuous breakdown} from the
stationary value \(1-1/e\) attained at \(\widetilde{\gamma}=1\) (see the case
\(S=5000\) in Figure~\ref{fig:upperbound_single_threshold}).

\medskip

This observation raises two natural questions. First, is this upper bound tight within the class of single-quantile threshold policies? Second, is the resulting \(1/2\) limit merely an artifact of the single-quantile threshold restriction, or is it an intrinsic limitation of the discounted-reward model? The next section answers the first question by giving a matching effective-horizon threshold rule; we return to the second question later by analyzing the optimal dynamic-programming policy.



%% file: STOC2027/upper_bound_common_decay.tex
\section{Single-Quantile Threshold Algorithm}
In this section, we present a single-quantile (or fixed-quantile) threshold algorithm that achieves tight competitive ratios for a common decay\footnote{The extension to heterogeneous decay factors and
unequal phase lengths is deferred to Appendix~\ref{app:general}.} factor with $\widetilde{\gamma}$ and equal asymptotic
phase lengths with $n_s=n/S$. Both the design of the
algorithm and its analysis are calibrated through the notion of an
\emph{effective horizon}:
\[
B := \sum_{s=1}^{S} \widetilde{\gamma}^{s-1}\, \frac{n}{S},
\]
which aggregates the total discounted mass of opportunities across all phases.

The central insight underlying our analysis is an effective-horizon principle: in terms of expected maxima, the prophet benchmark under decay is bounded above by that of a stationary i.i.d.\ sequence observed over an \emph{effective horizon} of length~$B$.  We begin by formalizing this intuition through a key comparison lemma,
 which  plays a crucial role in both the
calibration of the acceptance threshold and the ensuing competitive-ratio
analysis.

\subsection{Bounding the Expected Maximum}
We bound the prophet's expected maximum in terms of an effective horizon.
The full proof is deferred to Appendix~\ref{app:dis}.

\begin{lemma}\label{lem:multistage} Let $F$ denote the cumulative distribution function of the base reward distribution ${D}$. Consider a common decay factor and equal-length phases. 
 Then, we have
\begin{equation}\label{eq:multistage}
\int_0^\infty \Bigl(1-F(z)^{B}\Bigr)\,dz
\;\ge\;
\int_0^\infty \Bigl(1-\prod_{s=1}^{S} F\!\bigl(z/\widetilde\gamma^{\,s-1}\bigr)^{n/S}\Bigr)\,dz,
\end{equation}
where we adopt the convention $F(z/0):=\lim_{u\to\infty}F(u)=1$ for $z\ge0$.
\end{lemma}

\begin{proof}[Proof sketch]
Let $G(w)$ denote the left-hand side of~\eqref{eq:multistage} with $B$ replaced by $(n/S)\sum_s w_s$, and let $H(w):=\mathbb E[\max_s w_sM_s]$, where $M_s:=\max_{i\in\mathcal T_s}Y_i$. The function $G$ is concave because $b\mapsto 1-F(z)^b$ is concave, while $H$ is convex as the expectation of a maximum of linear functions. The two agree on each vector $e^{(k)}$ with $k$ leading ones and zeros thereafter. For the cumulative discount vector $v=(1,\widetilde\gamma,\ldots,\widetilde\gamma^{S-1})$, write $v=\sum_{k=1}^S\pi_ke^{(k)}$, where $\pi_k=v_k-v_{k+1}$ and $v_{S+1}:=0$. These coefficients are nonnegative and sum to one, so Jensen's inequality gives
\[
G(v)\ge\sum_{k=1}^S\pi_kG(e^{(k)})
=\sum_{k=1}^S\pi_kH(e^{(k)})\ge H(v).
\]
This is precisely~\eqref{eq:multistage}.
\end{proof}

We now explain the implication of the preceding lemma.
For each $i\in\mathcal{T}_s$, the cumulative distribution function of $X_i$ is
$z\mapsto F\!\bigl(z/\widetilde{\gamma}^{s-1}\bigr)$.
Combining the lemma with the standard survival-function formula for maxima yields the following, assuming $B\in\mathbb{N}$ for simplicity of exposition:
\[
\mathbb{E}\!\left[\max_{i \in [B]} Y_i\right]
=\int_0^\infty \bigl(1-F(z)^{B}\bigr)\,dz
\;\ge\;
\int_0^\infty \Bigl(1-\prod_{s=1}^{S} F\!\bigl(z/\widetilde\gamma^{s-1}\bigr)^{n/S}\Bigr)\,dz
=\mathbb{E}\!\left[\max_{i \in [n]} X_i\right].
\] This inequality identifies $B$ as an \emph{effective i.i.d. horizon} for the decayed prophet
benchmark. We use this effective-horizon comparison to calibrate the acceptance quantile and
to analyze the competitive ratio.


\subsection{Algorithm}






    
    
We consider a single-quantile stopping rule whose acceptance threshold adapts to
the discounted scaling of the reward distribution.
Let $Y \sim {D}$ denote a random variable drawn from the base reward distribution. The initial threshold \(\alpha_1\) is chosen to satisfy
\begin{align}\mathbb{P}( Y > \alpha_1)
= \frac{a}{n},\label{eq:threshold}
\end{align}
where the parameter $a$ determines the acceptance quantile. In the common-decay, equal-phase setting, we set \(a\) as the inverse of the \emph{average effective decay weight}:
\[
a =\frac{1}{(1/S)+\sum_{s=1}^{S-1}(1/S) \widetilde{\gamma}^s}.
\]
At the transition from phase $s$ to phase $s+1$ for $s\in[S-1]$, the threshold is updated by the decay factor, \[
\alpha_{s+1}=\widetilde{\gamma}\alpha_s.\]
Thus, for \(\widetilde{\gamma}>0\), the thresholds satisfy the same single-quantile condition\footnote{When $\widetilde{\gamma}=0$, all rewards after the first phase are zero, so the problem reduces to the first phase only.} such that  $\mathbb{P}(X_i>\alpha_s)=\mathbb{P}(\widetilde{\gamma}^{s-1}Y_i>\widetilde{\gamma}^{s-1}\alpha_1)=\mathbb{P}(Y_i>\alpha_1)=a/n$ for all $i\in \mathcal{T}_s$ and $s\in[S]$.
The pseudo-code of the threshold policy is provided in
Algorithm~\ref{alg:phase-threshold}.
\begin{algorithm}[t]
\caption{Single-Quantile Threshold Strategy for Decaying Rewards}\label{alg:phase-threshold}
\KwIn{Base threshold $\alpha_1$} 

\For{$s=1$ \KwTo $S$}{
    \For{$i \in \mathcal{T}_s$}{
        observe $X_i$\;
        \If{$X_i > \alpha_s$}{
            accept $X_i$ and stop\;
        }
    }
    \If{$s < S$}{
        $\alpha_{s+1} \leftarrow \gamma_s \alpha_s$\;
    }
}
\Return $0$\;
\end{algorithm}
\begin{remark}[Interpretation of the threshold choice]\label{rm:interpret_threshold}
As shown in Lemma~\ref{lem:multistage}, the decaying reward sequence can be
interpreted, in terms of expected maxima, as a stationary i.i.d.\ sequence with effective horizon
length $B=(1+\sum_{s=1}^{S-1}\widetilde{\gamma}^{s})\frac{n}{S}$.
Accordingly, a natural acceptance quantile is $1/B$, mirroring the classical
calibration in stationary i.i.d.\ prophet problems with horizon $n$ where the acceptance
probability scales as $1/n$ \citep{correa2019recent}.
This leads precisely to the choice $a/n (=1/B)$ in \eqref{eq:threshold}.
\end{remark}
\begin{remark}[Atomic base distributions]\label{rm:atom_alg}
If the base distribution has atoms, the strict threshold equation in~\eqref{eq:threshold}  may not admit a solution. In this case, we use a generalized quantile \(\alpha_s\) satisfying $
\mathbb{P}(X_i>\alpha_s)\le \frac{a}{n}\le \mathbb{P}(X_i\ge \alpha_s),$
together with randomized tie-breaking at \(\alpha_s\) so that the acceptance probability is exactly \(a/n\). All arguments remain unchanged after replacing the exceedance event \({X_i>\alpha_s}\) by this randomized acceptance event.
\end{remark}

Under this construction, we obtain the following lower bound on the competitive ratio, which is tight and matches the upper bound in Theorem~\ref{cor:upper}.

\begin{theorem}\label{thm:lower}
  Consider any decaying reward distributions with $S \ge 1$
nonstationary phases of equal length $n/S$ and a common decay factor
$\gamma_t \equiv \widetilde{\gamma}$ for all $t \in [S-1]$.
Then,  the stopping policy $\tau$ of Algorithm~\ref{alg:phase-threshold} with $\alpha_1$ from \eqref{eq:threshold} achieves a competitive ratio bound of
    \[\liminf_{n\to\infty}\mathrm{CR}_n(\tau)\ge \begin{cases}
\displaystyle
\frac{\big(1-\widetilde{\gamma}^{S}\exp\!\big(-S\tfrac{1-\widetilde{\gamma}}{1-\widetilde{\gamma}^{S}}\big)\big)
      \big(1-\exp\!\big(-\tfrac{1-\widetilde{\gamma}}{1-\widetilde{\gamma}^{S}}\big)\big)}
     {1-\widetilde{\gamma}\,\exp\!\big(-\tfrac{1-\widetilde{\gamma}}{1-\widetilde{\gamma}^{S}}\big)},
& \text{if }\; 0 \le \widetilde{\gamma} < 1,\\[12pt]
1 - \tfrac{1}{e}, & \text{if }\; \widetilde{\gamma} = 1.
\end{cases}\]
\end{theorem}
\begin{proof}[Proof Sketch] Here we provide a proof sketch, and the full version is provided in Appendix~\ref{app:lower}.  The proof has two main components: (i) compute the expected reward of the threshold policy, and (ii) relate the truncated tail expectation
$B\,\mathbb{E}[Y\mathbbm{1}\{Y>\alpha_1\}]$ to the prophet benchmark
$\mathbb{E}[\max_{i\in[n]}X_i]$.

Define $p:=\mathbb P(Y>\alpha_1)=a/n$.
By independence,
\begin{align}
    \mathbb E[X_\tau]
=\sum_{s=1}^S\sum_{i\in\mathcal T_s}\mathbb E\!\left[X_i\mathbf 1\{X_i>\alpha_s\}\right](1-p)^{i-1}.\label{eq:X_tau_proof2}
\end{align}
Since $X_i>\alpha_s \iff Y_i>\alpha_1$ for any $i\in\mathcal{T}_s$ and $s\in[S]$, the acceptance events are identical across phases.
For $i\in\mathcal T_s$ we have
\[
\mathbb E\!\left[X_i\mathbf 1\{X_i>\alpha_s\}\right]
=\widetilde{\gamma}^{\,s-1}\mathbb E\!\left[Y\mathbf 1\{Y>\alpha_1\}\right],
\]
where $Y\sim D$ (base reward).
Substituting into the expansion of \eqref{eq:X_tau_proof2} yields
\[
\mathbb E[X_\tau]
=\mathbb E\!\left[Y\mathbf 1\{Y>\alpha_1\}\right]
\sum_{s=1}^S \widetilde{\gamma}^{\,s-1}\sum_{i\in\mathcal T_s}(1-p)^{i-1}.
\]
Using the above results with $p=a/n$ and the geometric-series identity give the closed form
\begin{align}
\mathbb{E}[X_\tau]
=
\left(1-(1-a/n)^{n/S}\right)
\left(1+\sum_{s=2}^S (1-a/n)^{(s-1)n/S}\widetilde{\gamma}^{s-1}\right)
\,\frac{n}{a}\mathbb{E}[Y\mathbbm{1}\{Y>\alpha_1\}].\label{eq:exp_reward}    
\end{align}

From the definition,  $B=\frac{n}{S}\bigl(1+\sum_{s=1}^{S-1}\widetilde{\gamma}^s\bigr)$. For exposition, suppose that \(B\in\mathbb N\); the full proof handles
arbitrary real \(B\) using the corresponding fractional-maximum benchmark.
By Lemma~\ref{lem:multistage} and the standard survival-function formula for maxima,
one has the comparison
\[
\mathbb{E}\Big[\max_{i\in[B]}Y_i\Big]\ \ge\ \mathbb{E}\Big[\max_{i\in[n]}X_i\Big].
\]
It remains to relate $\mathbb{E}[\max_{i\in[B]}Y_i]$ to the tail moment
$\mathbb{E}[Y\mathbbm{1}\{Y>\alpha_1\}]$ in \eqref{eq:exp_reward}. Using a standard truncation argument (details omitted) and $\mathbb{P}(Y>\alpha_1)=a/n=1/B$, we obtain
\begin{align}
     B\mathbb{E}[Y\mathbbm{1}(Y>\alpha_1)]=\alpha_1+B\mathbb{E}[(Y-\alpha_1)_+]\ge \mathbb{E}\left[\max_{i\in[B]}Y_i\right] \ge \mathbb{E}\left[\max_{i\in [n]}{X_i}\right]. \raisetag{4ex}\label{eq:lower_bd_bE}
\end{align}

From 
$n/a=B$,
 \eqref{eq:exp_reward} together with \eqref{eq:lower_bd_bE} gives
\begin{align}\label{eq:lower_X_tau}
    \mathbb{E}[X_\tau]
\ \ge\
\left(1-(1-a/n)^{n/S}\right)
\left(1+\sum_{s=1}^{S-1}(1-a/n)^{sn/S}\widetilde{\gamma}^s\right)
\,
\mathbb{E}\Big[\max_{i\in[n]}X_i\Big].
\end{align}
Letting $n\to\infty$ 
leads to the explicit lower bound on the competitive ratio with
\[
(1-a/n)^{s(n/S)}\;\longrightarrow \;e^{-sa/S}.
\]
For $0\le\widetilde{\gamma}<1$, with $a =\frac{S(1-\widetilde{\gamma})}{1-\widetilde{\gamma}^{S}}$, \eqref{eq:lower_X_tau} simplifies to
\[
\liminf_{n\to\infty}\mathrm{CR}_n(\tau)\ \ge\
\frac{\bigl(1-\widetilde{\gamma}^{S}e^{-S(1-\widetilde{\gamma})/(1-\widetilde{\gamma}^S)}\bigr)\bigl(1-e^{-(1-\widetilde{\gamma})/(1-\widetilde{\gamma}^S)}\bigr)}
{1-\widetilde{\gamma} e^{-(1-\widetilde{\gamma})/(1-\widetilde{\gamma}^S)}},
\]
and for $\widetilde{\gamma}=1$, with $B=n$, \eqref{eq:lower_X_tau} reduces to the classical $1-1/e$ bound.
\end{proof}

Combining Theorems~\ref{cor:upper} and~\ref{thm:lower}, we obtain a tight characterization of the optimal single-quantile threshold guarantee in the canonical common-decay, equal-length setting. In particular, the effective-horizon threshold rule exactly matches the worst-case upper bound, and the deterioration toward \(1/2\) cannot be improved within the single-quantile threshold class.

\begin{remark}[Extension to heterogeneous phases]
The effective-horizon principle extends to heterogeneous decay factors and unequal phase lengths. More precisely, for cumulative decay coefficients \(\Gamma_{s-1}\) and phase lengths \(n_s\), the relevant effective horizon is
$B_G=\sum_{s=1}^S n_s\Gamma_{s-1}.$
Calibrating the single-quantile policy to \(B_G\) yields a generalized competitive-ratio guarantee. The corresponding bounds, together with the formal statements and proofs, are provided in Appendix~\ref{app:general}.
\end{remark}

The deterioration toward \(1/2\) within the single-quantile threshold class raises a more structural question: is this limit caused by the simplicity of single-quantile threshold policies, or is it an intrinsic limitation of the discounted-reward model itself? We next show that the latter is the case: even the optimal dynamic-programming policy, and therefore any stopping rule, faces the same worst-case \(1/2\) barrier in the near-stationary long-phase regime.


%% file: STOC2027/DP.tex
\section{The \texorpdfstring{\(1/2\)}{1/2} Barrier Beyond Single-Quantile Threshold Policies}

We now show that the \(1/2\) worst-case limit is not an artifact of restricting attention to single-quantile threshold policies. The following theorem analyzes the optimal dynamic-programming (DP) policy in the common-decay, equal-length setting. Since the DP policy dominates every admissible stopping rule, an upper bound on its competitive ratio immediately yields a fundamental limitation for all stopping rules. The detailed proof is provided in Appendix~\ref{app:DP}.


\begin{theorem}\label{prop:dp-worst}
Consider the setting with a common decay factor $\widetilde{\gamma}\in[0,1)$ and
equal-length phases. In the worst-case boundary regime of  $S\to\infty$   and $\widetilde{\gamma}\uparrow1$
such that $\widetilde{\gamma}^{S}\to 0$ (see Corollary~\ref{cor:worst}), the optimal dynamic programming policy $\tau_{\mathrm{DP}}$ satisfies the
following upper bound on the competitive ratio: 
\[
\limsup_{\substack{\widetilde{\gamma}\uparrow 1,\; S\to\infty\\
\widetilde{\gamma}^{S}\to 0}}\limsup_{n\to \infty}\mathrm{CR}_n(\tau_{\mathrm{DP}})
\le 
\frac{1}{2}.
\]
Consequently, since $\tau_{\mathrm{DP}}$ is optimal among all stopping rules,
for any stopping policy $\tau$,
\[
\limsup_{\substack{\widetilde{\gamma}\uparrow 1,\; S\to\infty\\
\widetilde{\gamma}^{S}\to 0}}\limsup_{n\to\infty}\mathrm{CR}_n(\tau)
\;\le\;
\frac{1}{2}.
\]
\end{theorem}


\paragraph{Structural implication.}
Theorem~\ref{prop:dp-worst} identifies a fundamental worst-case limitation of $1/2$
in the presence of decaying rewards.
Interestingly, this $1/2$ barrier coincides with the classical impossibility bound
for prophet inequalities under arbitrary (non-identical) reward distributions
\cite{samuel1984comparison}.
This suggests that, in the worst-case boundary regime between stationarity and nonstationarity
induced by decaying rewards ($S \to \infty$, $\widetilde{\gamma}\to 1$ such that $\widetilde{\gamma}^S\to 0$), \emph{the stopping problem can be as hard as
the fully general non-identically distributed setting.} 


Having established that the \(1/2\) barrier extends beyond single-quantile threshold rules in the canonical phase setting, we conclude our theoretical analysis with a complementary infinite-horizon benchmark that captures the long-run effect of discounting beyond the finite-horizon model.

%% file: STOC2027/infinite.tex
\section{Infinite-Horizon Geometric Discounting}\label{sec:continuous_decay}

We now apply the effective-horizon calibration to geometric discounting over an infinite sequence of observations, with one observation per phase. This benchmark connects the finite-phase analysis to the classical discounted prophet inequality of \citet*{boshuizen1994optimal}. The optimal ratio is known; our purpose here is to show that the explicit quantile calibration attains it.

Fix $0<\widetilde\gamma<1$, let $(Y_i)_{i\ge1}$ be i.i.d.\ with base distribution $D\in\mathfrak D$, and set $X_i=\widetilde\gamma^{\,i-1}Y_i$. Recall that $0<\mathbb E_D[Y]<\infty$. In particular,
\[
0<\mathbb E_D\!\left[\sup_{i\ge1}X_i\right]
\le \sum_{i\ge1}\mathbb E_D[X_i]
=\frac{\mathbb E_D[Y]}{1-\widetilde\gamma}<\infty.
\]
A stopping rule takes values in $\mathbb N\cup\{\infty\}$, with payoff $X_\infty:=0$ if it never stops. We define
\[
\mathrm{CR}_\infty(\tau)
:=\inf_{D\in\mathfrak D}
\frac{\mathbb E_D[X_\tau]}{\mathbb E_D[\sup_{i\ge1}X_i]}.
\]

\paragraph{The classical benchmark.}
Corollary~4.5(ii) of \citet*{boshuizen1994optimal} establishes the sharp prophet-to-optimal-stopping factor $1+\widetilde\gamma$. Its sharpness gives the following worst-case upper bound. Appendix~\ref{app:S=n} includes a direct two-point verification.

\begin{theorem}[Classical geometric-discount upper bound]\label{prop:S=n}
For $0<\widetilde\gamma<1$, every stopping policy in the model above satisfies
\[
\mathrm{CR}_\infty(\tau)\le\frac{1}{1+\widetilde\gamma}.
\]
\end{theorem}

\paragraph{Attainment by effective-horizon calibration.}
The discounted mass of opportunities is
\[
B_\infty:=\sum_{i=1}^\infty\widetilde\gamma^{\,i-1}
=\frac{1}{1-\widetilde\gamma}.
\]
Thus the $1/B_\infty$ rule chooses acceptance probability $1-\widetilde\gamma$. The associated value threshold uses the known base distribution, but this probability depends only on the discount factor.

\begin{theorem}[Guarantee of the calibrated quantile rule]\label{prop:continuous_lower}
Fix $0<\widetilde\gamma<1$. Choose a base threshold $\alpha$ satisfying
\[
\mathbb P(Y>\alpha)\le1-\widetilde\gamma\le\mathbb P(Y\ge\alpha),
\]
and randomize independently at ties so that the acceptance probability is exactly $1-\widetilde\gamma$. At time $i$, use value threshold $\alpha_i=\widetilde\gamma^{\,i-1}\alpha$ and the same tie-breaking probability, and stop at the first acceptance. This infinite-horizon version of Algorithm~\ref{alg:phase-threshold} satisfies
\[
\mathrm{CR}_\infty(\tau)\ge\frac{1}{1+\widetilde\gamma}.
\]
\end{theorem}


The full proof of the above theorem  is provided in Appendix~\ref{app:contiuous_lower}. 
Together with Theorem~\ref{prop:S=n}, this proves minimax optimality of the calibration. It does not assert that this rule maximizes expected reward for each individual distribution: the distribution-dependent optimal stopping rule can perform better on particular instances.


We next illustrate the finite-sample behavior of the canonical finite-phase guarantee and the infinite-horizon geometric-discount calibration.

%% file: STOC2027/experiments.tex
\section{Experiments}

\begin{figure}[ht]
  \centering
  \begin{minipage}[t]{0.49\linewidth}
    \centering
    \includegraphics[width=\linewidth]{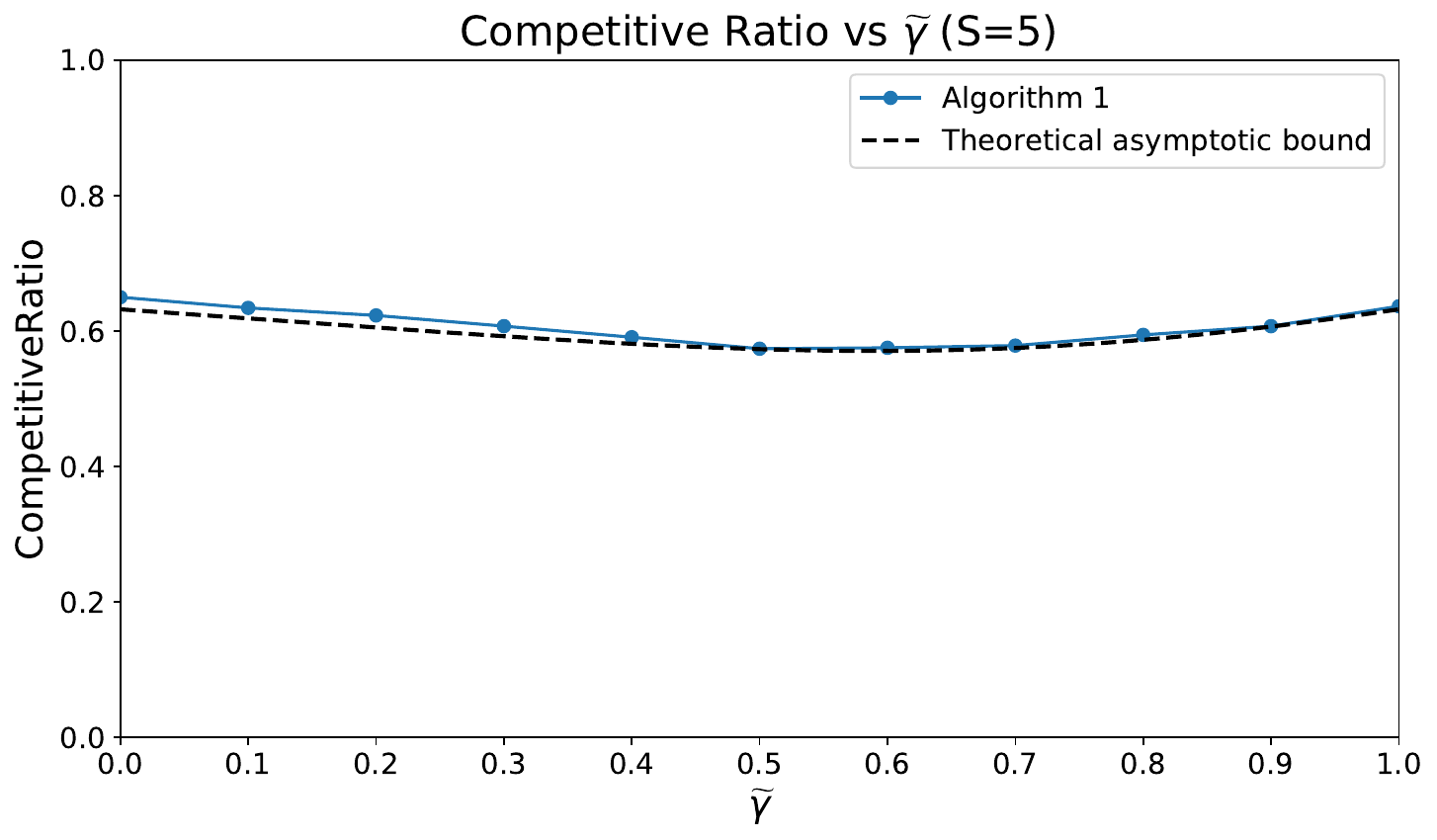}\\[-2mm]
    {\small (a) $S=5$ and $n=100000$}
  \end{minipage}\hfill
    \begin{minipage}[t]{0.49\linewidth}
    \centering
    \includegraphics[width=\linewidth]{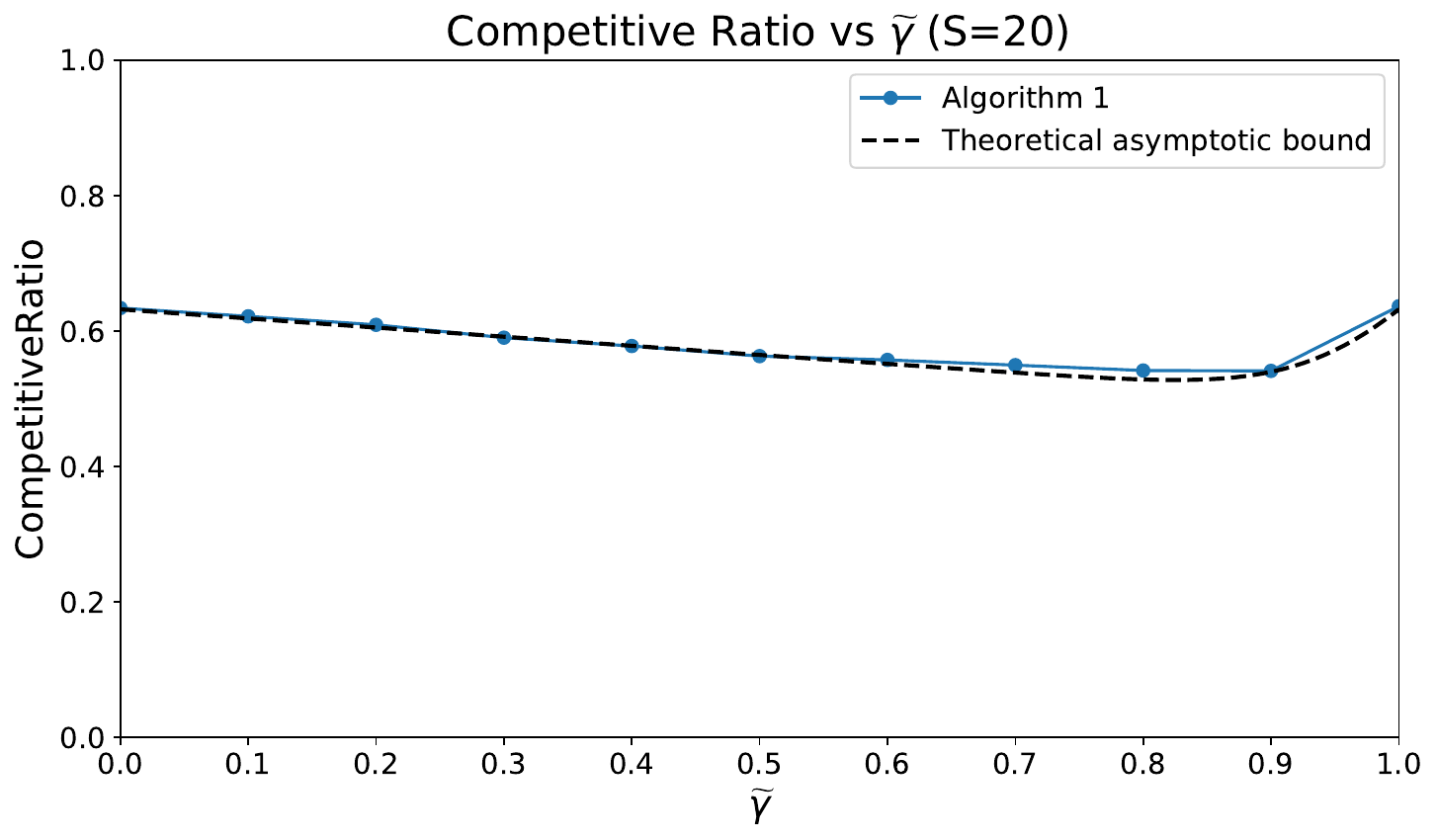}\\[-2mm]
    {\small (b) $S=20$ and $n=100000$}
  \end{minipage}
  \caption{Competitive ratio versus decay factor $\widetilde{\gamma}$.}
  \label{fig:ratio_vs_gamma}
\end{figure}

\begin{figure}[ht]
  \centering
    \begin{minipage}[t]{0.49\linewidth}
    \centering
    \includegraphics[width=\linewidth]{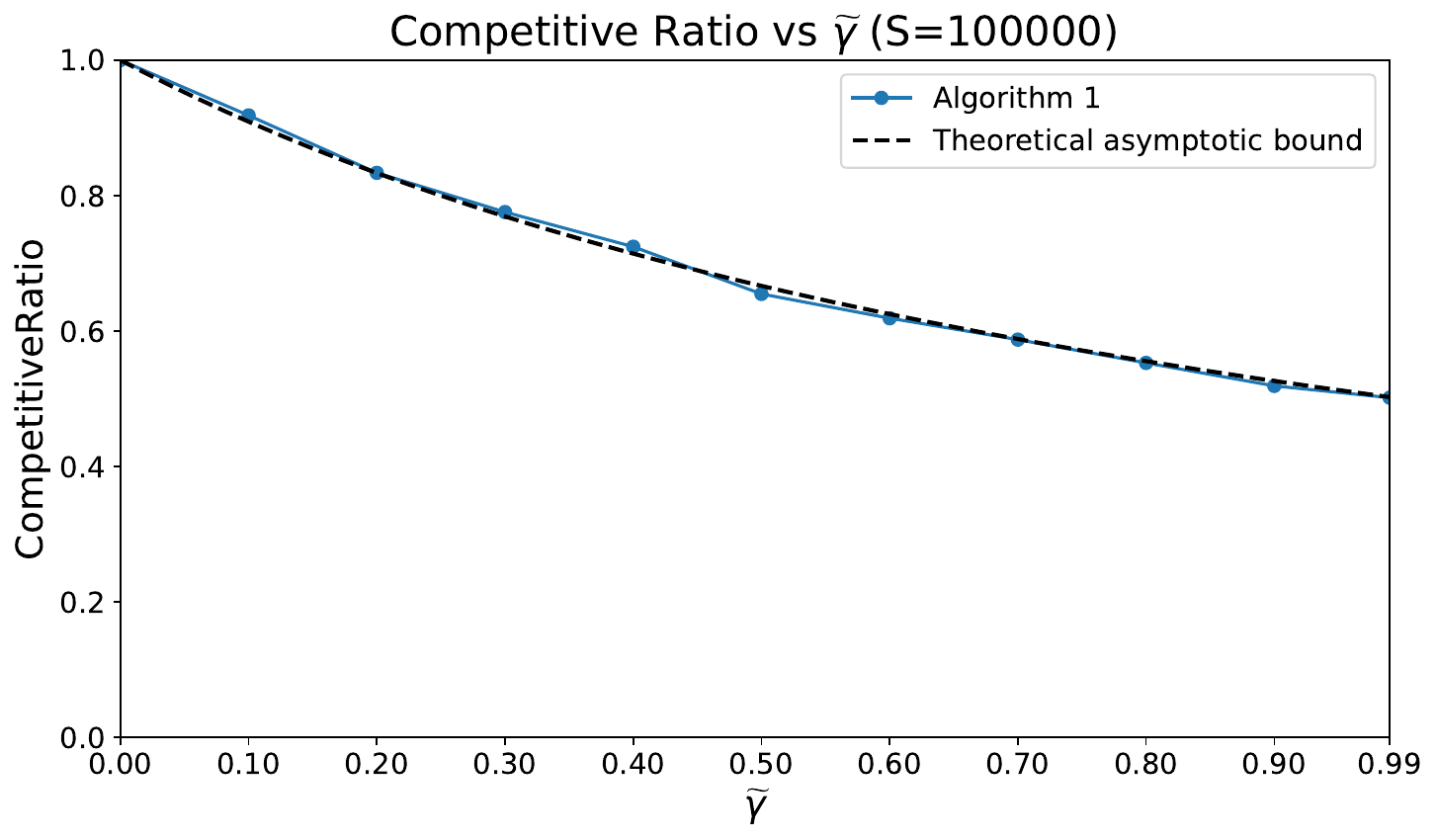}\\[-2mm]
    {\small (c) $S=n=100000$}
  \end{minipage}
  \caption{Competitive ratio versus the decay factor $\widetilde{\gamma}$ in the
  continuous-decay benchmark.}
  \label{fig:ratio_vs_gamma2}
\end{figure}

We present numerical experiments to validate the tightness of the
theoretical bounds and to illustrate their finite-sample behavior.
Motivated by the worst-case instance used to establish the upper bound in
Theorem~\ref{cor:upper}, we generate rewards as follows.
For each $i$, the base reward $Y_i$ takes value $n$ with probability $1/n^2$
and value $0.01$ otherwise.
The observed rewards are then defined as $X_i = \widetilde{\gamma}^{s-1} Y_i  \text{ for } i \in \mathcal{T}_s,$
with a common decay factor $\widetilde{\gamma}$ and equal-length phases $n_s = n/S$. We set $n=100000$.  Figure~\ref{fig:ratio_vs_gamma} reports the competitive ratio of
Algorithm~\ref{alg:phase-threshold} as a function of the decay factor~$\widetilde{\gamma}$
for $S \in \{5,20\}$. We estimate the numerator and denominator separately and report their ratio.
We observe that the empirical performance closely matches the asymptotically
tight theoretical bound characterized by Theorem~\ref{cor:upper} and
Theorem~\ref{thm:lower}. 
We additionally conduct experiments for the continuous-decay model with
$n = S = 100000$.
Figure~\ref{fig:ratio_vs_gamma2} shows that the empirical performance again
closely matches the corresponding tight theoretical guarantees
established in Theorems~\ref{prop:S=n} and~\ref{prop:continuous_lower}.

%% file: STOC2027/conclusion.tex
\section{Conclusion}

We characterized the optimal asymptotic single-quantile guarantee for equal-length phases with a common discount factor through matching bounds. The explicit guarantee recovers $1-1/e$ without discounting and approaches $1/2$ in the worst-case regime. The acceptance probability $1/B$ depends only on the phase profile, yet attains this optimum even when competing single-quantile policies may choose their probability using the known distribution.

The comparison and calibration extend to heterogeneous discounts and unequal phase lengths. For fixed $S$ and positive limiting phase proportions, the guarantees are asymptotically tight among single-quantile policies. The canonical worst-case limit remains $1/2$ even for arbitrary stopping rules. Under infinite-horizon geometric discounting, acceptance probability $1-\widetilde\gamma$ attains the known optimal worst-case ratio $1/(1+\widetilde\gamma)$ of \citet*{boshuizen1994optimal} without solving a Bellman equation.


%% file: STOC2027/comparision_value_quantile.tex
\subsection{Strict Separation Between Value and Quantile Thresholds}\label{app:comparison}

Here we provide a proposition to show that the worst-case asymptotic competitive guarantee achievable by single-value (or fixed-value) threshold policies is strictly smaller than that achievable by single-quantile (or fixed-quantile) threshold policies.
\begin{proposition}[Single value-thresholds are dominated by the single-quantile bound]
\label{thm:static-value-dominated-by-quantile}
Fix an integer $S\ge 2$ and a common decay factor $\widetilde{\gamma}\in(0,1)$.
Assume equal phase lengths,
\[
    n_s/n \to \frac1S,
    \qquad s=1,\ldots,S,
\]
and define
\[
    \Gamma_{s-1}:=\widetilde{\gamma}^{s-1},
    \qquad
    C_{\widetilde{\gamma},S}:=\frac1S\sum_{s=1}^S \widetilde{\gamma}^{s-1}
    =
    \frac{1-\widetilde{\gamma}^S}{S(1-\widetilde{\gamma})}.
\]
Let
\[
    r_{\widetilde{\gamma},S}
    :=
    \frac{1/S}{C_{\widetilde{\gamma},S}}
    =
    \frac{1-\widetilde{\gamma}}{1-\widetilde{\gamma}^S}.
\]
Let $R_q(\widetilde{\gamma},S)$ denote the single-quantile benchmark from Theorem~\ref{thm:lower}:
\[
    R_q(\widetilde{\gamma},S)
    :=
    \frac{
    \left(
    1-\widetilde{\gamma}^S
    \exp\!\left[-\frac{S(1-\widetilde{\gamma})}{1-\widetilde{\gamma}^S}\right]
    \right)
    \left(
    1-\exp\!\left[-\frac{1-\widetilde{\gamma}}{1-\widetilde{\gamma}^S}\right]
    \right)
    }{
    1-\widetilde{\gamma}
    \exp\!\left[-\frac{1-\widetilde{\gamma}}{1-\widetilde{\gamma}^S}\right]
    }.
\]
Then 
\[
    \limsup_{n\to\infty}\sup_{\tau_v\in \Pi_v}\mathrm{CR}_n(\tau_v)
    <
    R_q(\widetilde{\gamma},S).
\]
\end{proposition}

\begin{proof}
We use the same two-point hard instance as in the single-quantile upper
bound. For each $n$, let
\[
    q_n:=\frac1{n^2},
\]
and let $Y_1,\ldots,Y_n$ be i.i.d. with
\[
    Y_i=
    \begin{cases}
    n, & \text{with probability }q_n,\\
    y, & \text{with probability }1-q_n,
    \end{cases}
\]
where $y>0$ is a parameter to be chosen later. For $i\in T_s$, define
\[
    X_i=\widetilde{\gamma}^{s-1}Y_i.
\]

We first compute the prophet benchmark. If no rare value appears, then
the prophet obtains $y$, since the largest common observed value occurs
in the first phase. A rare value in phase $s$ has observed value
$\widetilde{\gamma}^{s-1}n$. Since the probability that a rare value appears in
phase $s$ is
\[
    \frac{n_s}{n^2}+o(1/n)
    =
    \frac{1}{Sn}+o(1/n),
\]
the rare values contribute asymptotically
\[
    \sum_{s=1}^S \widetilde{\gamma}^{s-1}n
    \left(\frac{1}{Sn}+o(1/n)\right)
    =
    C_{\widetilde{\gamma},S}+o(1).
\]
Therefore,
\[
    \mathbb{E}[\mathrm{OPT}]
    =
    y+C_{\widetilde{\gamma},S}+o(1).
\]

Now consider an arbitrary static observed-value threshold policy. Such a
policy fixes a single value threshold $\alpha$ and accepts the first
observation above $\alpha$, with arbitrary tie-breaking at atoms.

If $\alpha<y$, then the policy stops immediately at time $1$, because
$X_1=Y_1\ge y>\alpha$. Hence
\[
    \mathbb{E}[X_{\tau_v}]
    =
    y+o(1).
\]
If $\alpha>y$, then no common value can ever be accepted, because all
common observed values are at most $y$. Thus the policy can only collect
rare values, and its expected reward is at most the value of accepting
every rare value when it first appears:
\[
    \mathbb{E}[X_{\tau_v}]
    \le
    C_{\widetilde{\gamma},S}+o(1).
\]

It remains to consider the boundary case $\alpha=y$. Let $\rho_n\in[0,1]$
be the tie-breaking probability at the value $y$. The value $y$ occurs as
a common observed value only in phase $1$, since $\widetilde{\gamma}<1$ implies
$\widetilde{\gamma}^{s-1}y<y$ for all $s\ge2$.

Take any subsequence along which
\[
    n\rho_n\to d\in[0,\infty].
\]
For finite $d$, the probability of accepting a common value $y$ during
phase $1$ converges to
\[
 1-(1-(1-1/n^2)\rho_n)^{n_1} \to  1-e^{-d/S}.
\]
Thus the common-value contribution is
\[
    y(1-e^{-d/S}).
\]
Rare values in phase $1$ contribute
\[
    \frac{1}{n}\sum_{i=1}^{n/S}e^{-di/n}\to\int_0^{1/S} e^{-dt}\,dt
    =
    \frac{1-e^{-d/S}}{d},
\]
with the convention that this equals $1/S$ at $d=0$. Conditional on
surviving all common-value ties in phase $1$, which has limiting
probability $e^{-d/S}$, the expected rare-value contribution from phases
$2,\ldots,S$ is
\[
    C_{\widetilde{\gamma},S}-\frac1S.
\]
Therefore, for $\alpha=y$,
\[
    \limsup_{n\to\infty}
    \mathbb{E}[X_{\tau_v}]
    \le
    y(1-e^{-d/S})
    +
    \frac{1-e^{-d/S}}{d}
    +
    e^{-d/S}\left(C_{\widetilde{\gamma},S}-\frac1S\right).
\]
The endpoint $d=\infty$ gives the immediate-accept value $y$, and
$d=0$ gives the rare-only value $C_{\widetilde{\gamma},S}$. Hence all threshold
regimes are covered by
\[
    \limsup_{n\to\infty}
    \sup_{\tau_v}
    \mathbb{E}[X_{\tau_v}]
    \le
    \sup_{d\ge0}
    \left\{
    y(1-e^{-d/S})
    +
    \frac{1-e^{-d/S}}{d}
    +
    e^{-d/S}\left(C_{\widetilde{\gamma},S}-\frac1S\right)
    \right\}.
\]

Set
\[
    C:=C_{\widetilde{\gamma},S},
    \qquad
    r:=r_{\widetilde{\gamma},S}=\frac{1/S}{C},
    \qquad
    z:=\frac{y}{C},
    \qquad
    x:=\frac{d}{S}.
\]
After dividing by the prophet benchmark $y+C+o(1)$, the static-value
competitive ratio is bounded by
\[
    \sup_{x\ge0} F_r(z,x),
\]
where
\[
    F_r(z,x)
    :=
    \frac{
    z(1-e^{-x})
    +
    r\frac{1-e^{-x}}{x}
    +
    (1-r)e^{-x}
    }{
    z+1
    },
\]
with the convention $(1-e^{-x})/x=1$ at $x=0$. Optimizing over the hard
instance parameter $y$, equivalently over $z>0$, gives
\[
    \overline R_v^{\mathrm{stat}}(\widetilde{\gamma},S)
    :=
    \inf_{z>0}\sup_{x\ge0}F_r(z,x).
\]

We now compute this optimized upper bound. Define
\[
    a(x):=1-e^{-x},
\]
and
\[
    b_r(x):=
    r\frac{1-e^{-x}}{x}
    +
    (1-r)e^{-x}.
\]
Then
\[
    F_r(z,x)=\frac{za(x)+b_r(x)}{z+1}.
\]
The function $a$ is strictly increasing, while $b_r$ is strictly
decreasing. Moreover,
\[
    a(0)=0,\qquad b_r(0)=1,
\]
and
\[
    \lim_{x\to\infty}a(x)=1,
    \qquad
    \lim_{x\to\infty}b_r(x)=0.
\]
Hence there is a unique $x_r\in(0,\infty)$ satisfying
\[
    a(x_r)=b_r(x_r).
\]
This equation is exactly
\[
    r\left(
        \frac{1-e^{-x_r}}{x_r}-e^{-x_r}
    \right)
    =
    1-2e^{-x_r}.
\]
Since the left-hand side minus the right-hand side is positive at
$x=\ln2$ and nonpositive at $x=1$, we have
\[
    x_r\in[\ln2,1].
\]

For every $z>0$,
\[
    F_r(z,x_r)
    =
    \frac{za(x_r)+b_r(x_r)}{z+1}
    =
    a(x_r).
\]
Therefore,
\[
    \sup_{x\ge0}F_r(z,x)\ge a(x_r)
    \qquad
    \text{for every }z>0.
\]
Thus
\[
    \inf_{z>0}\sup_{x\ge0}F_r(z,x)\ge a(x_r).
\]

To prove the reverse inequality, define
\[
    h_r(x):=
    1-r
    +
    r\frac{e^x-x-1}{x^2}.
\]
A direct calculation gives
\[
    \frac{\partial}{\partial x}
    \bigl(za(x)+b_r(x)\bigr)
    =
    e^{-x}\bigl(z-h_r(x)\bigr).
\]
Moreover, $h_r$ is increasing on $(0,\infty)$ because
\[
    h_r'(x)
    =
    r\frac{e^x(x-2)+x+2}{x^3}
    \ge0,
\]
where the numerator is nonnegative since the function
\[
    \phi(x):=e^x(x-2)+x+2
\]
satisfies $\phi(0)=0$, $\phi'(0)=0$, and
\[
    \phi''(x)=xe^x\ge0.
\]
Now choose
\[
    z^\star:=h_r(x_r).
\]
Then $z^\star a(x)+b_r(x)$ is increasing on $(0,x_r)$ and decreasing on
$(x_r,\infty)$. Hence
\[
    \sup_{x\ge0}F_r(z^\star,x)
    =
    F_r(z^\star,x_r)
    =
    a(x_r).
\]
Therefore,
\[
    \overline R_v^{\mathrm{stat}}(\widetilde{\gamma},S)
    =
    a(x_r)
    =
    1-e^{-x_r}.
\]

It remains to compare this static-value upper bound with the
single-quantile benchmark. We first prove a slightly stronger
intermediate bound. Define
\[
    Q(r):=
    \frac{1-e^{-r}}{1-(1-r)e^{-r}}.
\]
We claim that
\[
    \overline R_v^{\mathrm{stat}}(\widetilde{\gamma},S)
    \le
    Q(r).
\]
Since
\[
    \overline R_v^{\mathrm{stat}}(\widetilde{\gamma},S)
    =
    \inf_{z>0}\sup_{x\ge0}F_r(z,x),
\]
it suffices to exhibit one value of $z$ whose supremum is at most
$Q(r)$. Let
\[
    z_0:=h_r(r)
    =
    1-r+r\frac{e^r-r-1}{r^2}
    =
    \frac{e^r-1-r^2}{r}.
\]
Since $e^r\ge1+r$, we have $z_0\ge0$. As above, because $h_r$ is
increasing, the function $z_0a(x)+b_r(x)$ is maximized at $x=r$.
Therefore,
\[
    \sup_{x\ge0}F_r(z_0,x)
    =
    F_r(z_0,r).
\]
A direct simplification yields
\[
    Q(r)-F_r(z_0,r)
    =
    \frac{
    (1-r)e^{-r}\bigl((e^r-1)^2-r^2e^r\bigr)
    }{
    (e^r-1+r)(e^r-1+r-r^2)
    }.
\]
The denominator is positive. The numerator is nonnegative because
\[
    (e^r-1)^2\ge r^2e^r
\]
is equivalent to
\[
    e^{r/2}-e^{-r/2}\ge r,
\]
which follows from $\sinh u\ge u$ with $u=r/2$. Hence
\[
    F_r(z_0,r)\le Q(r),
\]
and therefore
\[
    \overline R_v^{\mathrm{stat}}(\widetilde{\gamma},S)
    \le Q(r).
\]

Finally, we show that
\[
    Q(r)< R_q(\widetilde{\gamma},S).
\]
Using
\[
    r=\frac{1-\widetilde{\gamma}}{1-\widetilde{\gamma}^S},
\]
the single-quantile bound can be rewritten as
\[
    R_q(\widetilde{\gamma},S)
    =
    (1-e^{-r})
    \frac{1-(\widetilde{\gamma} e^{-r})^S}{1-\widetilde{\gamma} e^{-r}}.
\]
Let
\[
    a:=\widetilde{\gamma} e^{-r},
    \qquad
    b:=(1-r)e^{-r}.
\]
Then
\[
    R_q(\widetilde{\gamma},S)
    =
    (1-e^{-r})\frac{1-a^S}{1-a},
    \qquad
    Q(r)
    =
    \frac{1-e^{-r}}{1-b}.
\]
Thus it suffices to prove
\[
    \frac{1-a^S}{1-a}
    \ge
    \frac1{1-b}.
\]
This is equivalent to
\[
    a-b\ge (1-b)a^S.
\]
Since
\[
    r(1-\widetilde{\gamma}^S)=1-\widetilde{\gamma},
\]
we have
\[
    \widetilde{\gamma}-1+r=r\widetilde{\gamma}^S.
\]
Therefore,
\[
    a-b
    =
    (\widetilde{\gamma}-1+r)e^{-r}
    =
    r\widetilde{\gamma}^S e^{-r}.
\]
Also,
\[
    a^S=\widetilde{\gamma}^S e^{-rS}.
\]
Hence it remains to show
\[
    r e^{-r}
    >
    \left(1-(1-r)e^{-r}\right)e^{-rS}.
\]
Because
\[
    r=\frac{1}{1+\widetilde{\gamma}+\cdots+\widetilde{\gamma}^{S-1}},
\]
we have $r>1/S$, and therefore $rS>1$. Thus $e^{-rS}< e^{-1}$.
It is enough to prove
\[
    r e^{-r}
    >
    \frac1e\left(1-(1-r)e^{-r}\right).
\]
Multiplying by $e^r$ and then by $e$, this is equivalent to
\[
    er> e^r-1+r.
\]
But this follows from the convexity chord bound
\[
    e^r< (1-r)e^0+re^1=1+(e-1)r,
    \qquad r\in(0,1).
\]
Therefore,
\[
    Q(r)< R_q(\widetilde{\gamma},S).
\]

Combining the inequalities, we obtain
\[
    \overline R_v^{\mathrm{stat}}(\widetilde{\gamma},S)
    =
    1-e^{-x_{\widetilde{\gamma},S}}
    \le
    Q(r_{\widetilde{\gamma},S})
    <
    R_q(\widetilde{\gamma},S).
\]
This completes the proof.
\end{proof}

%% file: STOC2027/proof_upper.tex
\subsection{Proof of Theorem~\ref{cor:upper}: Upper bound for common decay and equal length phases}
\label{app:upper_adaptive}

Consider a horizon of $n$ stages and partition it into $S$ consecutive phases of equal length
(assume $S \mid n$ for simplicity; rounding effects do not affect the asymptotic bounds):
\[
\mathcal{T}_s
:=\Bigl\{(s-1)\tfrac{n}{S}+1,\dots, s\tfrac{n}{S}\Bigr\},
\qquad s\in[S].
\]
Let $Y_1,\dots,Y_n$ be i.i.d.\ base rewards and define the observed rewards by a phase-dependent decay:
\[
X_i:=\widetilde{\gamma}^{\,s-1}Y_i,
\qquad i\in\mathcal{T}_s,
\]
where $\widetilde{\gamma}\in[0,1]$. (In particular, the first phase is unscaled since $\widetilde{\gamma}^{0}=1$.)

Fix $0\le \widetilde{\gamma}\le 1$ and choose a fixed parameter $y_{\widetilde{\gamma}}\in[0,n)$ (specified later). Let
\[
Y_i \sim 
\begin{cases}
n, & \text{with probability } \tfrac{1}{n^2},\\[3pt]
y_{\widetilde{\gamma}}, & \text{with probability } 1-\tfrac{1}{n^2},
\end{cases}
\qquad i=1,\dots,n.
\]
For fixed $S\ge1$, $\widetilde{\gamma}\in (0,1]$ and $y_{\widetilde{\gamma}}\ge0$ (will be specified later), for large enough $n$, we have  \begin{equation}\label{eq:sep-condition}
\widetilde{\gamma}^{S-1}n \ \ge\ y_{\widetilde{\gamma}},
\end{equation}
so that every rare value $\widetilde{\gamma}^{s-1}n$ dominates the common value $y_{\widetilde{\gamma}}$ (and therefore also dominates
all scaled common values $\widetilde{\gamma}^{s-1}y_{\widetilde{\gamma}}$).
Let
\[
M:=\max_{1\le i\le n}X_i
\]
denote the offline optimum.

\medskip\noindent\textbf{Expected optimal reward.}
Let
\[
q_n:=\frac{1}{n^2},
\qquad
m:=\frac{n}{S},
\qquad
a_n:=(1-q_n)^m.
\]
For each \(s\in[S]\), let \(E_s\) denote the event that the first phase
containing a rare draw is phase \(s\); that is, no rare draw occurs in
phases \(1,\ldots,s-1\), and at least one rare draw occurs in phase \(s\).
Then
\[
\mathbb P(E_s)
=
a_n^{s-1}(1-a_n).
\]
Together with the event \(E_0\) that no rare draw occurs over the entire
horizon, where \(\mathbb P(E_0)=a_n^S\), the events
\(E_0,E_1,\ldots,E_S\) form a partition.

On \(E_s\), the earliest rare draw has the largest discount among all
rare draws, and hence
\[
M=\max\!\left\{y_{\widetilde{\gamma}},
\widetilde{\gamma}^{s-1}n\right\}.
\]
On \(E_0\), we have \(M=y_{\widetilde{\gamma}}\). Therefore, for every
\(0\le \widetilde{\gamma}\le 1\),
\begin{equation}\label{eq:opt-finite-n}
\mathbb E[M]
=
(1-a_n)\sum_{s=1}^{S}
a_n^{s-1}
\max\!\left\{y_{\widetilde{\gamma}},
\widetilde{\gamma}^{s-1}n\right\}
+
y_{\widetilde{\gamma}}a_n^S .
\end{equation}

Suppose first that \(\widetilde{\gamma}>0\). Since
\(y_{\widetilde{\gamma}}\) is fixed, for all sufficiently large \(n\),
\[
\widetilde{\gamma}^{S-1}n\ge y_{\widetilde{\gamma}}.
\]
Consequently, \eqref{eq:opt-finite-n} becomes
\[
\mathbb E[M]
=
n(1-a_n)\sum_{s=1}^{S}
\widetilde{\gamma}^{s-1}a_n^{s-1}
+
y_{\widetilde{\gamma}}a_n^S .
\]

When \(\widetilde{\gamma}=0\), only a rare draw in the first phase can
exceed \(y_{\widetilde{\gamma}}\). Hence
\[
\mathbb E[M]
=
n(1-a_n)
+
y_{\widetilde{\gamma}}
\left(
(1-a_n)\sum_{s=2}^{S}a_n^{s-1}+a_n^S
\right)
=
n(1-a_n)+y_{\widetilde{\gamma}}a_n.
\]

Finally,
\[
a_n
=
\left(1-\frac1{n^2}\right)^{n/S}
=
1-\frac{1}{Sn}+O\!\left(\frac1{n^2}\right),
\]
and therefore
\[
n(1-a_n)=\frac1S+O\!\left(\frac1n\right),
\qquad
a_n^k=1+O\!\left(\frac1n\right)
\quad\text{for each fixed }k.
\]
It follows in both cases that
\begin{equation}\label{eq:opt-asympt}
\mathbb E[M]
=
y_{\widetilde{\gamma}}
+
\frac1S\sum_{s=1}^{S}\widetilde{\gamma}^{s-1}
+
o(1)
=
y_{\widetilde{\gamma}}+c_{\widetilde{\gamma}}+o(1),
\end{equation}
where
\[
c_{\widetilde{\gamma}}
:=
\frac1S\sum_{s=1}^{S}\widetilde{\gamma}^{s-1}
=
\begin{cases}
\dfrac{1-\widetilde{\gamma}^{S}}
      {S(1-\widetilde{\gamma})},
& 0\le \widetilde{\gamma}<1,\\[1.2ex]
1,
& \widetilde{\gamma}=1.
\end{cases}
\]
\medskip\textbf{Expected performance of a single-threshold algorithm.}
We consider a single-threshold policy with $0 \le p \le 1$ that accepts $X_i$ if $\mathbb{P}(X_i \ge \alpha_s) = p$ for some $\alpha_s \ge 0$ for $s\in[S]$. Let $\tau$ be the (random) acceptance time and write $R_s(p):=\mathbb{E}[X_i\mid\text{accept at }i\in\mathcal{T}_s]$.
Because $Y_i$ is two-point, the conditional reward has an explicit form. If $0\le p\le 1/n^2$,
the algorithm can only accept the rare value, so
\[
R_s(p)=\widetilde{\gamma}^{s-1}n.
\]
If $p>1/n^2$, the algorithm accepts the rare value whenever it occurs and accepts the common value
with additional probability to reach total acceptance probability $p$, giving
\[
R_s(p)
=\frac{1}{p}\Bigl(\frac{\widetilde{\gamma}^{s-1}n}{n^2}+\widetilde{\gamma}^{s-1}y_{\widetilde{\gamma}}\bigl(p-\tfrac1{n^2}\bigr)\Bigr).
\]
In either case we have the uniform upper bound
\begin{equation}\label{eq:pRp-bound}
p\,R_s(p)\ \le\ \frac{\widetilde{\gamma}^{s-1}}{n}+\widetilde{\gamma}^{s-1}y_{\widetilde{\gamma}}\,p.
\end{equation}
Since the per-stage acceptance probability is $p$, we have $\mathbb{P}(\tau=i)=(1-p)^{i-1}p$ and hence
\begin{align}
\mathbb{E}[X_\tau]
&=\sum_{i=1}^n p\,R_{s}(p)\,(1-p)^{i-1}\nonumber\\
&=\sum_{s=1}^S\sum_{i\in\mathcal{T}_s} p\,R_s(p)\,(1-p)^{i-1}\nonumber\\
&\le \sum_{s=1}^S\Bigl(\frac{\widetilde{\gamma}^{s-1}}{n}+\widetilde{\gamma}^{s-1}y_{\widetilde{\gamma}}\,p\Bigr)
(1-p)^{(s-1)n/S}\sum_{j=1}^{n/S}(1-p)^{j-1}\nonumber\\
&=\sum_{s=1}^S\Bigl(\frac{\widetilde{\gamma}^{s-1}}{n}+\widetilde{\gamma}^{s-1}y_{\widetilde{\gamma}}\,p\Bigr)
(1-p)^{(s-1)n/S}\,\frac{1-(1-p)^{n/S}}{p}.\label{eq:alg-finite-n}
\end{align}

For an arbitrary sequence of single-threshold policies, let $p_n$
denote the per-stage acceptance probability and set $x_n:=np_n$. Passing to a subsequence along which the limsup of the competitive
ratio is attained, and then to a further subsequence, we may assume
that one of the following holds:
\[
x_n\to0,\qquad
x_n\to x\in(0,\infty),\qquad\text{or}\qquad
x_n\to\infty.
\]
If $p_n=0$ along an infinite subsequence, then the policy receives zero
reward on that subsequence, so the desired bound is immediate. Hence,
we assume below that $p_n>0$.
The boundary cases follow directly from \eqref{eq:alg-finite-n}.
If $x_n\to0$, then using
$1-(1-p_n)^{n/S}\le (n/S)p_n$ gives
\[
\mathbb E[X_{\tau}]
\le
c_{\widetilde{\gamma}}(1+y_{\widetilde{\gamma}}x_n),
\]
and therefore
\[
\limsup_{n\to\infty}\mathrm{CR}_n(\tau)
\le
\frac{c_{\widetilde{\gamma}}}
     {y_{\widetilde{\gamma}}+c_{\widetilde{\gamma}}}.
\]
If $x_n\to\infty$, the contribution from the first phase is at most
$y_{\widetilde{\gamma}}+1/x_n$, while the contributions from all
subsequent phases vanish because
\[
(1-p_n)^{(s-1)n/S}
\le
e^{-(s-1)x_n/S},
\qquad s\ge2.
\]
Hence
\[
\limsup_{n\to\infty}\mathrm{CR}_n(\tau)
\le
\frac{y_{\widetilde{\gamma}}}
     {y_{\widetilde{\gamma}}+c_{\widetilde{\gamma}}}.
\]
It therefore remains to consider the case
$x_n\to x\in(0,\infty)$. Since $p_n=x_n/n$, letting
$n\to\infty$ in \eqref{eq:alg-finite-n} yields
\begin{align}
\limsup_{n\to\infty}\mathbb{E}[X_{\tau}]
&\le \Bigl(\frac1x+y_{\widetilde{\gamma}}\Bigr)
\sum_{s=1}^S \widetilde{\gamma}^{s-1}
\exp\!\Bigl(-\frac{(s-1)x}{S}\Bigr)
\Bigl(1-e^{-x/S}\Bigr)\nonumber\\
&=\Bigl(\frac1x+y_{\widetilde{\gamma}}\Bigr)
\frac{1-\widetilde{\gamma}^S e^{-x}}
     {1-\widetilde{\gamma}e^{-x/S}}
\Bigl(1-e^{-x/S}\Bigr).
\label{eq:alg-limit}
\end{align}
Define
\begin{equation}\label{eq:def-fS-minimax}
f_S(x,y)
:=
\frac{\bigl(\frac1x+y\bigr)A_S(x)}
     {y+c_{\widetilde{\gamma}}},
\qquad x>0,\quad y\ge0,
\end{equation}
where
\begin{align}\label{eq:def-AS-equal}
A_S(x)
&:=
\sum_{s=1}^S \widetilde{\gamma}^{s-1}\exp\!\Bigl(-\frac{(s-1)x}{S}\Bigr)\Bigl(1-e^{-x/S}\Bigr),
\qquad x>0,
\end{align}
and
\begin{equation}\label{eq:def-cS}
c_{\widetilde{\gamma}} 
:=
\frac{1}{S}\sum_{s=1}^S \widetilde{\gamma}^{s-1}
=
\begin{cases}
\displaystyle \frac{1-\widetilde{\gamma}^S}{S(1-\widetilde{\gamma})}, & 0\le\widetilde{\gamma}<1,\\[6pt]
1,&\widetilde{\gamma}=1.
\end{cases}
\end{equation} 
Combining \eqref{eq:alg-limit} with \eqref{eq:opt-asympt} gives
\[
\limsup_{n\to\infty}\mathrm{CR}_n(\tau)
\le
f_S(x,y_{\widetilde{\gamma}}).
\]



\noindent\medskip\textbf{Minimax upper bound for $\inf_{y\ge0}\sup_{x> 0} f_S(x,y)$.}
We define the special point
\begin{equation}\label{eq:def-xgamma-equal}
x_{\widetilde{\gamma}}
:=\frac{1}{c_{\widetilde{\gamma}} }
=
\begin{cases}
\displaystyle \frac{S(1-\widetilde{\gamma})}{1-\widetilde{\gamma}^S}, & 0\le\widetilde{\gamma}<1,\\[6pt]
1,&\widetilde{\gamma}=1.
\end{cases}
\end{equation}

\medskip
\noindent\textit{Step 1: Construct $y_{\widetilde{\gamma}}$ such that $x_{\widetilde{\gamma}}$ is a global maximizer.}
Define
\begin{equation}\label{eq:def-yx-generalS}
y(x)
:=
\frac{A_S(x)}{x^2A_S'(x)}-\frac{1}{x},
\qquad x>0,
\end{equation}
and set
\begin{equation}\label{eq:def-ygamma-generalS}
y_{\widetilde{\gamma}}:=y(x_{\widetilde{\gamma}}).
\end{equation}
We will prove:
\begin{enumerate}
\item[(i)] $y_{\widetilde{\gamma}}\ge 0$;
\item[(ii)] $x\mapsto f_S(x,y_{\widetilde{\gamma}})$ attains its global maximum at $x=x_{\widetilde{\gamma}}$.
\end{enumerate}
Granting these facts,
\[
\sup_{x>0} f_S(x,y_{\widetilde{\gamma}})=f_S(x_{\widetilde{\gamma}},y_{\widetilde{\gamma}})=A_S(x_{\widetilde{\gamma}}),
\]
and hence
\begin{equation}\label{eq:upper-minimax}
\inf_{y\ge0}\sup_{x>0} f_S(x,y)
\ \le\
\sup_{x>0} f_S(x,y_{\widetilde{\gamma}})
=
A_S(x_{\widetilde{\gamma}}).
\end{equation}
It remains to justify (i)--(ii).

\medskip
\noindent\textit{Step 2: $y_{\widetilde{\gamma}}\ge0$ (concavity argument).}
From \eqref{eq:def-AS-equal},
\[
A_S(x)
=\sum_{s=1}^S \widetilde{\gamma}^{s-1}e^{-(s-1)x/S}-\sum_{s=1}^S \widetilde{\gamma}^{s-1}e^{-sx/S}
=\sum_{j=0}^S m_j e^{-xj/S},
\]
where $m_0=1$ and $m_j\le 0$ for all $j\ge 1$. Differentiating twice gives $A_S''(x)\le 0$ for $x>0$,
and also $A_S(0)=0$. Hence $A_S$ is concave on $(0,\infty)$, and the standard concavity inequality yields
\[
A_S(x)\ \ge\ xA_S'(x)\qquad (x>0).
\]
Therefore,
\[
y(x)=\frac{A_S(x)-xA_S'(x)}{x^2A_S'(x)}\ \ge\ 0,
\]
since $A_S'(x)>0$ for $x>0$. In particular, $y_{\widetilde{\gamma}}=y(x_{\widetilde{\gamma}})\ge0$.

\medskip
\noindent\textit{Step 3: Derivative factorization and reduction to monotonicity of $y(x)$.}
Differentiate \eqref{eq:def-fS-minimax} in $x$ (for fixed $y$):
\[
\frac{\partial}{\partial x}f_S(x,y)
=
\frac{1}{y+c_{\widetilde{\gamma}} }
\Big(
A_S'(x)\big(\tfrac1x+y\big)-\tfrac{A_S(x)}{x^2}
\Big).
\]
Using \eqref{eq:def-yx-generalS}, the bracket equals $A_S'(x)\,(y-y(x))$, so
\begin{equation}\label{eq:sign-derivative-factor}
\operatorname{sign}\!\Big(\frac{\partial}{\partial x}f_S(x,y)\Big)
=
\operatorname{sign}\big(y-y(x)\big),
\qquad x>0.
\end{equation}
Therefore, if $y(\cdot)$ is non-decreasing, then for $y=y_{\widetilde{\gamma}}=y(x_{\widetilde{\gamma}})$ we have
\[
x<x_{\widetilde{\gamma}} \implies y(x)\le y_{\widetilde{\gamma}} \implies \partial_x f_S(x,y_{\widetilde{\gamma}})\ge0,
\qquad
x>x_{\widetilde{\gamma}} \implies y(x)\ge y_{\widetilde{\gamma}} \implies \partial_x f_S(x,y_{\widetilde{\gamma}})\le0,
\]
so $f_S(\cdot,y_{\widetilde{\gamma}})$ is maximized at $x=x_{\widetilde{\gamma}}$, proving (ii).
Thus it suffices to prove the following lemma.

\begin{lemma}[Monotonicity of $y(x)$]\label{lem:y-monotone-equalS}
 Fix $\widetilde{\gamma}\in[0,1]$.
Then the function $y:(0,\infty)\to\mathbb R$ defined in \eqref{eq:def-yx-generalS} is non-decreasing:
\[
y'(x)\ \ge\ 0 \qquad\text{for all }x>0.
\]
\end{lemma}

\begin{proof}
The proof consists of 4 steps (A-D) as follows.

\paragraph{(A) An integral representation and the $m$-parametrization.}
Define the piecewise-constant weight function on $[0,1)$,
\[
w(u):=\widetilde{\gamma}^{\lfloor Su\rfloor}.
\]
For $k=0,1,2$ define
\[
L_k(x):=\int_0^1 w(u)\,u^k\,e^{-xu}\,du.
\]
From \eqref{eq:def-AS-equal},
\[
A_S(x)
=\sum_{s=1}^S \widetilde{\gamma}^{s-1}e^{-(s-1)x/S}-\sum_{s=1}^S \widetilde{\gamma}^{s-1}e^{-sx/S}\] Then partitioning $[0,1)$ into the $S$ equal subintervals yields
\begin{equation}\label{eq:A-xL0}
A_S(x)=\sum_{s=1}^S \widetilde{\gamma}^{s-1}\int_{(s-1)/S}^{s/S} x e^{-xu}\,du
= x \int_0^1 w(u)e^{-xu}\,du
= xL_0(x).
\end{equation}
Differentiating under the integral sign gives
\[
L_0'(x)=-L_1(x),\qquad L_1'(x)=-L_2(x),
\]
and therefore
\begin{equation}\label{eq:Aprime-L0L1}
A_S'(x)=L_0(x)-xL_1(x).
\end{equation}
Define the probability measure $\mathbb P_x$ on $[0,1]$ with density proportional to $w(u)e^{-xu}$:
\[
\mathbb P_x(du)=\frac{w(u)e^{-xu}}{L_0(x)}\,du,
\qquad L_0(x)>0.
\]
Let $U$ be the identity random variable under $\mathbb P_x$. Then
\[
m(x):=\mathbb E_x[U]=\frac{L_1(x)}{L_0(x)}.
\]
Using \eqref{eq:A-xL0}--\eqref{eq:Aprime-L0L1},
\[
y(x)=\frac{A_S(x)}{x^2A_S'(x)}-\frac1x
=\frac{xL_0(x)}{x^2(L_0(x)-xL_1(x))}-\frac1x
=\frac{L_1(x)}{L_0(x)-xL_1(x)}
=\frac{m(x)}{1-xm(x)}.
\]
Differentiate $m(x)=L_1(x)/L_0(x)$:
\begin{align*}
    m'(x)&=\frac{L_1'(x)L_0(x)-L_1(x)L_0'(x)}{L_0(x)^2}
=\frac{-L_2(x)L_0(x)+L_1(x)^2}{L_0(x)^2}
= -\Big(\frac{L_2(x)}{L_0(x)}-\Big(\frac{L_1(x)}{L_0(x)}\Big)^2\Big)\cr &=-\Big(\mathbb{E}_x[U^2]-(\mathbb{E}_x[U])^2\Big)=-\operatorname{Var}_x(U).
\end{align*}
Finally, differentiating $y(x)=m(x)/(1-xm(x))$ gives
\begin{equation}\label{eq:yprime-moment}
y'(x)=\frac{m'(x)+m(x)^2}{(1-xm(x))^2}
=\frac{m(x)^2-\operatorname{Var}_x(U)}{(1-xm(x))^2}.
\end{equation}
Since $1-xm(x)>0$ (equivalently $A_S'(x)>0$), it suffices to prove
\begin{equation}\label{eq:var-bound-target}
\operatorname{Var}_x(U)\ \le\ m(x)^2 \qquad (x>0).
\end{equation}

\paragraph{(B) Decomposition $U=(K+V)/S$ and independence.}
Let $K:=\lfloor SU\rfloor\in\{0,1,\dots,S-1\}$ and $V:=SU-K\in[0,1)$, so that $U=(K+V)/S$.
Fix $x>0$ and set
\[
a:=\frac{x}{S},\qquad t:=e^{-a}=e^{-x/S},\qquad r:=\widetilde{\gamma} t\in[0,1).
\]
For $k\in\{0,\dots,S-1\}$ and $v\in[0,1)$, write $u=(k+v)/S$, so $du=dv/S$ and $w(u)=\widetilde{\gamma}^k$. Then
\[
w(u)e^{-xu}\,du
=\widetilde{\gamma}^k e^{-x(k+v)/S}\frac{dv}{S}
=(\gamma e^{-a})^k e^{-av}\frac{dv}{S}
=r^k e^{-av}\frac{dv}{S}.
\]
Hence the joint law factorizes as
\[
\mathbb P_x(K=k,\ V\in dv)\propto r^k\cdot e^{-av}\,dv,
\]
so $K$ and $V$ are independent, with
\[
\mathbb P_x(K=k)=\frac{r^k}{\sum_{j=0}^{S-1}r^j},\qquad k=0,\dots,S-1,
\qquad
f_V(v)=\frac{a e^{-av}}{1-e^{-a}},\qquad v\in[0,1).
\]
Consequently,
\[
m(x)=\mathbb E_x[U]=\frac{\mathbb E[K]+\mathbb E[V]}{S},
\qquad
\operatorname{Var}_x(U)=\frac{\operatorname{Var}(K)+\operatorname{Var}(V)}{S^2},
\]
and \eqref{eq:var-bound-target} is equivalent to
\begin{equation}\label{eq:KV-ineq}
\operatorname{Var}(K)+\operatorname{Var}(V)\ \le\ (\mathbb E[K]+\mathbb E[V])^2.
\end{equation}

\paragraph{(C) Reduce \eqref{eq:KV-ineq} to one-dimensional bounds.}
For any square-integrable $Z$ define
\[
\Delta(Z):=(\mathbb E[Z])^2-\operatorname{Var}(Z)=2(\mathbb E[Z])^2-\mathbb E[Z^2].
\]
Then \eqref{eq:KV-ineq} is equivalent to $\Delta(K+V)\ge 0$, and by independence
\[
\Delta(K+V)=\Delta(K)+\Delta(V)+2\,\mathbb E[K]\mathbb E[V].
\]

\medskip
\noindent\textit{(i) The $V$-terms.}
A direct integration yields
\begin{equation}\label{eq:EV}
\mathbb E[V]
=\int_0^1 v\,\frac{a e^{-av}}{1-e^{-a}}\,dv
=\frac{1}{a}-\frac{1}{e^a-1},
\end{equation}
and
\begin{equation}\label{eq:DeltaV}
\Delta(V)=\frac{a e^a+a-2e^a+2}{a(e^a-1)^2}
=\frac{1-2\mathbb E[V]}{e^a-1}.
\end{equation}
Since $f_V$ is strictly decreasing on $[0,1)$ for $a>0$, we have $\mathbb E[V]<1/2$ and thus
$1-2\mathbb E[V]>0$.

\medskip
\noindent\textit{(ii) The $K$-terms.}
We use two elementary lemmas.

\begin{lemma}\label{lem:EK-upper}
Let $K$ have truncated geometric law $\mathbb P(K=k)\propto r^k$ on $\{0,\dots,S-1\}$ with $r\in[0,1)$.
Then
\[
\mathbb E[K]\ \le\ \frac{r}{1-r}.
\]
\end{lemma}
\begin{proof}
Let $G$ be geometric on $\{0,1,2,\dots\}$ with $\mathbb P(G=k)=(1-r)r^k$. Then
$K\stackrel{d}{=}G\mid(G\le S-1)$, so conditioning can only decrease the mean:
$\mathbb E[K]\le \mathbb E[G]=r/(1-r)$.
\end{proof}

\begin{lemma}\label{lem:DeltaK-lower}
Let $K$ have truncated geometric law $\mathbb P(K=k)\propto r^k$ on $\{0,\dots,S-1\}$ with $r\in[0,1)$.
Then
\[
\Delta(K)\ \ge\ -\mathbb E[K],
\qquad\text{equivalently}\qquad
\mathbb E[K^2]\ \le\ 2(\mathbb E[K])^2+\mathbb E[K].
\]
\end{lemma}
\begin{proof}
Write $G_0:=\sum_{k=0}^{S-1}r^k$, $G_1:=\sum_{k=0}^{S-1}k r^k$, $G_2:=\sum_{k=0}^{S-1}k^2 r^k$.
Then $\mathbb E[K]=G_1/G_0$ and $\mathbb E[K^2]=G_2/G_0$.
The desired inequality is equivalent to
\begin{equation}\label{eq:H-def}
H(r):=2G_1^2+G_0G_1-G_0G_2\ \ge\ 0.
\end{equation}
Using $G_0=(1-r^S)/(1-r)$ and the identities $G_1=rG_0'(r)$ and $G_2=rG_1'(r)$,
a direct computation gives
\[
H(r)= -\,\frac{S r^S}{(1-r)^3}\,B(r),
\]
where
\[
B(r):=(S+1)r-(S-1)+r^S\bigl((S-1)r-(S+1)\bigr).
\]
Since $S r^S/(1-r)^3>0$ for $r\in(0,1)$, it suffices to show $B(r)\le 0$ on $(0,1)$.
Differentiate:
\[
B'(r)=(S+1)\Bigl(1-r^{S-1}\bigl(S-(S-1)r\bigr)\Bigr).
\]
Apply AM--GM to the $S$ positive numbers $\underbrace{r,\dots,r}_{S-1\text{ times}}$ and $S-(S-1)r$,
whose arithmetic mean equals $1$; hence their geometric mean is at most $1$:
\[
r^{S-1}\bigl(S-(S-1)r\bigr)\le 1,
\]
so $B'(r)\ge 0$ on $(0,1)$. Since $B(1)=0$ and $B$ is increasing, $B(r)\le 0$ on $(0,1)$, proving
$H(r)\ge 0$ and the lemma.
\end{proof}

\medskip
\noindent\textit{Step 4: Complete the variance bound.}
Using $\Delta(K+V)=\Delta(K)+\Delta(V)+2\mathbb E[K]\mathbb E[V]$ and Lemma~\ref{lem:DeltaK-lower},
\[
\Delta(K+V)\ge -\mathbb E[K]+\Delta(V)+2\mathbb E[K]\mathbb E[V]
=\Delta(V)-\mathbb E[K]\bigl(1-2\mathbb E[V]\bigr).
\]
By \eqref{eq:DeltaV},
\[
\Delta(K+V)\ge \bigl(1-2\mathbb E[V]\bigr)\Bigl(\frac{1}{e^a-1}-\mathbb E[K]\Bigr).
\]
We already have $1-2\mathbb E[V]>0$. It remains to show $\mathbb E[K]\le 1/(e^a-1)$.
By Lemma~\ref{lem:EK-upper} and monotonicity of $u\mapsto u/(1-u)$ on $(0,1)$,
\[
\mathbb E[K]\le \frac{r}{1-r}
\le \frac{t}{1-t}
=\frac{e^{-a}}{1-e^{-a}}
=\frac{1}{e^a-1},
\]
because $r=\widetilde{\gamma} t\le t$. Hence $\Delta(K+V)\ge 0$, proving \eqref{eq:KV-ineq}, hence
\eqref{eq:var-bound-target}. Finally, \eqref{eq:yprime-moment} implies $y'(x)\ge 0$ for all $x>0$.
\end{proof}
Combining the interior estimate with the two boundary estimates, every
possible subsequential behavior of $x_n=np_n$ satisfies
\[
\limsup_{n\to\infty}\mathrm{CR}_n(\tau_n)
\le
\sup_{u>0}f_S(u,y_{\widetilde{\gamma}}).
\]
By the preceding argument,
$f_S(\cdot,y_{\widetilde{\gamma}})$ is maximized at
$u=x_{\widetilde{\gamma}}$. Therefore,
\[
\limsup_{n\to\infty}\mathrm{CR}_n(\tau_n)
\le
f_S(x_{\widetilde{\gamma}},y_{\widetilde{\gamma}})
=
A_S(x_{\widetilde{\gamma}}).
\]
This completes the proof of the theorem.

%% file: STOC2027/proof_dis.tex
\subsection{Proof of Lemma~\ref{lem:multistage}: Effective-horizon comparison for expected maxima}\label{app:dis}

Rather than treating the common-decay, equal-length case, we prove the
comparison directly in the general setting with arbitrary phase-wise decay
factors and phase lengths. In this setting, the effective horizon is defined as
\[
    B_G := \sum_{s=1}^{S} \Gamma_{s-1} n_s .
\]

\begin{lemma}\label{lem:multistage_gen} Let $F$ denote the cumulative distribution function of the base reward distribution $\mathcal{D}$. Fix an integer \(S\ge 1\), phase
lengths \(n_1,\dots,n_S\in\mathbb N\), and decay factors
\((\gamma_t)_{t=1}^{S-1}\in[0,1]^{S-1}\).
 Then, we have
\begin{equation}\label{eq:multistage_gen}
\int_0^\infty \Bigl(1-F(z)^{B_G}\Bigr)\,dz
\;\ge\;
\int_0^\infty \Bigl(1-\prod_{s=1}^{S} F\!\bigl(z/\Gamma_{\,s-1}\bigr)^{n_s}\Bigr)\,dz,
\end{equation}
where we adopt the convention $F(z/0):=\lim_{u\to\infty}F(u)=1$ for $z\ge0$.
\end{lemma}

\begin{proof}
For $w=(w_1,\dots,w_S)\in[0,1]^S$ define
\[
G_1(w)
:=\int_0^\infty\Bigl(1-F(z)^{\sum_{s=1}^S n_s w_s}\Bigr)\,dz,
\qquad
G_2(w)
:=\int_0^\infty\Bigl(1-\prod_{s=1}^S F(z/w_s)^{n_s}\Bigr)\,dz,
\]
with the convention $F(z/0)=1$ for $z>0$.
Note that \eqref{eq:multistage} is exactly $G_1(v)\ge G_2(v)$
for the geometric weight vector $v=(1,\Gamma_1,\Gamma_2,\dots,\Gamma_{S-1})$.

\medskip
\noindent\textbf{Step 1: $G_1$ is concave in $w$.}
Fix $z\ge 0$ and set $a:=F(z)\in[0,1]$.
The function $u\mapsto 1-a^u$ is concave on $[0,\infty)$ because
\[
\frac{d^2}{du^2}\bigl(1-a^u\bigr)=-a^u(\log a)^2\le 0.
\]
Since $u=\sum_{s=1}^S n_s w_s$ is affine in $w$, the map
$w\mapsto 1-F(z)^{\sum_s n_s w_s}$ is concave.
Integrating over $z$ preserves concavity (all integrals are finite by the
assumption $\int_0^\infty(1-F(z)^N)dz<\infty$), hence $G_1$ is concave on $[0,1]^S$.

\medskip
\noindent\textbf{Step 2: $G_2$ is convex in $w$.}
Let $\{Y_{i}\}_{i\in \mathcal{T}_s,\; 1\le s\le S}$ be independent i.i.d. copies of $Y\sim \mathcal{D}$
and define $M_s:=\max_{i\in\mathcal{T}_s} Y_{i}$.
For $w_s>0$,
\[
\mathbb P\!\left(\max_{1\le s\le S} w_s M_s \le z\right)
=\prod_{s=1}^S \mathbb P(M_s\le z/w_s)
=\prod_{s=1}^S F(z/w_s)^{n_s}.
\]
Thus, using $\mathbb E[X]=\int_0^\infty \mathbb P(X>z)dz$ for $X\ge 0$,
\[
G_2(w)=\mathbb E\!\left[\max_{1\le s\le S} w_s M_s\right],
\]
and the same holds for $w_s=0$ by the convention $0\cdot M_s=0$ and $F(z/0)=1$.
For each fixed realization $m=(m_1,\dots,m_S)$,
the function $w\mapsto \max_s w_s m_s$ is a pointwise maximum of linear functions,
hence convex. Taking expectation preserves convexity, so $G_2$ is convex on $[0,1]^S$.

\medskip
\noindent\textbf{Step 3: Equality on step vectors.}
For $k\in\{1,\dots,S\}$ define the step vector
\[
e^{(k)}:=(\underbrace{1,\dots,1}_{k\text{ times}},\underbrace{0,\dots,0}_{S-k\text{ times}}).
\]
Then
\[
G_1(e^{(k)})
=\int_0^\infty\Bigl(1-F(z)^{\sum_{s=1}^k n_s}\Bigr)\,dz,
\]
and, since $F(z/1)=F(z)$ and $F(z/0)=1$,
\[
G_2(e^{(k)})
=\int_0^\infty\Bigl(1-\prod_{s=1}^k F(z)^{n_s}\Bigr)\,dz
=\int_0^\infty\Bigl(1-F(z)^{\sum_{s=1}^k n_s}\Bigr)\,dz
=G_1(e^{(k)}).
\]

\medskip
\noindent\textbf{Step 4: Represent $v$ as a convex combination of step vectors.}
Recall $\Gamma_{s}=\prod_{t\in[s]}\gamma_t$ and $\Gamma_0:=1$. Define weights
\[
p_k:=(1-\gamma_k)\Gamma_{k-1}\quad(k=1,\dots,S-1),\qquad p_S:=\Gamma_{S-1}.
\]
Then $p_k\ge 0$ and $\sum_{k=1}^S p_k=1$.
Moreover, for each $s\in\{1,\dots,S\}$,
\[
\sum_{k=s}^S p_k
=\sum_{k=s}^{S-1}(1-\gamma_k)\Gamma_{k-1}+\Gamma_{S-1}
=\Gamma_{s-1}.
\]
Hence
\[
v=(\Gamma_0,\Gamma_1,\dots,\Gamma_{S-1})=\sum_{k=1}^S p_k\, e^{(k)}.
\]

\medskip
\noindent\textbf{Step 5: Jensen concave--convex sandwich.}
Using concavity of $G_1$, convexity of $G_2$, and $G_1(e^{(k)})=G_2(e^{(k)})$,
\[
G_1(v)
=G_1\Bigl(\sum_{k=1}^S p_k e^{(k)}\Bigr)
\;\ge\;\sum_{k=1}^S p_k G_1(e^{(k)})
=\sum_{k=1}^S p_k G_2(e^{(k)})
\;\ge\;G_2\Bigl(\sum_{k=1}^S p_k e^{(k)}\Bigr)
=G_2(v).
\]
This is exactly \eqref{eq:multistage}.\end{proof}

%% file: STOC2027/proof_lower.tex
\subsection{Proof of Theorem~\ref{thm:lower}: Lower bound under common decay and equal-length phases}\label{app:lower}
For simplicity, we present the argument for atomless $\mathcal D$; when $\mathcal D$ has atoms,
the same bounds follow by the randomized tie-breaking (see Remark~\ref{rm:atom_proof}). We first provide an upper bound on the expected reward of the phase-dependent
threshold policy. We can write
\begin{align}
&\mathbb{E}[X_\tau]
\cr &= \mathbb{E}\Bigg[
\sum_{s=1}^{S} \sum_{i \in \mathcal{T}_s}
X_i \,\mathbbm{1}\{X_i > \alpha_s\}\,
\mathbbm{1}\Big\{
X_j \le \alpha_t \ \forall t < s,\ \forall j \in \mathcal{T}_t,
\ \text{and}\
X_j \le \alpha_s \ \forall j \in \mathcal{T}_s \cap \{1,\dots,i-1\}
\Big\}
\Bigg].\cr
\end{align}
By independence of the $X_i$’s, this yields
\begin{align}
\mathbb{E}[X_\tau]
&= \sum_{i\in \mathcal{T}_1}F_1(\alpha_1)^{i-1}\mathbb{E}[X_i\mathbbm{1}\{X_i\ge \alpha_1\}]\cr& \qquad+ \sum_{s=2}^{S} \sum_{i \in \mathcal{T}_s}
\Bigg(\prod_{t=1}^{s-1} F_t(\alpha_t)^{n_t}\Bigg)
F_s(\alpha_s)^{\,i-\sum_{t=1}^{s-1}n_t-1}\,
\mathbb{E}\big[X_i \mathbbm{1}\{X_i > \alpha_s\}\big],
\end{align}
where $F_s$ is the CDF of $X_i$ in phase $s$. In our setting, $X_i = \widetilde{\gamma}^{s-1} Y_i$ with $Y_i \sim \mathcal{D}$ i.i.d.,
and the thresholds are chosen so that, from $Y\sim\mathcal{D}$,
\[
\mathbb{P}(Y > \alpha_1) = \frac{a}{n}\quad \text{and}\quad\mathbb{P}(\widetilde{\gamma}^{s-1}Y > \alpha_s) = \frac{a}{n}\qquad \text{for } s \in [2,S],
\]
which implies $F_s(\alpha_s) = 1 - a/n$ for all $s$ and
\[
\mathbb{E}[X_i \mathbbm{1}\{X_i > \alpha_s\}]
= \widetilde{\gamma}^{s-1}\mathbb{E}[Y \mathbbm{1}\{Y > \alpha_1\}],
\]
by construction of the thresholds. Therefore with $F_1(\alpha_1)=F_s(\alpha_s)$, we have
\begin{align*} 
\mathbb{E}[X_\tau]
&= \left(
\sum_{i \in \mathcal{T}_1} F_1(\alpha_1)^{\,i-1}
+\sum_{s=2}^{S}
\sum_{i \in \mathcal{T}_s}
F_1(\alpha_1)^{\sum_{t=1}^{s-1}n_t} F_1(\alpha_1)^{\,i-\sum_{t=1}^{s-1}n_t-1}
\,\widetilde{\gamma}^{s-1}\right)\,\mathbb{E}[Y \mathbbm{1}\{Y > \alpha_1\}] \nonumber\\
&= \left( \sum_{i \in \mathcal{T}_1} F_1(\alpha_1)^{\,i-1}+\sum_{s=2}^{S}
F_1(\alpha_1)^{\sum_{t=1}^{s-1}n_t} 
\left(\sum_{i \in \mathcal{T}_s} F_1(\alpha_1)^{\,i-\sum_{t=1}^{s-1}n_t-1}\right)
\widetilde{\gamma}^{s-1}\right)\mathbb{E}[Y \mathbbm{1}\{Y > \alpha_1\}].
\end{align*}
Using the geometric-series identity, $F_1(\alpha_1) = 1 - a/n$, and $n_s=\frac{n}{S}+O(1)$, we obtain
\begin{align}\label{eq:exp_reward_alg}
\mathbb{E}[X_\tau]
&= \frac{1 - F_1(\alpha_1)^{(n/S)+O(1)}}{a}
\left(1+\sum_{s=2}^{S} \big(F_1(\alpha_1)\big)^{(s-1) (n/S)+O(1)} \widetilde{\gamma}^{s-1}\,
\right)n\,\mathbb{E}[Y \mathbbm{1}\{Y > \alpha_1\}].
\end{align}

Now we provide a lower bound for \eqref{eq:exp_reward_alg}.
Recall that the effective horizon is
\[
B:=\sum_{s=1}^S n_s\widetilde{\gamma}^{\,s-1}
=
\frac{n}{S}\left(1+\sum_{s=1}^{S-1}
\widetilde{\gamma}^{\,s}\right)
\]
under equal-length phases. 
For any real \(b\ge 1\), define the fractional i.i.d.\ maximum benchmark
\[
\mathcal M_b(F)
:=
\int_0^\infty \bigl(1-F(z)^b\bigr)\,dz.
\]
By Lemma~\ref{lem:multistage} and the standard survival-function
representation of the prophet benchmark,
\begin{align}
\mathcal M_B(F)
=
\int_0^\infty \bigl(1-F(z)^B\bigr)\,dz
\ge
\mathbb E\left[\max_{i\in[n]}X_i\right].
\label{eq:fractional_max}
\end{align}
We next relate \(\mathcal M_B(F)\) to the truncated tail moment in
\eqref{eq:exp_reward_alg}. Since
\(\mathbb P(Y>\alpha_1)=a/n=1/B\), we have
\begin{align}
B\mathbb E\!\left[Y\mathbbm{1}\{Y>\alpha_1\}\right]
&=
B\mathbb E[(Y-\alpha_1)_+]
+
B\alpha_1\mathbb P(Y>\alpha_1)
\nonumber\\
&=
\alpha_1+B\mathbb E[(Y-\alpha_1)_+].
\label{eq:B=B+alpha}
\end{align}
Moreover,
\begin{align}
\mathcal M_B(F)
&=
\int_0^{\alpha_1}\bigl(1-F(z)^B\bigr)\,dz
+
\int_{\alpha_1}^{\infty}\bigl(1-F(z)^B\bigr)\,dz
\nonumber\\
&\le
\alpha_1
+
B\int_{\alpha_1}^{\infty}\bigl(1-F(z)\bigr)\,dz
\nonumber\\
&=
\alpha_1+B\mathbb E[(Y-\alpha_1)_+],
\label{eq:fractional_max_truncation}
\end{align}
where the inequality uses
\[
1-u^B\le B(1-u),
\qquad u\in[0,1],\quad B\ge1.
\]
Combining \eqref{eq:B=B+alpha},
\eqref{eq:fractional_max_truncation}, and
\eqref{eq:fractional_max}, we obtain
\begin{align}
B\mathbb E\!\left[Y\mathbbm{1}\{Y>\alpha_1\}\right]
\ge
\mathcal M_B(F)
\ge
\mathbb E\left[\max_{i\in[n]}X_i\right].
\label{eq:upper_prophet}
\end{align}
Therefore, using \eqref{eq:exp_reward_alg} and 
$
F_1(\alpha_1)=1-\frac{a}{n},$ we have 
\begin{align}
\mathbb{E}[X_\tau]
&=\frac{1-F_1(\alpha_1)^{(n/S)+O(1)}}{a}
\Bigl(1+\sum_{s=2}^{S}F_1(\alpha_1)^{(s-1)(n/S)+O(1)}\,\widetilde{\gamma}^{\,s-1}\Bigr)
\, n\,\mathbb{E}\big[Y\mathbbm{1}\{Y>\alpha_1\}\big]\notag\\
&=\frac{1-F_1(\alpha_1)^{(n/S)+O(1)}}{a}
\Bigl(1+\sum_{s=1}^{S-1}F_1(\alpha_1)^{s(n/S)+O(1)}\,\widetilde{\gamma}^{\,s}\Bigr)
\, n\,\mathbb{E}\big[Y\mathbbm{1}\{Y>\alpha_1\}\big]\notag\\
&=\Bigl(1-F_1(\alpha_1)^{(n/S)+O(1)}\Bigr)
\Bigl(1+\sum_{s=1}^{S-1}F_1(\alpha_1)^{s(n/S)+O(1)}\,\widetilde{\gamma}^{\,s}\Bigr)
\, B\,\mathbb{E}\big[Y\mathbbm{1}\{Y>\alpha_1\}\big]\notag\\
&\ge
\Bigl(1-F_1(\alpha_1)^{n/S}\Bigr)
\Bigl(1+\sum_{s=1}^{S-1}F_1(\alpha_1)^{sn/S}\,\widetilde{\gamma}^{\,s}\Bigr)
\,\mathbb{E}\Big[\max_{i\in[n]}X_i\Big]\notag\\
&=
\Bigl(1-(1-a/n)^{(n/S)+O(1)}\Bigr)
\Bigl(1+\sum_{s=1}^{S-1}\widetilde{\gamma}^{\,s}(1-a/n)^{s(n/S)+O(1)}\Bigr)
\,\mathbb{E}\Big[\max_{i\in[n]}X_i\Big],
\label{eq:X_tau_lower_bd_prelimit}
\end{align}
where the inequality follows from \eqref{eq:upper_prophet}, namely
$B\,\mathbb{E}[Y\mathbbm{1}\{Y>\alpha_1\}]\ge \mathbb{E}[\max_{i\in[n]}X_i]$.

For each fixed $s\in\{1,\dots,S-1\}$,
\[
(1-a/n)^{s(n/S)+O(1)}\to e^{-sa/S}
\qquad\text{and}\qquad
(1-a/n)^{(n/S)+O(1)}\to e^{-a/S}
\quad\text{as }n\to\infty.
\]
Therefore,
\begin{align}
\liminf_{n\to\infty}
\inf_{D}\frac{\mathbb{E}_D[X_\tau]}{\mathbb{E}_D\big[\max_{i\in[n]}X_i\big]}
\ge
\Bigl(1-e^{-a/S}\Bigr)
\Bigl(1+\sum_{s=1}^{S-1}\widetilde{\gamma}^{\,s}e^{-sa/S}\Bigr).
\label{eq:X_tau_lower_bd}
\end{align}
In particular, from $a/S=\frac{1}{1+\sum_{s=1}^{S-1}\widetilde{\gamma}^{\,s}}$, we have
\[
\liminf_{n\to\infty}\inf_{D}
\frac{\mathbb{E}_D[X_\tau]}{\mathbb{E}_D\big[\max_{i\in[n]}X_i\big]}
\ge
\Bigl(1-\exp\Bigl(-\frac{1}{1+\sum_{s=1}^{S-1}\widetilde{\gamma}^{\,s}}\Bigr)\Bigr)
\Bigl(1+\sum_{s=1}^{S-1}\widetilde{\gamma}^{\,s}
\exp\Bigl(-\frac{s}{1+\sum_{s=1}^{S-1}\widetilde{\gamma}^{\,s}}\Bigr)\Bigr).
\]
From the above, if $0\le\widetilde{\gamma}<1$, we have
\begin{align*}\liminf_{n\to\infty}\mathrm{CR}_n(\tau)&\ge  
    (1-\exp(-{(1-\widetilde{\gamma})/(1-\widetilde{\gamma}^{S}))})\left(1+\sum_{s=1}^{S-1}\widetilde{\gamma}^s\exp(-s(1-\widetilde{\gamma})/(1-\widetilde{\gamma}^{S}))\right) \cr &=\frac{\big(1-\widetilde{\gamma}^{S}\exp\!\big(-S\tfrac{1-\widetilde{\gamma}}{1-\widetilde{\gamma}^{S}}\big)\big)
      \big(1-\exp\!\big(-\tfrac{1-\widetilde{\gamma}}{1-\widetilde{\gamma}^{S}}\big)\big)}
     {1-\widetilde{\gamma}\,\exp\!\big(-\tfrac{1-\widetilde{\gamma}}{1-\widetilde{\gamma}^{S}}\big)},
\end{align*}
and if $\widetilde{\gamma}=1$, by telescoping with \eqref{eq:X_tau_lower_bd}, we have \[\liminf_{n\to\infty}\mathrm{CR}_n(\tau)\ge 1-1/e,\]  which concludes the proof.
\begin{remark}[Extension to atomic distributions]\label{rm:atom_proof}
To extend the analysis to base distributions with atoms, we replace the event
$\{X_i > \alpha_s\} (=\{\widetilde{\gamma}^{s-1}Y_i>\alpha_s\})$ by a randomized acceptance event
\[
A_i := \{X_i > \alpha_s\} \cup \{X_i = \alpha_s,\; U_i \le \theta\},
\]
where $U_i \sim \mathrm{Unif}[0,1]$ is independent of $Y_i$, and
\[
\theta := \frac{a/n - \mathbb{P}(X_i > \alpha_s)}{\mathbb{P}(X_i = \alpha_s)} \; \text{ if \;$\mathbb{P}(X_i=\alpha_s)>0$}.
\]
This construction ensures that $\mathbb{P}(A_i) = a/n$, and all
analytical arguments go through unchanged by replacing
$\mathbbm{1}\{X_i>\alpha_s\}$ with $\mathbbm{1}\{A_i\}$.
\end{remark}

%% file: STOC2027/proof_DP.tex
\subsection{Proof of Theorem~\ref{prop:dp-worst}: Worst-case analysis for dynamic programming (DP)}\label{app:DP}
We assume equal phase lengths $n_s=n/S$ for simplicity; rounding effects are $O(1)$ and do not affect the asymptotic bound.   We first show the upper bound of the competitive ratio for the policy of dynamic programming $\tau_{DP}$. Then we finalize the proof with $\limsup_{n\to\infty}\mathrm{CR}_n(\tau)\le \limsup_{n\to\infty}\mathrm{CR}_n(\tau_{DP}).$ 

\vspace{3mm}
\textbf{Setup (common decay, equal-length phases).}
Fix an integer $S\ge 1$ and assume $S\mid n$. Let $m:=n/S$ and partition the horizon into $S$
consecutive phases
\[
\mathcal T_s:=\{(s-1)m+1,\ldots,sm\},\qquad s=1,\ldots,S.
\]
Assume a common decay factor $\gamma_t=\widetilde\gamma\in[0,1)$ for $t\in[S-1]$ so that for $i\in\mathcal T_s$,
\[
X_i=\widetilde\gamma^{\,s-1}Y_i.
\]
Consider the two-point ``hard'' base distribution:
\[
Y_i=
\begin{cases}
n & \text{with probability } q:=1/n^2,\\
y & \text{with probability } 1-q,
\end{cases}
\qquad i=1,\ldots,n,
\]
where $y\in[0,n)$ is a parameter.

\vspace{3mm}
\textbf{Dynamic programming value recursion.}
Let $V_i$ denote the optimal expected reward \emph{before} observing $X_i$, given we have not
stopped yet. Then the Bellman recursion is
\[
V_i=\mathbb E\big[\max\{X_i,V_{i+1}\}\big],\qquad V_{n+1}=0.
\]

\paragraph{Structure of the optimal policy.}
Fix $i\in\mathcal T_s$. If $Y_i=n$, then $X_i=\widetilde\gamma^{\,s-1}n$. Since future phases have
coefficients at most $\widetilde\gamma^{\,s-1}$ and always $Y_j\le n$, we have
$\max_{j\ge i+1} X_j\le \widetilde\gamma^{\,s-1}n$ almost surely, and hence $V_{i+1}\le
\widetilde\gamma^{\,s-1}n$. Therefore $\max\{X_i,V_{i+1}\}=X_i$ and it is optimal to stop.

Fix $i\in\mathcal T_s$ with $i<sm$ (i.e.\ not the last time in phase $s$). If $Y_i=y$, then
$X_i=\widetilde\gamma^{\,s-1}y$. Consider continuing to time $i+1$. Inside the same phase, the
coefficient remains $\widetilde\gamma^{\,s-1}$; at time $i+1$ we see $Y_{i+1}=n$ with probability
$q$ and $Y_{i+1}=y$ with probability $1-q$. Hence a lower bound on the continuation value is
\[
V_{i+1}\ \ge\ q\cdot \widetilde\gamma^{\,s-1}n + (1-q)\cdot \widetilde\gamma^{\,s-1}y
= \widetilde\gamma^{\,s-1}y + q\,\widetilde\gamma^{\,s-1}(n-y)
>\widetilde\gamma^{\,s-1}y = X_i,
\]
since $q>0$ and $n>y$. Thus $\max\{X_i,V_{i+1}\}=V_{i+1}$ and it is strictly better to continue.
Therefore the DP policy can only ever accept a $y$-realization at \emph{phase boundaries} $i=sm$.

These observations imply that any optimal policy is of the following form: accept immediately upon the first
occurrence of $Y=n$; otherwise, if no $n$ has appeared yet, accept $Y=y$ at the end of some phase.
Hence it suffices to consider the one-parameter family of stopping times $\{\tau_K\}_{K=1}^S$:
\[
\tau_K := \min\Big\{\,\min\{i\le Km : Y_i=n\},\ Km \Big\},\qquad K=1,\ldots,S,
\]
i.e.\ stop at the first $n$ observed up to the end of phase $K$, and if none occurs, stop at time
$Km$ and take $\widetilde\gamma^{K-1}y$. The DP optimal stopping time is
$\tau^\star\in\arg\max_{1\le K\le S}\mathbb E[X_{\tau_K}]$.

\paragraph{Expected reward and optimal phase.}
Let
\[
r := \Pr(\text{no $n$ occurs in a phase})=(1-q)^m,\qquad p:=1-r.
\]
Then $\Pr(\text{no $n$ occurs in phases }1,\ldots,K)=r^K$, and\newline
$\Pr(\text{the first phase containing an $n$ is }s)=r^{s-1}p$.
Therefore
\begin{align}
\mathbb E[X_{\tau_K}]
&= r^K\cdot \widetilde\gamma^{K-1}y
\;+\;\sum_{s=1}^K r^{s-1}p\cdot \widetilde\gamma^{s-1}n. \label{eq:EK}
\end{align}
Similarly, the prophet's expected maximum satisfies
\begin{align}
\mathbb E[\mathrm{OPT}]
&= r^S\cdot y
\;+\;\sum_{s=1}^S r^{s-1}p\cdot \widetilde\gamma^{s-1}n. \label{eq:OPT}
\end{align}

Compute the one-step increment using \eqref{eq:EK}:
\begin{align*}
\mathbb E[X_{\tau_{K+1}}]-\mathbb E[X_{\tau_K}]
&= r^{K+1}\widetilde\gamma^{K}y - r^{K}\widetilde\gamma^{K-1}y
\;+\; r^{K}p\,\widetilde\gamma^{K}n \\
&= r^{K}\widetilde\gamma^{K-1}\Big( \widetilde\gamma n(1-r) + y(\widetilde\gamma r-1)\Big).
\end{align*}
Since $r^{K}\widetilde\gamma^{K-1}\ge0$, the sign is determined by
\[
D := \widetilde\gamma n(1-r) + y(\widetilde\gamma r-1),
\]
which does not depend on $K$. Hence $\mathbb E[X_{\tau_K}]$ is monotone in $K$, and the optimizer
satisfies $K^\star\in\{1,S\}$ (ties possible).

\paragraph{Asymptotic upper bound.}
With $q=1/n^2$ and $m=n/S$,
\[
r=(1-q)^m=(1-1/n^2)^{n/S} = 1-\frac{1}{Sn}+o(1/n),
\qquad
p=1-r=\frac{1}{Sn}+o(1/n).
\]
Define the decay-average constant 
\[
c_{\widetilde\gamma}:=\frac{1}{S}\sum_{s=1}^S \widetilde\gamma^{s-1}
=
\begin{cases}
\dfrac{1-\widetilde\gamma^S}{S(1-\widetilde\gamma)}, & \widetilde\gamma\neq 1,\\[6pt]
1, & \widetilde\gamma=1.
\end{cases}
\]
Using \eqref{eq:EK}--\eqref{eq:OPT} and the above expansions, as $n\to\infty$, one obtains
\[
\mathbb E[\mathrm{OPT}] \to y + c_{\widetilde\gamma},
\qquad
\mathbb E[X_{\tau_1}] \to y + \frac{1}{S},
\qquad
\mathbb E[X_{\tau_S}] \to \widetilde\gamma^{S-1}y + c_{\widetilde\gamma}.
\]
Therefore the asymptotic competitive ratios for $K=1$ and $K=S$ are
\[
A(y):=\lim_{n\to\infty}\frac{\mathbb{E}[X_{\tau_1}]}{\mathbb{E}[OPT]}=\frac{y+1/S}{y+c_{\widetilde\gamma}},
\qquad
B(y):=\lim_{n\to\infty}\frac{\mathbb{E}[X_{\tau_S}]}{\mathbb{E}[OPT]}=\frac{\widetilde\gamma^{S-1}y+c_{\widetilde\gamma}}{y+c_{\widetilde\gamma}}.
\]
Since the DP policy selects the better of these two,
\[
\lim_{n\to\infty}\frac{\mathbb{E}[X_{\tau_{\mathrm{DP}}}]}{\mathbb{E}[OPT]}=\max\{A(y),B(y)\}.
\]

Assume $\widetilde\gamma\in[0,1)$.
Then $A(y)$ is increasing in $y$ (as $c_{\widetilde\gamma}\ge 1/S$) and $B(y)$ is decreasing in $y$
(as $\widetilde\gamma^{S-1}<1$). Hence $\max\{A(y),B(y)\}$ is minimized at the unique intersection
$A(y)=B(y)$:
\[
\frac{y+1/S}{y+c_{\widetilde\gamma}}=\frac{\widetilde\gamma^{S-1}y+c_{\widetilde\gamma}}{y+c_{\widetilde\gamma}}
\quad\Longleftrightarrow\quad
y+\frac{1}{S}=\widetilde\gamma^{S-1}y+c_{\widetilde\gamma}.
\]
Solving gives
\[
y(1-\widetilde\gamma^{S-1}) = c_{\widetilde\gamma}-\frac{1}{S}.
\]
Using
\[
c_{\widetilde\gamma}-\frac{1}{S}
= \frac{1}{S}\sum_{k=1}^{S-1}\widetilde\gamma^k
= \frac{1}{S}\cdot \frac{\widetilde\gamma(1-\widetilde\gamma^{S-1})}{1-\widetilde\gamma},
\]
we cancel $(1-\widetilde\gamma^{S-1})$ and obtain the worst-case parameter
\[
{
y^\star=\frac{\widetilde\gamma}{S(1-\widetilde\gamma)}\qquad(\widetilde\gamma<1).
}
\]
At $y=y^\star$,  the two ratios coincide, and as $n\to \infty$, the worst-case asymptotic competitive ratio is upper bounded by:
\[
\frac{\mathbb{E}[X_{\tau_{DP}}]}{\mathbb{E}[OPT]}
\to \frac{y^\star+1/S}{y^\star+c_{\widetilde\gamma}}
=\frac{1}{1+\widetilde\gamma-\widetilde\gamma^S}.
\]
In a joint limit
$\widetilde{\gamma}\uparrow 1$ and $S \to \infty$ such that
$\widetilde{\gamma}^{S}\to 0$, we have
\begin{align*}
    \frac{1}{1+\widetilde\gamma-\widetilde\gamma^S}\to \frac{1}{2}
\end{align*}

\paragraph{Optimality among all stopping rules.}
Let $\mathcal F_i:=\sigma(Y_1,\ldots,Y_i)=\sigma(X_1,\ldots,X_i)$ be the natural filtration.
A \emph{stopping rule} is an $(\mathcal F_i)$-stopping time $\tau$ taking values in $\{1,\ldots,n\}$.

The following proposition is the standard finite-horizon optimal-stopping characterization via the Snell envelope \citep[Section~1.1]{peskir2006optimal}. We include a proof for completeness.

\begin{proposition}[DP is optimal among all stopping rules]\label{prop:DP_optimal_S_equals_n}
Let $(V_i)_{i=1}^{n+1}$ be defined by the Bellman recursion (value \emph{before} observing $X_i$)
\[
V_{n+1}=0,\qquad 
V_i=\mathbb E\big[\max\{X_i,V_{i+1}\}\big],\qquad i=n,n-1,\ldots,1.
\]
Define the DP stopping time
\[
\tau_{\mathrm{DP}}:=\inf\{\,i\in\{1,\ldots,n\}: X_i\ge V_{i+1}\,\}.
\]
Then for every stopping rule $\tau$ we have
\[
\mathbb E[X_\tau]\ \le\ \mathbb E[X_{\tau_{{DP}}}].
\]
\end{proposition}

\begin{proof}
We recall the Snell-envelope argument and connect it to the scalar recursion above.

Let $Z_i:=X_i$. Define the process $(W_i)_{i=1}^{n+1}$ by
\begin{equation}\label{eq:Snell}
W_{n+1}:=0,\qquad 
W_i:=\max\Big\{Z_i,\ \mathbb E[W_{i+1}\mid\mathcal F_i]\Big\},\qquad i=n,\ldots,1.
\end{equation}
Then $W_i$ is $\mathcal F_i$-measurable and satisfies $W_i\ge Z_i$ almost surely for all $i$.

By construction,
\[
W_i\ \ge\ \mathbb E[W_{i+1}\mid\mathcal F_i],
\]
so $(W_i)_{i=1}^{n+1}$ is a supermartingale.

\smallskip
Let $\tau$ be any $(\mathcal F_i)$-stopping time with values in $\{1,\ldots,n\}$. Since $Z_\tau\le W_\tau$
a.s., it suffices to bound $\mathbb E[W_\tau]$. Because $(W_i)$ is a supermartingale and $\tau\le n$ is bounded, we may write
\[
\mathbb E[W_\tau]
=\mathbb E[W_1]+\sum_{i=1}^{n-1}\mathbb E\big[(W_{i+1}-W_i)\mathbf 1\{\tau>i\}\big].
\]
Now $\mathbf 1\{\tau>i\}$ is $\mathcal F_i$-measurable, hence
\[
\mathbb E\big[(W_{i+1}-W_i)\mathbf 1\{\tau>i\}\big]
=\mathbb E\Big[\mathbf 1\{\tau>i\}\,\mathbb E[W_{i+1}-W_i\mid\mathcal F_i]\Big]
\le 0,
\]
since $\mathbb E[W_{i+1}\mid\mathcal F_i]-W_i\le 0$. Therefore $\mathbb E[W_\tau]\le \mathbb E[W_1]$,
and thus
\begin{equation}\label{eq:domination}
\mathbb E[X_\tau]=\mathbb E[Z_\tau]\ \le\ \mathbb E[W_\tau]\ \le\ \mathbb E[W_1].
\end{equation}

Finally, define the stopping time
\[
\tau^\star:=\inf\{\,i\in\{1,\ldots,n\}: Z_i=W_i\,\}.
\]
On the event $\{i<\tau^\star\}$ we have $Z_i<W_i$, hence by \eqref{eq:Snell} necessarily
\[
W_i=\mathbb E[W_{i+1}\mid\mathcal F_i]\qquad\text{on }\{i<\tau^\star\}.
\]
This implies the stopped process $(W_{i\wedge\tau^\star})_{i=1}^{n+1}$ is a martingale, and since
$\tau^\star\le n$ is bounded,
\[
\mathbb E[W_{\tau^\star}]=\mathbb E[W_1].
\]
Moreover, by definition of $\tau^\star$ we have $W_{\tau^\star}=Z_{\tau^\star}=X_{\tau^\star}$ almost surely.
Hence
\begin{equation}\label{eq:attain}
\mathbb E[W_1]=\mathbb E[W_{\tau^\star}]
=\mathbb E[Z_{\tau^\star}]
=\mathbb E[X_{\tau^\star}].
\end{equation}
Combining \eqref{eq:domination} and \eqref{eq:attain} yields
\[
\mathbb E[X_\tau]\le \mathbb E[X_{\tau^\star}]=\mathbb E[X_{\tau_{DP}}]\qquad\text{for all stopping rules }\tau.
\]
\end{proof}
\textbf{Conclusion.}
From the above results, we conclude that for any stopping rule $\tau$,
\[
\mathrm{CR}_n(\tau)\le \mathrm{CR}_n(\tau_{DP}).
\]
Moreover, in the joint limit $n\to\infty$, $\widetilde{\gamma}\uparrow 1$, and $S\to\infty$
such that $\widetilde{\gamma}^{S}\to 0$,
\[
\limsup_{\substack{
\widetilde{\gamma}\uparrow1,\,
S\to\infty\\
\widetilde{\gamma}^S\to0}}
\limsup_{n\to\infty}\mathrm{CR}_n(\tau_{\mathrm{DP}})
\le 
\frac12.
\]


%% file: STOC2027/proof_general.tex
\subsection{Extension to General Decay Factors and Phase Lengths}\label{app:general}

The common-decay, equal-length setting makes the effective-horizon principle
transparent: discounting reduces the total effective mass of opportunities, and
the appropriate acceptance quantile should be calibrated to this discounted
horizon rather than to the raw horizon. We now show that the same idea extends to heterogeneous
phase profiles.\footnote{ChatGPT and Codex assisted in developing and refining proof ideas for this extension; see the AI methodology statement at the end of the paper.}

Specifically, we allow the phase-wise discount factors to vary arbitrarily,
with \(\gamma_s\in[0,1]\) for \(s\in[S-1]\), and allow the phase lengths
\(n_s\) to vary across phases, assuming that each phase occupies a
nonvanishing asymptotic fraction of the horizon. Formally, there exist phase
proportions \(\lambda_1,\dots,\lambda_S\in(0,1]\), with
\(\sum_{s=1}^S\lambda_s=1\), such that, as \(n\to\infty\),
\[
\frac{n_s}{n}\to \lambda_s, \qquad s\in[S].
\]
In this setting, the contribution of each phase is determined by two quantities:
its relative length \(\lambda_s\), and its cumulative reward scale
\(\Gamma_{s-1}\). Accordingly, define
\[
\Lambda_{s-1}:=\sum_{t=1}^{s-1}\lambda_t \quad (\Lambda_0:=0),
\qquad
C_S:=\sum_{t=1}^S\Gamma_{t-1}\lambda_t.
\]
We also define the effective horizon as
\[
B_G:=\sum_{s=1}^{S}\Gamma_{s-1}n_s .
\]
The effective-horizon comparison for expected maxima in
Lemma~\ref{lem:multistage} extends to this generalized horizon $B_G$
(Appendix~\ref{app:dis}). Moreover, since \(B_G/n\to C_S\), the quantity
\(C_S\) can be viewed as the normalized effective horizon associated with the
heterogeneous phase profile.

\paragraph{Effective-Horizon Calibration.}

Since
\(B_G=\sum_{s=1}^{S}\Gamma_{s-1}n_s\) is the effective number of discounted
opportunities, the natural analogue of the classical \(1/n\)-quantile rule \citep{correa2019recent} is
a \(1/B_G\)-quantile rule. 
For each phase \(s\in[S]\) with \(\Gamma_{s-1}>0\),\footnote{
If \(\Gamma_{s-1}=0\) for some phase \(s\), then all subsequent rewards are
identically zero, so the problem reduces to the shorter horizon ending before
phase \(s\).} choose the threshold \(\alpha_s\) so that
\begin{align}
\mathbb{P}(\Gamma_{s-1}Y>\alpha_s)
=
\frac{\widetilde a}{n}, \text{ where } \widetilde{a}=\frac{1}{(n_1/n)+\sum_{t=1}^{S-1}(n_{t+1}/n)\Gamma_t}.
\label{eq:threshold_gen}
\end{align}
This calibration satisfies $
{\widetilde a}/{n}
=
{1}/{B_G}
$ and $\widetilde a\to 1/C_S$. Under this calibration, the simple threshold update $\alpha_{s+1}=\gamma_s\alpha_s$ still applies. 
The policy uses the same single-parameter quantile rule as in the canonical
case, but expressed in the correct effective time scale. We show a lower bound of this policy and the proof is provided in Appendix~\ref{app:lower2}.

\begin{theorem}\label{thm:lower2}
Consider the decaying-reward model with \(S\ge1\) phases, decay factors
\((\gamma_s)_{s=1}^{S-1}\in[0,1]^{S-1}\), and phase proportions
\((\lambda_s)_{s=1}^{S}\in(0,1]^S\) with \(\sum_{s=1}^S\lambda_s=1\).
Then Algorithm~\ref{alg:phase-threshold}, with threshold $\alpha_1$ from 
\eqref{eq:threshold_gen}, satisfies
\[
\liminf_{n\to\infty}\mathrm{CR}_n(\tau)\ge
\sum_{s=1}^S
\Gamma_{s-1}
\exp\left(-\frac{\Lambda_{s-1}}{C_S}\right)
\left(1-\exp\left(-\frac{\lambda_s}{C_S}\right)\right).
\]
\end{theorem}

The expression in Theorem~\ref{thm:lower2} has a direct phase-by-phase
interpretation. Under the effective-horizon calibration, the probability of
reaching phase \(s\) is asymptotically
\[
\exp\left(-\frac{\Lambda_{s-1}}{C_S}\right),
\]
while the probability of accepting within phase \(s\), conditional on reaching
it, is
\[
1-\exp\left(-\frac{\lambda_s}{C_S}\right).
\]
Multiplying these terms by the reward scale \(\Gamma_{s-1}\) and summing over
phases gives exactly the right-hand side of the theorem. 


This formula recovers the earlier canonical setting as a special case. In
particular, under a common decay factor and equal-length phases, i.e.,
\(\Gamma_{s-1}=\widetilde{\gamma}^{s-1}\) and \(\lambda_s=1/S\), it reduces to
the bound derived for the common-decay, equal-length model in Theorem~\ref{thm:lower}.

\paragraph{Tightness of the Effective-Horizon Bound.} We next show that the effective-horizon guarantee above is tight within the single-threshold class considered here.
The proof is
deferred to Appendix~\ref{app:upper}.

\begin{theorem}\label{thm:upper}
Fix any \(S\ge 1\), decay factors
\((\gamma_s)_{s=1}^{S-1}\in[0,1]^{S-1}\), and phase proportions
\((\lambda_s)_{s=1}^S\in(0,1]^S\) with \(\sum_{s=1}^S\lambda_s=1\).
For any single-threshold stopping policy \(\tau\) with a fixed acceptance
quantile independent of the base distribution,
\[
\limsup_{n\to\infty}\mathrm{CR}_n(\tau) \le
\sum_{s=1}^S
\Gamma_{s-1}
\exp\left(-\frac{\Lambda_{s-1}}{C_S}\right)
\left(1 - \exp\left(-\frac{\lambda_s}{C_S}\right)\right).
\]
\end{theorem}

Combining Theorems~\ref{thm:lower2} and~\ref{thm:upper}, we obtain a tight characterization for such single-threshold policies in the general phase-dependent setting. Heterogeneous decay factors and unequal phase lengths enter only through the effective horizon used to calibrate the acceptance quantile and the corresponding phase-by-phase contribution terms in the competitive-ratio bound.

%% file: STOC2027/proof_lower_gen.tex
\subsection{Proof of Theorem~\ref{thm:lower2}: Lower bound  under  general decay factors and phase length}\label{app:lower2}
For simplicity, we present the argument for atomless $\mathcal D$; when $\mathcal D$ has atoms,
the same bounds follow by the randomized tie-breaking (see Remark~\ref{rm:atom_proof}). 

We first provide a lower bound on the expected reward of the phase-dependent
threshold policy.
We can write
\begin{align}
\mathbb{E}[X_\tau]
&= \mathbb{E}\Bigg[
\sum_{s=1}^{S} \sum_{i \in \mathcal{T}_s}
X_i \,\mathbbm{1}\{X_i > {\alpha}_s\}\,
\mathbbm{1}\Big\{
\begin{gathered}
X_j \le {\alpha}_t \ \forall t < s,\ \forall j \in \mathcal{T}_t,\\
\text{and}\quad X_j \le {\alpha}_s \ \forall j \in \mathcal{T}_s \cap \{1,\dots,i-1\}
\end{gathered}
\Big\}
\Bigg].
\end{align}
By independence of the $X_i$’s, this yields
\begin{align}
\mathbb{E}[X_\tau]
&= \sum_{i\in \mathcal{T}_1}F_1({\alpha}_1)^{i-1}\mathbb{E}[X_i\mathbbm{1}\{X_i\ge {\alpha}_1\}]\cr& \qquad+ \sum_{s=2}^{S} \sum_{i \in \mathcal{T}_s}
\Bigg(\prod_{t=1}^{s-1} F_t({\alpha}_t)^{n_t}\Bigg)
F_s({\alpha}_s)^{\,i-\sum_{t=1}^{s-1}n_t-1}\,
\mathbb{E}\big[X_i \mathbbm{1}\{X_i > {\alpha}_s\}\big],
\end{align}
where $F_s$ is the CDF of $X_i$ in phase $s$. In our setting, for $i\in\mathcal{T}_{s}$, $X_i = \Gamma_{s-1} Y_i$ with $Y_i \sim \mathcal{D}$ i.i.d.,
and the thresholds are chosen so that
\[
\mathbb{P}(Y > {\alpha}_1) = \frac{\widetilde{a}}{n}\quad \text{and}\quad\mathbb{P}(\Gamma_{s-1}Y > {\alpha}_s) = \frac{\widetilde{a}}{n}\qquad \text{for } s \in [2,S],
\]
which implies $F_s({\alpha}_s) = 1 - \widetilde{a}/n$ for all $s$ and
\[
\mathbb{E}[X_i \mathbbm{1}\{X_i > {\alpha}_s\}]
= \Gamma_{s-1}\mathbb{E}[Y \mathbbm{1}\{Y > {\alpha}_1\}],
\]
by construction of the thresholds.
 Therefore with $F_1({\alpha}_1)=F_s({\alpha}_s)$, we have
\begin{align}
\mathbb{E}[X_\tau]
&= \left(
\sum_{i \in \mathcal{T}_1} F_1({\alpha}_1)^{\,i-1}
+\sum_{s=2}^{S}
\sum_{i \in \mathcal{T}_s}
F_1({\alpha}_1)^{\sum_{t=1}^{s-1}n_t} F_1({\alpha}_1)^{\,i-\sum_{t=1}^{s-1}n_t-1}
\,\Gamma_{s-1}\right)\,\mathbb{E}[Y \mathbbm{1}\{Y > {\alpha}_1\}] \nonumber\\
&= \left( \sum_{i \in \mathcal{T}_1} F_1({\alpha}_1)^{\,i-1}+\sum_{s=2}^{S}
F_1({\alpha}_1)^{\sum_{t=1}^{s-1}n_t} 
\left(\sum_{i \in \mathcal{T}_s} F_1(\alpha_1)^{\,i-\sum_{t=1}^{s-1}n_t-1}\right)
\Gamma_{s-1}\right)\mathbb{E}[Y \mathbbm{1}\{Y > {\alpha}_1\}].
\end{align}
Using the geometric-series identity and $F_1({\alpha}_1) = 1 - \widetilde{a}/n$, we obtain
\begin{align}\label{eq:exp_reward_alg2}
\mathbb{E}[X_\tau]
&= 
\left(\frac{1 - F_1({\alpha}_1)^{n_1}}{\widetilde{a}}+\sum_{s=2}^{S}\frac{1 - F_1({\alpha}_1)^{n_s}}{\widetilde{a}} F_1({\alpha}_1)^{\sum_{t=1}^{s-1}n_t} \Gamma_{s-1}\,
\right)n\,\mathbb{E}[Y \mathbbm{1}\{Y > {\alpha}_1\}].
\end{align}
Now we provide a lower bound for \eqref{eq:exp_reward_alg2}.
Recall that
\[
B_G=\sum_{s=1}^{S}n_s\Gamma_{s-1}.
\] Define
\[
\mathcal M_{B_G}(F)
:=
\int_0^\infty \bigl(1-F(x)^{B_G}\bigr)\,dx.
\]
By Lemma~\ref{lem:multistage_gen},
\begin{align}
\mathcal M_{B_G}(F)
\ge
\mathbb E\left[\max_{i\in[n]}X_i\right].
\label{eq:upper_bd_max2}
\end{align}
Moreover, since
\(\mathbb P(Y>{\alpha}_1)=1/B_G\),
\begin{align}
B_G\mathbb E\!\left[
Y\mathbbm{1}\{Y>{\alpha}_1\}
\right]
=
{\alpha}_1
+
B_G\mathbb E[(Y-{\alpha}_1)_+].
\label{eq:B=B+alpha2}
\end{align}
Using \(1-u^{B_G}\le B_G(1-u)\) for \(u\in[0,1]\), we obtain
\begin{align}
\mathcal M_{B_G}(F)
&\le
{\alpha}_1
+
B_G\int_{{\alpha}_1}^{\infty}
(1-F(x))\,dx
\nonumber\\
&=
{\alpha}_1
+
B_G\mathbb E[(Y-{\alpha}_1)_+].
\label{eq:Y_max_bd_alpha2}
\end{align}
Combining the preceding inequalities gives
\begin{align}
B_G\mathbb E\!\left[
Y\mathbbm{1}\{Y>{\alpha}_1\}
\right]
\ge
\mathcal M_{B_G}(F)
\ge
\mathbb E\left[\max_{i\in[n]}X_i\right].
\label{eq:upper_prophet3}
\end{align}
Let
\[
\widetilde a
:=\frac{1}{(n_1/n)+\sum_{s=2}^{S}(n_s/n)\Gamma_{s-1}}
=\frac{n}{B_G},
\qquad
p:=F_1({\alpha}_1)=1-\frac{\widetilde a}{n}=1-\frac{1}{B_G}.
\]
Using \eqref{eq:exp_reward_alg2} and the identity $n/\widetilde a = B_G$, we can rewrite
\begin{align}
\mathbb{E}[X_\tau]
&=\left(\frac{1-p^{n_1}}{\widetilde a}
+\sum_{s=2}^{S}\frac{1-p^{n_s}}{\widetilde a}\,p^{\sum_{t=1}^{s-1}n_t}\,\Gamma_{s-1}\right)
n\,\mathbb{E}\big[Y\mathbbm{1}\{Y>{\alpha}_1\}\big]\notag\\
&=\left((1-p^{n_1})
+\sum_{s=2}^{S}(1-p^{n_s})\,p^{\sum_{t=1}^{s-1}n_t}\,\Gamma_{s-1}\right)
B_G\,\mathbb{E}\big[Y\mathbbm{1}\{Y>{\alpha}_1\}\big]\notag\\
&\ge
\left((1-p^{n_1})
+\sum_{s=2}^{S}(1-p^{n_s})\,p^{\sum_{t=1}^{s-1}n_t}\,\Gamma_{s-1}\right)
\mathbb{E}\Big[\max_{i\in[n]}X_i\Big]\notag\\
&=
\left((1-(1-1/B_G)^{n_1})
+\sum_{s=2}^{S}(1-(1-1/B_G)^{n_s})\,(1-1/B_G)^{\sum_{t=1}^{s-1}n_t}\,\Gamma_{s-1}\right)
\mathbb{E}\Big[\max_{i\in[n]}X_i\Big],
\label{eq:X_tau_lower_bd_prelimit_general}
\end{align}
where the inequality uses \eqref{eq:upper_prophet3}.

\medskip

\noindent
\textbf{Asymptotics.}
Since $\widetilde a = n/B_G$ is $O(1)$ (indeed, $\widetilde a$ depends only on the proportions $\lambda_s$),
we have $B_G=n/\widetilde a\to\infty$ as $n\to\infty$. Moreover, for each fixed $s$,
\[
(1-1/B_G)^{n_s}=\exp\!\Big(-\frac{n_s}{B_G}+o(1)\Big),
\qquad
(1-1/B_G)^{\sum_{t=1}^{s-1}n_t}
=\exp\!\Big(-\frac{\sum_{t=1}^{s-1}n_t}{B_G}+o(1)\Big).
\]
Plugging these into \eqref{eq:X_tau_lower_bd_prelimit_general} with $e^{o(1)}=1+o(1)$ yields
\begin{align}
\mathbb{E}[X_\tau]
&\ge
\Bigg(
\Bigl(1-e^{-n_1/B_G}\Bigr)
+\sum_{s=2}^{S}\Bigl(1-e^{-n_s/B_G}\Bigr)\Gamma_{s-1}
\exp\Bigl(-\frac{\sum_{t=1}^{s-1}n_t}{B_G}\Bigr)
+o(1)
\Bigg)\,
\mathbb{E}\Big[\max_{i\in[n]}X_i\Big].
\label{eq:X_tau_lower_bd_general_asymp}
\end{align}

In particular, if $n_s/n\to \lambda_s$ with $\sum_{s=1}^S\lambda_s=1$, then
$\widetilde a\to \big(\lambda_1+\sum_{s=2}^S\lambda_s\Gamma_{s-1}\big)^{-1}$ and
$n_s/B_G=\widetilde a\,(n_s/n)\to \big(\lambda_1+\sum_{s=2}^S\lambda_s\Gamma_{s-1}\big)^{-1}\lambda_s$, so \eqref{eq:X_tau_lower_bd_general_asymp} implies
\begin{align}
\liminf_{n\to\infty}\inf_{D}\frac{\mathbb{E}_D[X_\tau]}{\mathbb{E}_D[\max_{i\in[n]}X_i]}
\ge\sum_{s=1}^{S}\Bigl(1-\exp\Bigl(-\frac{\lambda_s}{\sum_{t=1}^S\lambda_t\Gamma_{t-1}}\Bigr)\Bigr)\Gamma_{s-1}\,
\exp\Bigl(-\frac{\sum_{t=1}^{s-1}\lambda_t}{\sum_{t=1}^S\lambda_t\Gamma_{t-1}}\Bigr).
\label{eq:X_tau_lower_bd_general_limit}
\end{align}

%% file: STOC2027/proof_upper_final.tex
\subsection{Proof of Theorem~\ref{thm:upper}: Upper bound under general decay factors and phase lengths}\label{app:upper}

We consider a time horizon of $n$ stages, partitioned into $S$ phases
\[
    \mathcal{T}_s \subset [n],
    \qquad |\mathcal{T}_s| = n_s>0,
    \qquad \sum_{s=1}^S n_s = n.
\]
We define
\[
    N_{s-1} := \sum_{t=1}^{s-1} n_t,
    \qquad N_0 := 0,
\]
and we have
\[
    \mathcal{T}_s = \{N_{s-1}+1,\dots,N_{s-1}+n_s\},\qquad s\in[S].
\]
Recall $\Gamma_{s-1}:=\prod_{t\in[s-1]}\gamma_t$. The reward sequence is defined by
\[
    X_i =
    \begin{cases}
        Y_i, & i\in\mathcal{T}_1,\\[4pt]
        \Gamma_{s-1} Y_i, & i\in\mathcal{T}_s,\; s \in [2,S],
    \end{cases}
\]
where $0 \le  \gamma_t \le  1$  for $t\in[s-1]$ represent the decay factors across phases.

Fix $0 \le \gamma_t < 1$ for $t\in[S-1]$ and a parameter
$1/n\le y \le \sqrt{n}$ (to be specified later).  
Let $Y_1, \dots, Y_n$ be i.i.d.\ random variables drawn from the two-point distribution
\[
    Y_i \sim 
    \begin{cases}
        n, & \text{with probability } \tfrac{1}{n^2},\\[3pt]
        y, & \text{with probability } 1 - \tfrac{1}{n^2}.
    \end{cases}
\]
The observed rewards $\{X_i\}_{i=1}^n$ are scaled according to the phase-dependent factors
$\Gamma_{s-1}$  defined above.
We will work in an asymptotic regime $n\to\infty$ where the relative phase lengths converge:
\[
    \lambda_s := \frac{n_s}{n} \in (0,1],
    \qquad
    \sum_{s=1}^S \lambda_s = 1,
\]
and hence
\[
    \frac{N_{s-1}}{n} = \sum_{r=1}^{s-1}\frac{n_r}{n}
    \longrightarrow
    \Lambda_{s-1} := \sum_{r=1}^{s-1}\lambda_r,
    \qquad s\in[S],
\]
where $N_{0}:=0$ and $\Lambda_{0}:=0$.

\vspace{3mm}
\textbf{Expected optimal reward.} We first consider $\Gamma_{s-1}>0$ for all $s$. 
As $n \to \infty$, the maximum reward is either $\Gamma_{s-1}n$ for some $s$, if a value
$Y_i=n$ appears first in phase $s$, or $y$ if no such value appears.
For each phase $s\in[S]$, the probability that the first occurrence of $Y_i=n$ falls in
$\mathcal{T}_s$ is
\[
    \mathbb{P}\left(\max_{i\in[n]}X_i = \Gamma_{s-1} n\right)
    =
    \big(1 - (1 - 1/n^2)^{n_s}\big)\,(1 - 1/n^2)^{N_{s-1}}.
\]
Using $n_s = \lambda_s n$ and $N_{s-1} = O(n)$, a Taylor expansion gives
\[
    (1 - 1/n^2)^{n_s}
    = 1 - \frac{n_s}{n^2} + o\bigl(1/n\bigr)
    = 1 - \frac{\lambda_s}{n} + o\bigl(1/n\bigr),
\]
and
\[
    (1 - 1/n^2)^{N_{s-1}}
    = 1 -\frac{N_{s-1}}{n^2}+ o\bigl(1/n\bigr).
\]
Hence
\[
    \mathbb{P}\left(\max_{i\in[n]}X_i = \Gamma_{s-1}n\right)
    =  \frac{\lambda_s}{n} + o\bigl(1/n\bigr).
\]
Similarly,
\[
    \mathbb{P}\left(\max_{i\in[n]}X_i = y\right)
    = (1 - 1/n^2)^n
    = 1 - \frac{1}{n} + o\bigl(1/n\bigr).
\]
Therefore,
\begin{align*}
    \mathbb{E}\left[\max_{i\in[n]} X_i\right]
    &= \sum_{s=1}^{S} \Gamma_{s-1} n \,\mathbb{P}\left(\max_{i\in[n]}X_i  = \Gamma_{s-1} n\right)
       + y \,\mathbb{P}\left(\max_{i\in[n]}X_i  = y\right) \\
    &\xrightarrow[n\to \infty]{}
       y + \sum_{s=1}^{S} \lambda_s \Gamma_{s-1}.
\end{align*}
Recall
\[
    C_S 
    := \sum_{s=1}^{S}\lambda_s \Gamma_{s-1},
\]
so that, asymptotically,
\[
    \mathbb{E}\left[\max_{i\in[n]} X_i\right] \;\to\; y + C_S .
\]
The same result can be obtained for the case when $\Gamma_{s-1}=0$ for some $s$ with $\Gamma_0:=1$.

\vspace{3mm}
\textbf{Expected performance of a single-threshold algorithm.}
We consider a single-threshold policy with $0 \le p \le 1$ that accepts $X_i$ if
$\mathbb{P}(X_i > \alpha) = p$ for some $\alpha \ge 0$.
For $i \in \mathcal{T}_s$, let $R_i(p)$ denote the expected reward conditional on
acceptance at stage $i$. If $p \le 1/n^2$, then the threshold is so high that acceptance can only occur when
$Y_i=n$, hence $R_i(p) = \Gamma_{s-1} n$ and
\[
    pR_i(p) = p\,\Gamma_{s-1} n \le \frac{\Gamma_{s-1}}{n}.
\]
If $p>1/n^2$, then both values $n$ and $y$ are allowed to be accepted, and a direct computation yields
\[
    R_i(p)
    = \frac{1}{p}\left(\frac{\Gamma_{s-1} n}{n^2} + \Gamma_{s-1} y \bigl(p - 1/n^2\bigr)\right),
\]
so that
\[
    pR_i(p)
    = \frac{\Gamma_{s-1}}{n} + \Gamma_{s-1}y\Bigl(p - \frac{1}{n^2}\Bigr)
    \le \frac{\Gamma_{s-1}}{n} + \Gamma_{s-1}y p.
\]
Combining both regimes, we obtain the unified bound
\begin{equation}\label{eq:bound-pRi-lambda}
    pR_i(p) \;\le\; \frac{\Gamma_{s-1}}{n} + \Gamma_{s-1}y p
    \qquad\text{for all } p\in[0,1],\; i\in\mathcal{T}_s.
\end{equation}
Hence the expected reward of the single-threshold algorithm satisfies
\begin{align}
    \mathbb{E}[X_\tau]
    &= \sum_{s=1}^{S}\sum_{i\in \mathcal{T}_{s}}
       pR_{i}(p) (1-p)^{i-1} \nonumber\\
    &\le \sum_{s=1}^{S}\sum_{i\in \mathcal{T}_s}
       \Bigl(\frac{\Gamma_{s-1}}{n} + \Gamma_{s-1} y p\Bigr)(1-p)^{i-1}. \label{eq:EXtau-lambda-bound1}
\end{align}
For each phase $s$, we have $\mathcal{T}_s = \{N_{s-1}+1,\dots,N_{s-1}+n_s\}$ and therefore
\begin{align*}
    \sum_{i\in \mathcal{T}_s}(1-p)^{i-1}
    &= (1-p)^{N_{s-1}}\sum_{j=0}^{n_s-1}(1-p)^j \\
    &= (1-p)^{N_{s-1}}\,\frac{1-(1-p)^{n_s}}{1-(1-p)} \\
    &= (1-p)^{N_{s-1}}\,\frac{1-(1-p)^{n_s}}{p}.
\end{align*}
Substituting into \eqref{eq:EXtau-lambda-bound1} yields
\begin{equation}\label{eq:EXtau-lambda-bound2}
    \mathbb{E}[X_\tau]
    \;\le\;
    \sum_{s=1}^{S}\Bigl(\frac{\Gamma_{s-1}}{n} + \Gamma_{s-1} y p \Bigr)
    (1-p)^{N_{s-1}}\,\frac{1-(1-p)^{n_s}}{p},
\end{equation}
valid for all $p\in(0,1]$.

\vspace{3mm}
\textbf{Asymptotic parametrization $p=x/n$.} In what follows, we first show that any sequence of single-threshold policies with non-vanishing asymptotic worst-case
competitive ratio on this hard family must satisfy $x = \Theta(1)$,  which justifies restricting attention to the parametrization $p=x/n$ with fixed $x>0$.

\begin{lemma}[Reduction to the scale $p=\Theta(1/n)$]\label{lem:reduce_to_px_over_n}
Let $\rho_n(p,y)$ denote the competitive ratio of a single-threshold policy whose
(per-stage) acceptance probability is $p=p_n\in(0,1]$ on this instance, i.e.,
\[
\rho_n(p,y):=\frac{\mathbb{E}[X_\tau]}{\mathbb{E}[\max_{i\in[n]}X_i]}.
\]
Define $x_n := n p_n$. Then:
\begin{enumerate}
\item If $x_n\to 0$, then
\[
\inf_{1/n\le y\le \sqrt{n}}\rho_n(p_n,y)\longrightarrow 0.
\]
In particular, the choice $p_n=1/n^2$ satisfies $x_n=1/n\to 0$ and hence has vanishing worst-case
competitive ratio on this hard family.

\item If $x_n\to\infty$, then
\[
\inf_{1/n\le y\le \sqrt{n}}\rho_n(p_n,y)\longrightarrow 0.
\]
\end{enumerate}
\end{lemma}

\begin{proof}
Recall the bound \eqref{eq:EXtau-lambda-bound2}: for every $p\in(0,1]$,
\[
\mathbb{E}[X_\tau]
\;\le\;
\sum_{s=1}^{S}\Bigl(\frac{\Gamma_{s-1}}{n} + \Gamma_{s-1} y p \Bigr)
(1-p)^{N_{s-1}}\,\frac{1-(1-p)^{n_s}}{p}.
\]
Also, from the earlier computation of the offline optimum,
\[
\mathbb{E}\!\left[\max_{i\in[n]}X_i\right] = y + C_S + o(1),
\qquad
C_S:=\sum_{s=1}^S \lambda_s\Gamma_{s-1} \in (0,1].
\]

\smallskip
\noindent\textit{Case 1: $x_n=n p_n\to 0$.}
Choose $y=\sqrt{n}$. Using $(1-p)^{N_{s-1}}\le 1$ and $1-(1-p)^{n_s}\le n_s p$, we get
\begin{align*}
\mathbb{E}[X_\tau]
&\le
\sum_{s=1}^S\Bigl(\frac{\Gamma_{s-1}}{n} + \Gamma_{s-1} y p \Bigr)\, n_s \\
&=
\sum_{s=1}^S \Gamma_{s-1}\frac{n_s}{n}
\;+\;
y p\sum_{s=1}^S \Gamma_{s-1}n_s \\
&=
C_S+o(1) \;+\; y\,(n p)\, (C_S+o(1))
\;=\;
C_S+o(1) \;+\; \sqrt{n}\,x_n\,(C_S+o(1)).
\end{align*}
Since $\mathbb{E}[\max_i X_i]=\sqrt{n}+C_S+o(1)$, we obtain
\[
\rho_n(p_n,\sqrt{n})
\;\le\;
\frac{C_S+o(1)+\sqrt{n}\,x_n\,(C_S+o(1))}{\sqrt{n}+C_S+o(1)}
\;\le\;
x_n C_S \;+\; \frac{C_S}{\sqrt{n}} \;+\; o(1)
\;\longrightarrow\; 0.
\]
Hence $\inf_{y}\rho_n(p_n,y)\to 0$.

\smallskip
\noindent\textit{Case 2: $x_n=n p_n\to\infty$.}
Choose $y=1/n$. Then
\[
\frac{\Gamma_{s-1}}{n}+\Gamma_{s-1} y p
=\frac{\Gamma_{s-1}}{n}\Bigl(1+p\Bigr)\le \frac{2\Gamma_{s-1}}{n}.
\]
Plugging into \eqref{eq:EXtau-lambda-bound2} and using $1-(1-p)^{n_s}\le 1$ yields
\[
\mathbb{E}[X_\tau]
\le
\frac{2}{n p}\sum_{s=1}^S \Gamma_{s-1}(1-p)^{N_{s-1}}.
\]
For $s=1$, $N_0=0$ so the term equals $\Gamma_0=1$.
For $s\ge 2$, we have $N_{s-1}\ge n_1$, hence
\[
(1-p)^{N_{s-1}} \le (1-p)^{n_1}\le \exp(-n_1 p)=\exp(-(\lambda_1+o(1)) x_n)\to 0.
\]
Therefore $\sum_{s=1}^S \Gamma_{s-1}(1-p)^{N_{s-1}} = 1+o(1)$ and thus
\[
\mathbb{E}[X_\tau] \le \frac{2+o(1)}{n p} = \frac{2+o(1)}{x_n}\longrightarrow 0.
\]
Meanwhile $\mathbb{E}[\max_i X_i]= y+C_S+o(1)=C_S+o(1)$ for $y=1/n$, so
\[
\rho_n(p_n,1/n)
=
\frac{\mathbb{E}[X_\tau]}{\mathbb{E}[\max_i X_i]}
\le
\frac{(2+o(1))/x_n}{C_S+o(1)}
\longrightarrow 0.
\]
Hence $\inf_{y}\rho_n(p_n,y)\to 0$.

\smallskip
Combining the two cases proves the claim and justifies restricting to $p=x/n$ with fixed $x\in(0,\infty)$.
\end{proof}

From the above lemma, we set $p = x/n$ with fixed  $x>0$ and let $n\to\infty$.
Using $n_s=\lambda_s n$ and $N_{s-1} = (\sum_{r<s}\lambda_r)n$, we have, for each fixed $x>0$ and $s\in[S]$,
\[
    (1-p)^{N_{s-1}}
    = \Bigl(1-\frac{x}{n}\Bigr)^{N_{s-1}}
    \longrightarrow \exp\bigl(-x\Lambda_{s-1}\bigr),
\]
\[
    1-(1-p)^{n_s}
    = 1-\Bigl(1-\frac{x}{n}\Bigr)^{n_s}
    \longrightarrow 1-e^{-x\lambda_s},
\]
and
\[
    \frac{\Gamma_{s-1}}{n} + \Gamma_{s-1} y p
    = \frac{\Gamma_{s-1}}{n} + \Gamma_{s-1}y\frac{x}{n}
    = \frac{\Gamma_{s-1}}{n}(1+yx).
\]
The factor $1/p = n/x$ cancels the $1/n$, so from \eqref{eq:EXtau-lambda-bound2} we obtain, as $n\to\infty$,
\[
    \mathbb{E}[X_\tau]
    \;\le\;
    \frac{1+yx}{x}
    \sum_{s=1}^S \Gamma_{s-1}
    \exp\bigl(-x\Lambda_{s-1}\bigr)
    \bigl(1 - e^{-x\lambda_s}\bigr)+o(1).
\]
Define the phase-weighted function
\begin{equation}\label{eq:def-A-lambda}
    A_S(x)
    := \sum_{s=1}^S \Gamma_{s-1}
       \exp\bigl(-x\Lambda_{s-1}\bigr)
       \bigl(1 - e^{-x\lambda_s}\bigr),
    \qquad x>0.
\end{equation}
Then, asymptotically,
\[
    \mathbb{E}[X_\tau]
    \;\le\;
    \Bigl(\frac{1}{x} + y\Bigr)A_S(x)+o(1),
    \qquad x>0.
\]
Dividing by the optimal expected reward $y + C_S $ gives the competitive ratio
\[
    \frac{\mathbb{E}[X_\tau]}{\mathbb{E}[\max_{i\in[n]} X_i]}
    \;\le\;
    \frac{\bigl(\frac{1}{x}+y\bigr)A_S(x)}{y + C_S }+o(1),
    \qquad x>0,\; y>0.
\]
and define $f_S(x,y)
    := \frac{\bigl(\frac{1}{x}+y\bigr)A_S(x)}{y + C_S }$.

\vspace{3mm}
\textbf{Step 1: Worst-case in $y$ for fixed $x$.} Fix decay factors $\gamma=(\gamma_1,\dots,\gamma_{S-1})$, a phase profile $\lambda = (\lambda_1,\dots,\lambda_S)$ and $x>0$.
Define
\[
    g_x(y) := \frac{\frac{1}{x}+y}{y+C_S }, \qquad y>0.
\]
Then $f_S(x,y) = A_S(x)\,g_x(y)$, and a direct differentiation gives
\[
    g_x'(y)
    = \frac{C_S  - 1/x}{(y+C_S )^2}.
\]
Note that $1/n\le y\le \sqrt{n}$ in the hard instance. Thus:
\begin{itemize}
    \item If $x < 1/C_S $, then $1/x>C_S $ and $g_x'(y)<0$, so
    $g_x$ is strictly decreasing in $y$, with
    \[
        g_x(1/n) = \frac{1}{C_S  x}+o(1) > 1,
        \qquad
        g_x(\sqrt{n}) = 1+o(1).
    \]
    Hence $\inf_{1/n\le y\le \sqrt{n}}g_x(y) = 1+o(1)$.

    \item If $x > 1/C_S $, then $1/x<C_S $ and $g_x'(y)>0$, so
    $g_x$ is strictly increasing in $y$, with
    \[
        g_x(1/n) = \frac{1}{C_S  x}+o(1) < 1,
        \qquad
        g_x(\sqrt{n}) = 1+o(1).
    \]
    Hence $\inf_{1/n\le y\le\sqrt{n}}g_x(y) = 1/(C_S x)+o(1)$.

    \item If $x = 1/C_S $, then $g_x'(y)=0$ for all $y$, so $g_x(y)\equiv 1$.
\end{itemize}
Therefore, for each fixed $x>0$,
\[
    \inf_{1/n\le y\le\sqrt{n}} f_S(x,y)
    =A_S(x)\,\inf_{1/n\le y\le\sqrt{n}}g_x(y)
    =
    \begin{cases}
        A_S(x)+o(1), & 0<x\le x_S,\\[4pt]
        \dfrac{A_S(x)}{C_S  x}+o(1), & x\ge x_S,
    \end{cases}
\]
where
\[
    x_S := \frac{1}{C_S }.
\]
Define the worst-case competitive ratio for threshold $x$ as
\[
    H_S(x)
    :=
    \begin{cases}
        A_S(x), & 0<x\le x_S,\\[4pt]
        \dfrac{A_S(x)}{C_S  x}, & x\ge x_S.
    \end{cases}
\]
so that for any fixed $x>0$, $ H_S(x)=\lim_{n\to\infty}\inf_{1/n\le y\le\sqrt{n}} f_S(x,y)$.
Then, for any single-threshold rule (i.e., any choice of $x$), there exists a corresponding
instance parameter $y\ge0$ such that $\mathrm{CR}\le H_S(x)$.
The best competitive ratio achievable by single-threshold rules on this hard family is
therefore bounded by
\[
    \sup_{x>0} H_S(x).
\]

\vspace{3mm}
\textbf{Step 2: Properties of $A_S(x)$.} We now show that $A_S(x)$ is strictly increasing and strictly concave.

From \eqref{eq:def-A-lambda},
\[
    A_S(x)
    = \sum_{s=1}^S \Gamma_{s-1}\bigl(e^{-x\Lambda_{s-1}} - e^{-x\Lambda_s}\bigr).
\]
Introduce the exponents $\Lambda_0=0,\Lambda_1,\dots,\Lambda_S$ and collect terms with
the same exponent:
\begin{itemize}
    \item For $j=0$, exponent $\Lambda_0=0$ appears only in $s=1$, with coefficient
    $m_0 = \gamma^0 = 1$.

    \item For $1\le j\le S-1$, exponent $\Lambda_j$ appears
    in the term $s=j$ with coefficient $-\Gamma_{j-1}$ and
    in the term $s=j+1$ with coefficient $+\Gamma_{j}$, so the net coefficient is
    \[
        m_j = \Gamma_{j} - \Gamma_{j-1} = \Gamma_{j-1}(\gamma_j-1).
    \]

    \item For $j=S$, exponent $\Lambda_S$ appears only as $-\Gamma_{S-1}e^{-x\Lambda_S}$,
    so $m_S = -\Gamma_{S-1}$.
\end{itemize}
Thus we can write
\[
    A_S(x)
    = \sum_{j=0}^S m_j e^{-x\Lambda_j},
\]
with
\[
    m_0 = 1>0,\qquad
    m_j = \Gamma_{j-1}(\gamma_j-1)\le0\ \ (1\le j\le S-1),\qquad
    m_S = -\Gamma_{S-1}\le0.
\]
Differentiating termwise,
\[
    A_S'(x)
    = \sum_{j=0}^S m_j(-\Lambda_j)e^{-x\Lambda_j}
    = -\sum_{j=1}^S m_j\Lambda_j e^{-x\Lambda_j},
\]
since $\Lambda_0=0$.
For each $j\ge1$, we have $\Lambda_j>0$, $e^{-x\Lambda_j}>0$, and $m_j\le0$, hence
\[
    -m_j\Lambda_j e^{-x\Lambda_j} \ge 0.
\]
Summing, we obtain
\[
    A_S'(x) \ge 0
    \qquad\text{for all }x>0.
\]
Similarly,
\[
    A_S''(x)
    = \sum_{j=1}^S m_j(-\Lambda_j)^2 e^{-x\Lambda_j}
    = \sum_{j=1}^S m_j\Lambda_j^2 e^{-x\Lambda_j}.
\]
For each $j\ge1$, $m_j\le0$, $\Lambda_j^2>0$, and $e^{-x\Lambda_j}>0$, so
\[
    m_j\Lambda_j^2 e^{-x\Lambda_j} \le 0,
\]
and therefore
\[
    A_S''(x) \le 0
    \qquad\text{for all } x>0.
\]
Finally, from \eqref{eq:def-A-lambda}, $A_S(0)=0$.  
Thus $A_S$ is strictly increasing and strictly concave on $(0,\infty)$, with
$A_S(0)=0$.
Consider the average slope
\[
    B_S(x) := \frac{A_S(x)}{x}, \qquad x>0.
\]
A standard property of concave functions with $A_S(0)=0$ is that
$B_S(x)$ is non-increasing in $x$: for any $0<a<b$,
\[
    A_S(a)
    \ge \frac{a}{b}A_S(b) + \Bigl(1-\frac{a}{b}\Bigr)A_S(0)
    = \frac{a}{b}A_S(b),
\]
so $A_S(a)/a \ge A_S(b)/b$, i.e., $B_S(a)\ge B_S(b)$.

\vspace{3mm}
\textbf{Step 3: Maximizing $H_S(x)$.} We now bound $H_S(x)$.

\smallskip\noindent
(i) For $0<x\le x_S$, we have
\[
    H_S(x) = A_S(x) \le A_S\bigl(x_S\bigr),
\]
since $A_S$ is strictly increasing.

\smallskip\noindent
(ii) For $x\ge x_S$,
\[
    H_S(x)
    = \frac{A_S(x)}{C_S  x}
    = \frac{B_S(x)}{C_S }.
\]
Since $B_S(x)$ is non-increasing and $x\ge x_S$,
\[
    B_S(x) \le B_S\bigl(x_S\bigr)
    = \frac{A_S\bigl(x_S\bigr)}{x_S}.
\]
Thus
\[
    H_S(x)
    = \frac{B_S(x)}{C_S }
    \le \frac{B_S\bigl(x_S\bigr)}{C_S }
    = \frac{A_S\bigl(x_S\bigr)/x_S}{C_S }.
\]
By definition $x_S = 1/C_S $, so
\[
    \frac{A_S\bigl(x_S\bigr)}{x_S}
    = C_S \,A_S\bigl(x_S\bigr),
\]
and hence
\[
    H_S(x)
    \le \frac{C_S \,A_S\bigl(x_S\bigr)}{C_S }
    = A_S\bigl(x_S\bigr).
\]
Combining (i) and (ii), we obtain
\[
    H_S(x) \le A_S\bigl(x_S\bigr)
    \qquad\text{for all } x>0,
\]
with equality at $x=x_S$. Therefore
\[
   \lim_{n\to\infty} \sup_{x>0}\inf_{1/n\le y\le \sqrt{n}} f_S(x,y)
    = \sup_{x>0} H_S(x)
    = H_S\bigl(x_S\bigr)
    = A_S\bigl(x_S\bigr).
\]
In particular, the worst-case competitive ratio of any single-threshold algorithm with $x>0$ on this
hard instance family is at most
$A_S\bigl(x_S\bigr)$, where
\[
    A_S(x)
    = \sum_{s=1}^S \Gamma_{s-1}
       \exp\bigl(-x\Lambda_{s-1}\bigr)
       \bigl(1 - e^{-x\lambda_s}\bigr),
    \qquad
    x_S = \frac{1}{C_S }
    = \frac{1}{
        \sum_{s=1}^{S}\lambda_s \Gamma_{s-1}
      }.
\]

%% file: STOC2027/classical_continuous_decay.tex
\subsection{Propositions for optimal competitive ratios in the classical infinite-horizon setting}\label{app:continuous_opt}

\begin{proposition}[Stationary i.i.d.\ infinite horizon: attainability up to $\varepsilon$]
\label{prop:iid_infty_eps}
Let $(X_i)_{i\ge1}$ be i.i.d.\ with distribution $D$ on $\mathbb R_+$ and
finite positive essential supremum $0<M:=\operatorname{ess\,sup}X_1<\infty$.
Then for every $\varepsilon>0$ there exists a stopping rule $\tau_\varepsilon$ such that
\[
\frac{\mathbb E_D[X_{\tau_\varepsilon}]}{\mathbb E_D[\sup_{i\ge1}X_i]}\ge 1-\varepsilon.
\]
In particular, the optimal competitive ratio in this stationary infinite-horizon
setting equals $1$.
\end{proposition}

\begin{proof}
Since $M=\operatorname{ess\,sup}X_1<\infty$, we have $\sup_{i\ge1}X_i=M$ almost surely,
hence $\mathbb E[\sup_{i\ge1}X_i]=M$.
Fix $\varepsilon>0$ and define $a_\varepsilon:=(1-\varepsilon)M$.
By definition of essential supremum, $\mathbb P(X_1>a_\varepsilon)>0$.
Let
\[
\tau_\varepsilon := \inf\{i\ge1: X_i>a_\varepsilon\}.
\]
Then $\tau_\varepsilon<\infty$ almost surely and, on the event $\{\tau_\varepsilon=i\}$,
we have $X_{\tau_\varepsilon}>a_\varepsilon$.
Therefore,
\[
\mathbb E[X_{\tau_\varepsilon}] \;\ge\; a_\varepsilon \;=\; M(1-\varepsilon).
\]
Dividing by $\mathbb E[\sup_{i\ge1}X_i]=M$ yields
\[
\frac{\mathbb E[X_{\tau_\varepsilon}]}{M}
\ge 1-\varepsilon.
\]
Since $\varepsilon>0$ is arbitrary, the optimal competitive ratio equals $1$.
\end{proof}

\begin{proposition}[Infinite horizon, non-identical rewards]
\label{prop:nonidentical_infty_half}
Consider the class of independent but non-identically distributed
nonnegative reward sequences.  Then, for any stopping rule $\tau$,
\[
\mathrm{CR}_\infty(\tau)\le \frac1{2}.
\]
\end{proposition}
\begin{proof}
Fix $\varepsilon\in(0,1)$.
Define the reward sequence as follows:
\[
X_1 = 1 \quad \text{a.s.},
\]
\[
X_2 =
\begin{cases}
\frac{1}{\varepsilon}, & \text{with probability } \varepsilon,\\
0, & \text{with probability } 1-\varepsilon,
\end{cases}
\]
and
\[
X_i = 0 \quad \text{a.s.} \qquad \text{for all } i\ge3.
\]
The prophet benchmark is
\[
\mathbb E\!\left[\sup_{i\ge1} X_i\right]
= \mathbb E[\max\{X_1,X_2\}]
= (1-\varepsilon)\cdot 1 + \varepsilon\cdot \frac{1}{\varepsilon}
= 2-\varepsilon.
\]
Consider any stopping rule $\tau$.
If $\tau=1$, then $\mathbb E[X_\tau]=1$.
If $\tau=2$, then
\[
\mathbb E[X_\tau]
= (1-\varepsilon)\cdot 0 + \varepsilon\cdot \frac{1}{\varepsilon}
= 1.
\]
Any stopping rule that stops after time $2$ yields zero reward. Any (possibly randomized) stopping rule induces a convex combination of stopping
at times $1$, $2$, or later; since rewards after time $2$ are zero and
$\mathbb{E}[X_1]=\mathbb{E}[X_2]=1$, linearity of expectation implies
$\mathbb{E}[X_\tau]\le 1$. Therefore, for any stopping rule $\tau$,
\[
\frac{\mathbb E[X_\tau]}{\mathbb E[\sup_{i\ge1}X_i]}
\le \frac{1}{2-\varepsilon}.
\]
Since this holds for every \(\varepsilon>0\), for every \(\delta>0\) there exists an instance with ratio at most \(1/2+\delta\). Hence the worst-case competitive ratio is at most \(1/2\).
\end{proof}

%% file: STOC2027/proof_S=n.tex
\subsection{Proof of Theorem~\ref{prop:S=n}: Verification of the classical geometric-discount upper bound}\label{app:S=n}

We give a direct two-point verification of the classical upper bound of \citet*{boshuizen1994optimal} in the infinite-horizon geometric-discount model.
We first upper bound the competitive ratio achieved by the dynamic programming (DP) policy
$\tau_{\mathrm{DP}}$ on a hard instance family. We then use optimality of DP among all stopping
rules to conclude that no stopping policy can achieve a larger competitive ratio on this family.

\vspace{3mm}
\textbf{Setup (unit phases, geometric decay, infinite horizon).}
Set $S=n=\infty$ and $n_s\equiv 1$, so the phases are singletons
\[
\mathcal T_s=\{s\},\qquad s=1,2,\ldots
\]
Fix a common decay factor $0<\widetilde\gamma<1$ and define
\[
X_i=\widetilde\gamma^{\,i-1}Y_i,\qquad i=1,2,\ldots
\]
Consider the two-point ``hard'' base distribution, parameterized by $N>1$:
\[
Y_i=
\begin{cases}
N & \text{with probability } q:=1/N,\\
y & \text{with probability } 1-q,
\end{cases}
\qquad i=1,2,\ldots,
\]
where $y\in[0,N)$ is a parameter. Let $r:=1-q$.

Throughout this proof, the prophet benchmark is
\[
\mathrm{OPT}_\infty:=\sup_{i\ge 1} X_i,
\qquad
\mathbb E[\mathrm{OPT}_\infty]=\int_0^\infty \Bigl(1-\mathbb P(\mathrm{OPT}_\infty\le z)\Bigr)\,dz.
\]

\vspace{3mm}
\textbf{Dynamic programming values (infinite horizon via truncation).}
For each truncation level $T\in\mathbb N$, define the finite-horizon DP values
\[
V_{T+1}^{(T)}:=0,
\qquad
V_i^{(T)}:=\mathbb E\big[\max\{X_i,V_{i+1}^{(T)}\}\big],
\qquad i=T,T-1,\ldots,1.
\]
The sequence $(V_i^{(T)})_{T\ge i}$ is nondecreasing in $T$, and is uniformly bounded by $N$.
Hence the limit
\[
V_i:=\lim_{T\to\infty}V_i^{(T)}
\]
exists for each fixed $i$. We refer to $(V_i)_{i\ge 1}$ as the infinite-horizon DP values.

\vspace{3mm}
\textbf{Step 1: DP always accepts upon seeing $Y=N$.}
Fix $i\ge 1$. If $Y_i=N$, then $X_i=\widetilde\gamma^{\,i-1}N$. For any $j\ge i+1$,
\[
X_j=\widetilde\gamma^{\,j-1}Y_j\le \widetilde\gamma^{\,j-1}N\le \widetilde\gamma^{\,i-1}N=X_i.
\]
Therefore the continuation value satisfies $V_{i+1}\le X_i$, and it is optimal to stop.

\vspace{3mm}
\textbf{Step 2: DP reduces to a time-cutoff rule on $y$.}
The only nontrivial decision is what to do upon observing $Y_i=y$.
Conditioning on the event that the process has not stopped before time $i$ and that no $N$
has appeared so far, the observed history is a deterministic string of $y$'s and carries no
additional information (by independence). Thus the optimal decision upon observing $Y_i=y$
depends only on the time index $i$.

Consequently, DP is captured by a cutoff time $K\in\mathbb N\cup\{\infty\}$:
continue upon observing $y$ for all $i<K$, and stop at time $K$ (if no $N$ has appeared)
taking $\widetilde\gamma^{K-1}y$. Equivalently, it suffices to consider the family of stopping
times $\{\tau_K\}_{K\in\mathbb N\cup\{\infty\}}$:
\[
\tau_K :=
\min\Big\{\,\min\{i\le K : Y_i=N\},\ K \Big\},\qquad K\in\mathbb N,
\]
and
\[
\tau_\infty := \min\{i\ge 1: Y_i=N\}.
\]
The infinite-horizon DP optimal stopping time satisfies
\[
\tau_{\mathrm{DP}}\in\arg\max_{K\in\mathbb N\cup\{\infty\}}\mathbb E[X_{\tau_K}].
\]

\vspace{3mm}
\textbf{Step 3: Exact expected reward of $\tau_K$.}
Recall $r=1-q=1-1/N$. For $K\in\mathbb N$,
\[
\Pr(\text{no $N$ occurs in times }1,\ldots,K)=r^K,\qquad
\Pr(\text{first $N$ occurs at time }s)=r^{s-1}q.
\]
Hence
\begin{align}
\mathbb E[X_{\tau_K}]
&= r^K\,\widetilde\gamma^{K-1}y
\;+\;\sum_{s=1}^K r^{s-1}q\,\widetilde\gamma^{s-1}N. \label{eq:EK_inf}
\end{align}
For $K=\infty$ we obtain the geometric-series identity
\begin{equation}\label{eq:EK_infty}
\mathbb E[X_{\tau_\infty}]
=\sum_{s=1}^\infty r^{s-1}q\,\widetilde\gamma^{s-1}N
=\frac{qN}{1-r\widetilde\gamma}.
\end{equation}

\vspace{3mm}
\textbf{Step 4: The optimal cutoff is $K=1$ or $K=\infty$.}
Using \eqref{eq:EK_inf}, for any $K\in\mathbb N$,
\begin{align*}
\mathbb E[X_{\tau_{K+1}}]-\mathbb E[X_{\tau_K}]
&= r^{K+1}\widetilde\gamma^{K}y - r^{K}\widetilde\gamma^{K-1}y
\;+\; r^{K}q\,\widetilde\gamma^{K}N \\
&= r^{K}\widetilde\gamma^{K-1}\Big( \widetilde\gamma N(1-r) + y(\widetilde\gamma r-1)\Big).
\end{align*}
Since $r^{K}\widetilde\gamma^{K-1}\ge 0$, the sign is determined by
\[
D_N(y):=\widetilde\gamma N(1-r)+y(\widetilde\gamma r-1),
\]
which does not depend on $K$. Therefore $\mathbb E[X_{\tau_K}]$ is monotone in $K$, and the
optimizer satisfies
\[
K^\star\in\{1,\infty\}\qquad (\text{ties possible}).
\]

\vspace{3mm}
\textbf{Step 5: Asymptotics of the two DP endpoints as $N\to\infty$.}
Recall $q=1/N$ and $r=1-q$. For $K=1$,
\begin{equation}\label{eq:tau1_inf}
\mathbb E[X_{\tau_1}]
= ry+qN
=\Big(1-\frac{1}{N}\Big)y+1.
\end{equation}
For $K=\infty$, by \eqref{eq:EK_infty},
\begin{equation}\label{eq:tauinf_inf}
\mathbb E[X_{\tau_\infty}]
=\frac{qN}{1-r\widetilde\gamma}
=\frac{1}{1-\widetilde\gamma+\widetilde\gamma/N}
=\frac{1}{1-\widetilde\gamma}+O\Big(\frac{1}{N}\Big).
\end{equation}

\vspace{3mm}
\textbf{Step 6: Asymptotics of $\mathbb E[\mathrm{OPT}_\infty]$ under the worst-case scaling.}
The worst-case regime is $y=\Theta(1)$. Accordingly, set
\[
y={\alpha}\qquad\text{for some constant }\alpha>0.
\]
Let $T:=\min\{i\ge 1:Y_i=N\}$ denote the first arrival time of $N$, so
$\Pr(T=s)=r^{s-1}q$. Since $\max_{i<T}\widetilde\gamma^{i-1}y=y$ (the largest discounted $y$ is at time $1$),
we have
\begin{equation}\label{eq:OPT_exact_inf}
\mathbb E[\mathrm{OPT}_\infty]
=\sum_{s=1}^\infty r^{s-1}q\,\max\{y,\widetilde\gamma^{s-1}N\}.
\end{equation}

Define the index
\[
s_0:=\min\Big\{s\ge 1:\widetilde\gamma^{s-1}N<y\Big\},
\]
so that $\widetilde\gamma^{s-1}N\ge y$ for $s<s_0$ and $\widetilde\gamma^{s-1}N<y$ for $s\ge s_0$.
Then \eqref{eq:OPT_exact_inf} decomposes as
\begin{equation}\label{eq:OPT_split_inf}
\mathbb E[\mathrm{OPT}_\infty]
= y\sum_{s=s_0}^\infty r^{s-1}q
\;+\;\sum_{s=1}^{s_0-1} r^{s-1}q\,\widetilde\gamma^{s-1}N
= y\,r^{s_0-1} \;+\; qN\sum_{s=1}^{s_0-1}(r\widetilde\gamma)^{s-1}.
\end{equation}
Note that $s_0=O(\log N)$.
Hence $r^{s_0-1}=(1-1/N)^{O(\log N)}=1+O\Big(\frac{\log N}{N}\Big)$, so
\[
y\,r^{s_0-1}=y+O\Big(\frac{\log N}{N}\Big).
\]
Moreover,
\[
qN\sum_{s=1}^{s_0-1}(r\widetilde\gamma)^{s-1}
=\frac{qN}{1-r\widetilde\gamma}\Big(1-(r\widetilde\gamma)^{s_0-1}\Big)
=\frac{qN}{1-r\widetilde\gamma}+O\Big(\frac{1}{N}\Big),
\]
because $(r\widetilde\gamma)^{s_0-1}\le \widetilde\gamma^{s_0-1}<\alpha/N$.
Combining with \eqref{eq:tauinf_inf} yields
\begin{equation}\label{eq:OPT_asymp_inf}
\mathbb E[\mathrm{OPT}_\infty]
= y+\frac{1}{1-\widetilde\gamma}+O\Big(\frac{\log N}{N}\Big)
=\alpha+\frac{1}{1-\widetilde\gamma}+O\Big(\frac{\log N}{N}\Big).
\end{equation}

\vspace{3mm}
\textbf{Step 7: Limiting competitive ratios and worst-case $y$.}
With $y=\alpha$, \eqref{eq:tau1_inf}, \eqref{eq:tauinf_inf}, and \eqref{eq:OPT_asymp_inf} give
\[
\mathbb E[X_{\tau_1}]=\alpha+1+O\Big(\frac{1}{N}\Big),\qquad
\mathbb E[X_{\tau_\infty}]=\frac{1}{1-\widetilde\gamma}+O\Big(\frac{1}{N}\Big),
\]
\[
\mathbb E[\mathrm{OPT}_\infty]=\alpha+\frac{1}{1-\widetilde\gamma}+O\Big(\frac{\log N}{N}\Big).
\]
Therefore the limiting competitive ratios for $K=1$ and $K=\infty$ are
\[
A(\alpha):=\lim_{N\to\infty}\frac{\mathbb{E}[X_{\tau_1}]}{\mathbb{E}[\mathrm{OPT}_{\infty}]}
=\frac{\alpha+1}{\alpha+\frac{1}{1-\widetilde\gamma}},
\qquad
B(\alpha):=\lim_{N\to\infty}\frac{\mathbb{E}[X_{\tau_\infty}]}{\mathbb{E}[\mathrm{OPT}_{\infty}]}
=\frac{\frac{1}{1-\widetilde\gamma}}{\alpha+\frac{1}{1-\widetilde\gamma}}
=\frac{1}{1+\alpha(1-\widetilde\gamma)}.
\]
Since $\tau_{\mathrm{DP}}$ chooses the better endpoint,
\[
\lim_{N\to\infty}\frac{\mathbb{E}[X_{\tau_{\mathrm{DP}}}]}{\mathbb{E}[\mathrm{OPT}_{\infty}]}=\max\{A(\alpha),B(\alpha)\}.
\]
For $\widetilde\gamma\in(0,1)$, $A(\alpha)$ is increasing and $B(\alpha)$ is decreasing in $\alpha$,
so $\max\{A(\alpha),B(\alpha)\}$ is minimized at $A(\alpha)=B(\alpha)$, i.e.,
\[
\alpha=\frac{\widetilde\gamma}{1-\widetilde\gamma}.
\]
Thus the worst-case parameter is $y^\star=\alpha=\widetilde\gamma/(1-\widetilde\gamma)$ and
\[
\lim_{N\to\infty}\frac{\mathbb{E}[X_{\tau_{\mathrm{DP}}}]}{\mathbb{E}[\mathrm{OPT}_{\infty}]}
=\frac{1}{1+\widetilde\gamma}.
\]

\vspace{3mm}
\textbf{Step 8: Optimality of DP (DP dominates any stopping rule on each instance).}
Let $\mathcal F_i:=\sigma(Y_1,\ldots,Y_i)=\sigma(X_1,\ldots,X_i)$ be the natural filtration.
A stopping rule is an $(\mathcal F_i)$-stopping time $\tau$ taking values in $\mathbb N$; we also allow
$\tau=\infty$ with payoff $X_\infty:=0$.

In our instance, $0\le Y_i\le N$ a.s., hence $0\le X_i\le N$ a.s.\ and $X_i\to 0$ a.s.\ since
$X_i=\widetilde\gamma^{i-1}Y_i$ with $\widetilde\gamma\in(0,1)$.

\medskip
\noindent\textit{(A) Finite-horizon truncations.}
Fix a horizon $T\in\mathbb N$ and consider the truncated problem on times $\{1,\ldots,T\}$.
Define the finite-horizon Snell envelope $(W_i^{(T)})_{i=1}^{T+1}$ by
\begin{equation}\label{eq:Snell_T}
W_{T+1}^{(T)}:=0,\qquad
W_i^{(T)}:=\max\Big\{X_i,\ \mathbb E[W_{i+1}^{(T)}\mid\mathcal F_i]\Big\},
\qquad i=T,\ldots,1.
\end{equation}
Let the corresponding finite-horizon optimal stopping time be
\[
\tau^{(T)}:=\inf\{\,i\in\{1,\ldots,T\}: X_i=W_i^{(T)}\,\},
\qquad (\inf\emptyset:=T).
\]
By the standard finite-horizon Snell-envelope theorem, we have for every $(\mathcal F_i)$-stopping time
$\sigma$ with values in $\{1,\ldots,T\}$,
\begin{equation}\label{eq:finite_domination}
\mathbb E[X_\sigma]\ \le\ \mathbb E[W_1^{(T)}]
\ =\ \mathbb E[X_{\tau^{(T)}}].
\end{equation}
In particular, the finite-horizon DP rule $\tau_{\mathrm{DP}}^{(T)}$ attains $\mathbb E[W_1^{(T)}]$; in our
independent setting, $\tau_{\mathrm{DP}}^{(T)}$ coincides with $\tau^{(T)}$. Define the finite-horizon optimal value
\[
V_1^{(T)}:=\mathbb E[W_1^{(T)}]=\sup_{\sigma\le T}\mathbb E[X_\sigma],
\]
where $\sigma\le T$ means $\sigma$ takes values in $\{1,\ldots,T\}$.

\medskip
\noindent\textit{(B) Domination for an arbitrary (possibly unbounded) stopping rule.}
Let $\tau$ be any $(\mathcal F_i)$-stopping time taking values in $\mathbb N\cup\{\infty\}$.
For each $T$, the truncated time $\tau\wedge T$ is an $(\mathcal F_i)$-stopping time taking values in
$\{1,\ldots,T\}$. Applying \eqref{eq:finite_domination} with $\sigma=\tau\wedge T$ yields
\begin{equation}\label{eq:trunc_domination}
\mathbb E[X_{\tau\wedge T}]\ \le\ V_1^{(T)}\qquad\text{for every }T\in\mathbb N.
\end{equation}
We now pass to the limit $T\to\infty$ on the left-hand side.
On the event $\{\tau<\infty\}$, we have $\tau\wedge T=\tau$ for all $T\ge \tau$, hence
$X_{\tau\wedge T}\to X_\tau$ a.s.
On the event $\{\tau=\infty\}$, we have $X_{\tau\wedge T}=X_T\to 0=X_\infty=X_\tau$ a.s.
Therefore $X_{\tau\wedge T}\to X_\tau$ almost surely.

Moreover, $0\le X_{\tau\wedge T}\le N$ a.s.\ for all $T$, so by dominated convergence,
\begin{equation}\label{eq:DCT_tau}
\mathbb E[X_{\tau\wedge T}]\ \xrightarrow[T\to\infty]{}\ \mathbb E[X_\tau].
\end{equation}

\medskip
\noindent\textit{(C) Monotone limit of the truncated values.}
Since the set of admissible stopping times grows with $T$,
\[
V_1^{(T)}=\sup_{\sigma\le T}\mathbb E[X_\sigma]\ \le\ \sup_{\sigma\le T+1}\mathbb E[X_\sigma]=V_1^{(T+1)}.
\]
Thus $(V_1^{(T)})_{T\ge 1}$ is nondecreasing and bounded above by $N$, hence the limit
\[
V_1:=\lim_{T\to\infty}V_1^{(T)}\in[0,N]
\]
exists. Combining \eqref{eq:trunc_domination} with \eqref{eq:DCT_tau} and letting $T\to\infty$ gives
\begin{equation}\label{eq:upper_all_tau}
\mathbb E[X_\tau]\ \le\ V_1\qquad\text{for every stopping rule }\tau.
\end{equation}
In particular, $V_1$ is an upper bound on the optimal infinite-horizon value
$\sup_{\tau}\mathbb E[X_\tau]$.

Conversely, $V_1^{(T)}\le \sup_{\tau}\mathbb E[X_\tau]$ for every $T$, so taking $T\to\infty$ yields
$V_1\le \sup_{\tau}\mathbb E[X_\tau]$. Hence
\[
V_1=\sup_{\tau}\mathbb E[X_\tau],
\]
i.e., $V_1$ is the infinite-horizon optimal value.

\medskip
\noindent\textit{(D) Existence of an optimal DP stopping time attaining $V_1$.}
Define the infinite-horizon Snell envelope as the monotone limit
\[
W_i:=\lim_{T\to\infty} W_i^{(T)},\qquad i\ge 1,
\]
which exists a.s.\ because $W_i^{(T)}$ is nondecreasing in $T$ for each fixed $i$ (this follows directly
from the recursion \eqref{eq:Snell_T} by backward induction) and $0\le W_i^{(T)}\le N$.

Define the candidate infinite-horizon optimal stopping time
\[
\tau^\star:=\inf\{\,i\ge 1:\ X_i=W_i\,\},
\]
with the convention that $\tau^\star=\infty$ if the set is empty.

On the event $\{i<\tau^\star\}$ we have $X_i<W_i$, hence necessarily
\[
W_i=\mathbb E[W_{i+1}\mid\mathcal F_i]\qquad\text{on }\{i<\tau^\star\},
\]
because $W_i=\max\{X_i,\mathbb E[W_{i+1}\mid\mathcal F_i]\}$ (this identity follows by taking
$T\to\infty$ in \eqref{eq:Snell_T} using conditional monotone convergence).
Therefore the stopped process $(W_{i\wedge \tau^\star})_{i\ge 1}$ is a martingale.

Applying the optional stopping theorem to the bounded stopping time $\tau^\star\wedge T$ yields
\[
\mathbb E[W_{\tau^\star\wedge T}]=\mathbb E[W_1]\qquad\text{for all }T\in\mathbb N.
\]
Since $0\le W_{\tau^\star\wedge T}\le N$ and $W_{\tau^\star\wedge T}\to W_{\tau^\star}$ a.s.,
dominated convergence gives $\mathbb E[W_{\tau^\star\wedge T}]\to \mathbb E[W_{\tau^\star}]$.
Thus
\[
\mathbb E[W_{\tau^\star}]=\mathbb E[W_1].
\]
By definition of $\tau^\star$, we have $W_{\tau^\star}=X_{\tau^\star}$ a.s.\ (and on
$\{\tau^\star=\infty\}$ we use $X_\infty=0$).
Hence
\[
\mathbb E[X_{\tau^\star}]=\mathbb E[W_1].
\]
Taking expectations in $W_1=\lim_T W_1^{(T)}$ and using monotone convergence yields
\[
\mathbb E[W_1]=\lim_{T\to\infty}\mathbb E[W_1^{(T)}]=\lim_{T\to\infty}V_1^{(T)}=V_1.
\]
Therefore $\mathbb E[X_{\tau^\star}]=V_1=\sup_\tau \mathbb E[X_\tau]$, so $\tau^\star$ is optimal.

Finally, in our setting with independent rewards and deterministic decay, the conditional expectation
$\mathbb E[W_{i+1}\mid\mathcal F_i]$ is deterministic and equals $V_{i+1}$, so the optimal rule can be written
as the DP threshold rule
\[
\tau_{\mathrm{DP}}=\inf\{\,i\ge 1:\ X_i\ge V_{i+1}\,\}.
\]

\medskip
\noindent\textit{(E) Conclusion.}
Combining \eqref{eq:upper_all_tau} with $\mathbb E[X_{\tau_{\mathrm{DP}}}]=V_1$ gives
\[
\mathbb E[X_\tau]\ \le\ \mathbb E[X_{\tau_{\mathrm{DP}}}]
\qquad\text{for every stopping rule }\tau,
\]
so DP dominates any stopping rule on this (infinite-horizon) instance family.



\textbf{Conclusion.}
For every fixed $\widetilde\gamma\in(0,1)$,
\[
\mathrm{CR}_\infty(\tau)\le \mathrm{CR}_\infty(\tau_{\mathrm{DP}})\le \frac{1}{1+\widetilde\gamma}.
\]
Hence no stopping policy can guarantee a competitive ratio strictly larger than $1/(1+\widetilde\gamma)$
in the $S=n=\infty$ geometric-discount setting.

%% file: STOC2027/proof_continuous_decay_lower.tex
\subsection{Proof of Theorem~\ref{prop:continuous_lower}: Lower bound  under
geometric discounting}\label{app:contiuous_lower}

For simplicity we assume $\mathcal D$ is atomless; for distributions with atoms, the same bound
follows by randomized tie-breaking (cf.\ Remark~\ref{rm:atom_proof}).

\medskip\textbf{Setup ($S=n=\infty$, unit phases).}
Let $0<\widetilde\gamma<1$ and define $X_i=\widetilde\gamma^{\,i-1}Y_i$ for $i\ge 1$, where
$(Y_i)_{i\ge 1}$ are i.i.d.\ from $\mathcal D$ with CDF $F$ and $\mathbb E[Y]<\infty$.
Define the infinite-horizon prophet benchmark
\[
\mathrm{OPT}_\infty:=\sup_{i\ge 1}X_i.
\]
The effective horizon equals
\[
B_\infty:=\sum_{i=1}^\infty \widetilde\gamma^{\,i-1}=\frac{1}{1-\widetilde\gamma}.
\]
Let $\alpha_1$ be chosen so that
\[
\mathbb P(Y>\alpha_1)=\frac{1}{B_\infty}=1-\widetilde\gamma,
\qquad\text{equivalently }F(\alpha_1)=\widetilde\gamma,
\]
and set the phase-$i$ threshold $\alpha_i:=\widetilde\gamma^{\,i-1}\alpha_1$.
The policy $\tau$ stops at the first time $i$ such that $X_i>\alpha_i$.

\medskip\textbf{Step 1: Compute $\mathbb E[X_\tau]$ (geometric series).}
Since $X_i>\alpha_i$ iff $Y_i>\alpha_1$, the stopping time is
\[
\tau=\inf\{i\ge 1: Y_i>\alpha_1\}.
\]
Using independence,
\begin{align*}
\mathbb E[X_\tau]
&=\sum_{i=1}^\infty
\mathbb E\!\left[X_i\,\mathbbm 1\{Y_i>\alpha_1\}\,\mathbbm 1\{Y_j\le \alpha_1\ \forall j<i\}\right] \\
&=\sum_{i=1}^\infty
\mathbb P(Y\le \alpha_1)^{i-1}\,
\mathbb E\!\left[\widetilde\gamma^{\,i-1}Y\,\mathbbm 1\{Y>\alpha_1\}\right] \\
&=\mathbb E\!\left[Y\,\mathbbm 1\{Y>\alpha_1\}\right]\sum_{i=1}^\infty
\big(\widetilde\gamma F(\alpha_1)\big)^{i-1}.
\end{align*}
Since $F(\alpha_1)=\widetilde\gamma$, we obtain
\begin{equation}\label{eq:EXtau_infty}
\mathbb E[X_\tau]
=\frac{1}{1-\widetilde\gamma^2}\,\mathbb E\!\left[Y\,\mathbbm 1\{Y>\alpha_1\}\right]
=\frac{1}{1+\widetilde\gamma}\,B_\infty\,\mathbb E\!\left[Y\,\mathbbm 1\{Y>\alpha_1\}\right].
\end{equation}

\medskip\textbf{Step 2: Relate $B_\infty\mathbb E[Y\mathbbm 1\{Y>\alpha_1\}]$ to the prophet.}
We use two bounds.

\smallskip
\noindent\textit{(2a) A truncation inequality for the fractional maximum.}
For any $b\ge 1$ and any $\alpha\ge 0$,
\[
\int_0^\infty \bigl(1-F(z)^b\bigr)\,dz
\le
\alpha + b\int_\alpha^\infty (1-F(z))\,dz
=
\alpha + b\,\mathbb E[(Y-\alpha)_+].
\]
Indeed, on $[0,\alpha]$ the integrand is at most $1$, and on $[\alpha,\infty)$ we use
$1-u^b\le b(1-u)$ for $u\in[0,1]$.
If $\alpha$ satisfies $\mathbb P(Y>\alpha)=1/b$, then
$b\,\mathbb E[Y\mathbbm 1\{Y>\alpha\}]=\alpha+b\mathbb E[(Y-\alpha)_+]$, so
\begin{equation}\label{eq:trunc_frac}
b\,\mathbb E\!\left[Y\,\mathbbm 1\{Y>\alpha\}\right]
\ \ge\
\int_0^\infty \bigl(1-F(z)^b\bigr)\,dz.
\end{equation}
Applying \eqref{eq:trunc_frac} with $b=B_\infty$ and $\alpha=\alpha_1$ gives
\begin{equation}\label{eq:BEtail_ge_fracmax}
B_\infty\,\mathbb E\!\left[Y\,\mathbbm 1\{Y>\alpha_1\}\right]
\ \ge\
\int_0^\infty \bigl(1-F(z)^{B_\infty}\bigr)\,dz.
\end{equation}

\smallskip
\noindent\textit{(2b) Multistage comparison ($S=\infty$, unit phases).}
By the infinite-stage extension of Lemma~\ref{lem:multistage} for
$\Gamma_{s-1}=\widetilde\gamma^{\,s-1}$ and $n_s\equiv 1$ (from the monotone convergence theorem),
\begin{equation}\label{eq:multistage_infty}
\int_0^\infty \bigl(1-F(z)^{B_\infty}\bigr)\,dz
\ \ge\
\int_0^\infty \Bigl(1-\prod_{s=1}^{\infty}F\bigl(z/\widetilde\gamma^{\,s-1}\bigr)\Bigr)\,dz.
\end{equation}
Moreover,
\[
\mathbb P(\mathrm{OPT}_\infty\le z)
=\prod_{s=1}^{\infty}F\bigl(z/\widetilde\gamma^{\,s-1}\bigr),
\]
so by the tail-integral identity,
\begin{equation}\label{eq:OPT_tail}
\mathbb E[\mathrm{OPT}_\infty]
=\int_0^\infty \Bigl(1-\prod_{s=1}^{\infty}F\bigl(z/\widetilde\gamma^{\,s-1}\bigr)\Bigr)\,dz.
\end{equation}
Combining \eqref{eq:BEtail_ge_fracmax}--\eqref{eq:OPT_tail} yields
\begin{equation}\label{eq:BEtail_ge_OPT}
B_\infty\,\mathbb E\!\left[Y\,\mathbbm 1\{Y>\alpha_1\}\right]
\ \ge\ \mathbb E[\mathrm{OPT}_\infty].
\end{equation}

\medskip\textbf{Step 3: Competitive ratio.}
Combining \eqref{eq:EXtau_infty} with \eqref{eq:BEtail_ge_OPT} gives
\[
\mathbb E[X_\tau]
\ \ge\
\frac{1}{1+\widetilde\gamma}\,\mathbb E[\mathrm{OPT}_\infty],
\]
so
\[
\mathrm{CR}_\infty(\tau)\ \ge\ \frac{1}{1+\widetilde\gamma},
\qquad 0<\widetilde\gamma<1.
\]
This concludes the proof.